\pdfoutput=1 
\documentclass[11pt]{article}
\usepackage{fullpage}
\usepackage[pagebackref, colorlinks = true, linkcolor = blue, urlcolor  = blue, citecolor = red, bookmarks, hypertexnames=false]{hyperref}
\usepackage{amssymb}
\usepackage{amsmath}
\usepackage{amsthm}
\usepackage{graphicx,color,colordvi}
\usepackage{bbm}
\usepackage{varioref}
\usepackage[capitalise]{cleveref}
\usepackage{stmaryrd}
\usepackage[utf8]{inputenc}
\usepackage[blocks]{authblk}
\usepackage{dsfont}
\usepackage{mathtools}
\usepackage{authblk}
\usepackage{wrapfig}
\usepackage[english]{babel}
\usepackage{physics}
\usepackage{braket}
\usepackage[table]{xcolor}
\usepackage[center]{caption}
\usepackage{tikz}
\usepackage{ifthen}
\usepackage{pgfplots}
\pgfplotsset{compat=1.9}
\usetikzlibrary{shapes,arrows.meta}
\usetikzlibrary{positioning}
\usetikzlibrary{shapes.geometric}
\RequirePackage[framemethod=default]{mdframed}
\usepackage[margin=1in]{geometry}
\usepackage{comment}
\usepackage{url}
\usepackage{ upgreek }
\usepackage{makecell}
\usepackage{fancyref} 
\usepackage{float}
\usepackage{enumitem}
\usepackage{autonum}
\usepackage{subcaption}
\newfloat{algorithm}{t}{lop}
\usepackage{array,rotating ,makecell, multirow, tabularx}
\usepackage{scrextend}
\usepackage{mathrsfs}
\usepackage{enumitem}
\usepackage{soul}

\usepackage[normalem]{ulem}

\DeclareMathOperator{\diam}{diam}
\DeclareMathOperator{\height}{height}
\DeclareMathOperator{\supp}{supp}
\DeclareMathOperator{\poly}{poly}
\DeclareMathOperator{\Cov}{Cov}
\DeclareMathOperator{\Var}{Var}
\DeclareMathOperator{\dist}{dist}  
\DeclareMathOperator{\EP}{EP}
\DeclareMathOperator{\id}{id}
\newcommand{\inn}{\text{in}}
\newcommand{\out}{\text{out}}
\newcommand{\1}{\ensuremath{\mathbbm{1}}}

\newcommand{\Li}{\mathcal{L}}
\newcommand{\KMS}{\text{KMS}}
\newcommand{\GNS}{\text{GNS}}
\newcommand{\ceil}[1]{\lceil {#1} \rceil}
\newcommand{\floor}[1]{\lfloor {#1} \rfloor}
\newcommand{\BH}{\mathcal{B}(\mathcal{H})}
\usepackage{cancel}

\newtheoremstyle{newdefinition}{}{}{\normalfont}{}{\bfseries}{}{\newline}
{\thmname{#1} \thmnumber{#2}\thmnote{ (#3)}}

\newtheoremstyle{newplain}{}{}{\itshape}{}{\bfseries}{}{1em}
{\thmname{#1} \thmnumber{#2}\thmnote{ (#3)}}

\newtheoremstyle{newremark}{}{}{\normalfont}{}{\bfseries}{}{1em}
{\thmname{#1}}

\theoremstyle{newdefinition}
\newtheorem{definition}{Definition}[section]

\theoremstyle{newplain}
\newtheorem{theorem}[definition]{Theorem}
\newtheorem{lemma}[definition]{Lemma}
\newtheorem{proposition}[definition]{Proposition}
\newtheorem{corollary}[definition]{Corollary}
\newtheorem{remark}[definition]{Remark}

\newtheoremstyle{myplain}{5pt}{5pt}{\itshape}{0pt}{\bfseries}{}{5pt plus 1pt minus 1pt}{}
\theoremstyle{myplain}
\newtheorem*{theorem*}{Theorem}
\newtheorem*{corollary*}{Corollary}

\DeclareMathOperator{\R}{\mathbb{R}} 

\DeclareMathOperator{\Z}{\mathbb{Z}}

\DeclareMathOperator{\HH}{\mathcal{H}}

\definecolor{Turquise}{HTML}{14C7DE}
\definecolor{SkyBlue}{HTML}{3498DB}
\definecolor{RoyalBlue}{HTML}{444FAD}
\definecolor{RoyalYellow}{HTML}{F1C40F}
\definecolor{purplemod}{HTML}{AF82F3}
\definecolor{darkgreen}{RGB}{106, 134, 104}
\definecolor{darkgreen2}{RGB}{90, 168, 143}
\definecolor{bluedots}{RGB}{15, 5, 168}
\definecolor{purpledots}{RGB}{168, 39, 110}

\begin{document}

\title{Rapid thermalization of dissipative many-body dynamics of commuting Hamiltonians}
\author[1,2]{Jan Kochanowski\thanks{jan.kochanowski@inria.fr}}
\author[3]{\'Alvaro M. Alhambra\thanks{alvaro.alhambra@csic.es}}
\author[4,5]{{\'A}ngela Capel\thanks{ac2722@cam.ac.uk}}
\author[2,6]{Cambyse Rouzé\thanks{rouzecambyse@gmail.com}}
\affil[1]{\small Max-Planck-Institute of Quantum Optics, Hans-Kopfermann-Strasse 1, 85748 Garching, Germany}
\affil[2]{Inria, Télécom Paris - LTCI, Institut Polytechnique de Paris, 91120 Palaiseau, France}
\affil[3]{Instituto de F\'isica Te\'orica UAM/CSIC, C/ Nicol\'as Cabrera 13-15, Cantoblanco, 28049 Madrid, Spain}
\affil[4]{Fachbereich Mathematik, Eberhard Karls Universit\"at T\"ubingen, Germany}
\affil[5]{Department of Applied Mathematics and Theoretical Physics, University of Cambridge, United Kingdom}
\affil[6]{Department of Mathematics, Technische Universit\"at München, Germany}

\maketitle
\vspace{-1cm}
\begin{abstract}

Quantum systems typically reach thermal equilibrium rather quickly when coupled to a thermal environment. The usual way of bounding the speed of this process is by estimating the spectral gap of the dissipative generator. However the gap, by itself, does not always yield a reasonable estimate for the thermalization time in many-body systems: without further structure, a uniform lower bound on it only constrains the thermalization time to grow polynomially with system size.

Here, instead, we show that for a large class of geometrically-2-local models of Davies generators with commuting Hamiltonians, the thermalization time is much shorter than one would naïvely estimate from the gap: at most logarithmic in the system size. This yields the so-called rapid mixing of dissipative dynamics. The result is particularly relevant for 1D systems, for which we prove rapid thermalization with a system size independent decay rate only from a positive gap in the generator. We also prove that systems in hypercubic lattices of any dimension, and exponential graphs, such as trees, have rapid mixing at high enough temperatures. We do this by introducing a novel notion of clustering which we call ``strong local indistinguishability'' based on a max-relative entropy, and then proving that it implies a lower bound on the modified logarithmic Sobolev inequality (MLSI) for nearest neighbour commuting models.  

This has consequences for the rate of thermalization towards Gibbs states, and also for their relevant Wasserstein distances and transportation cost inequalities.
Along the way, we show that several measures of decay of correlations on Gibbs states of commuting Hamiltonians are equivalent, a result of independent interest. At the technical level, we also show a direct relation between properties of Davies and Schmidt dynamics, that allows to transfer results of thermalization between both.

\end{abstract}

\newpage 
\setcounter{tocdepth}{2}        
\renewcommand{\contentsname}{Table of Contents}
\tableofcontents

\newpage

\section{Introduction} 

Physical systems in nature are most often coupled to an external environment, with which they eventually equilibrate. For quantum ones, that coupling implies that their dynamics are described by Quantum Markov Semigroups (QMS) of the form $\{e^{t\mathcal{L}}\}_{t\geq 0}$, which are generated by a Lindbladian super-operator $  \mathcal{L}(\cdot)$.

This so-called \emph{dissipative} evolution monotonically converges to a unique fixed point under a weak set of conditions \cite{art:Spohn1976,art:Frigerio1977}, which, roughly speaking, are satisfied as long as the evolution induced by the external coupling is sufficiently ergodic. Additionally, when that external coupling is very weak, and to an environment with a fixed temperature $\beta$, that unique fixed point is the Gibbs state
\begin{equation}
    \sigma^\Lambda=\frac{e^{-\beta H_\Lambda}}{\Tr[e^{-\beta H_\Lambda}]},
\end{equation}
where $H_\Lambda$ is the Hamiltonian of the system. The QMS describing those thermalization processes are then referred to as \emph{Davies maps} \cite{art:Davies1,DaviesBook}.

The Davies evolution is a Markovian approximation of the reduced state dynamics of a many-body spin system weakly-coupled to an infinite-dimensional environment in thermal equilibrium. This type of open system dynamics described by a master equations, which always has a QMS as a solution, is of high interest in the fields of quantum optics, condensed matter, chemical physics, statistical physics, quantum information, and mathematical physics. The interest in Markovian descriptions of open system dynamics has been further motivated by developments in quantum information theory and the study of decoherence. Davies evolutions, originally studied in \cite{art:Davies1archbold_1983}, frequently feature in the literature concerning thermalization of quantum systems, both from the physical and computational perspectives \cite{art:QuantumGibbsSamplers-kastoryano2016quantum, art:QuantumThermalchifang2023quantum,art:2localPaper, art:Bluhm2022exponentialdecayof}. 

One of the more important aspects to understand about these processes is: if the Gibbs state is always reached, independently of the initial conditions, how quickly does that happen? The speed of convergence to equilibrium or \emph{thermalization} can be expressed through the notion of the \emph{mixing time}.  Write $\rho_t:=e^{t\mathcal{L}_*}(\rho)$, and let $\mathcal{D}(\mathcal{H}) $ be the set of normalised density operators. Then for $\epsilon>0$, it is defined as
\begin{align}
    t_\text{mix}(\epsilon) := \inf\{t\geq 0 \, |  \, \forall \rho\in\mathcal{D}(\mathcal{H}) \;  \|\rho_t-\sigma \|_1\leq\epsilon\}.
\end{align}

The most frequent way of estimating this mixing time, both in quantum and classical scenarios, is through the spectral gap of the generator $\lambda(\mathcal{L})$. 
This can be expressed variationally through the \textit{Poincaré inequality} as 
\begin{align}
    \lambda\Var^{\KMS}_\sigma(X_t) \leq -\frac{d}{dt}\bigg|_{t=0}\Var^{\KMS}_{\sigma}(X) = -\langle X,\mathcal{L}(X)\rangle^\KMS_\sigma,
\end{align} where $ \langle X,Y\rangle_\sigma^{\KMS} := \Tr[\sqrt{\sigma}X^*\sqrt{\sigma}Y]$ and  $\Var^{\KMS}_\sigma(X):=  \langle X-\Tr[\sigma X]\1,X-\Tr[\sigma X]\1\rangle_\sigma^{\KMS}$ (see \Cref{subsec:weightednorms}).  The \textit{spectral gap} is the largest constant $\lambda$ which satisfies this inequality for all $X\in\mathcal{B}(\mathcal{H})$ \cite{art:bardet2017estimating,art:QLSIandrapidmixingKastoryano_2013,art:QuantumGibbsSamplers-kastoryano2016quantum}, i.e.
\begin{align}
    \lambda(\mathcal{L}) := \inf_{X\in\mathcal{B}(\mathcal{H})}\frac{-\langle X,\mathcal{L}(X)\rangle^\KMS_\sigma}{\Var^{\KMS}_\sigma(X)}.
\end{align}
This directly implies exponential decay of the variance, i.e.
$
    \Var^{\KMS}_\sigma(X_t)\leq e^{-\lambda(\mathcal{L})t}\Var^{\KMS}_\sigma(X)
$,
from which it follows that $ \|\rho_t-\sigma\|_1\leq \|\sigma^{-1}\|^{-\frac{1}{2}}e^{-\lambda(\Li) t}$, so that
\begin{equation}\label{def:VarianceMixingtime}
 t_\text{mix}(\epsilon)\le \frac{1}{\lambda(\mathcal{L})}\log(\epsilon^{-1}\|\sigma^{-\frac{1}{2}}\|)\,.
\end{equation} 
While this inequality is often a good approximation in small systems, it can be an enormous overestimation of the mixing time in many-particle settings. In that case, $\|\sigma^{-1}\|^{-1}=e^{\mathcal{O}(|\Lambda|)}$ with $|\Lambda|$ the system size and the upper bound of Eq. \eqref{def:VarianceMixingtime} scales as $\text{poly}(|\Lambda|)$. However, when the interactions among the particles have an underlying local structure, we expect that very often the mixing time in the worst case will be of the form $t_\text{mix}(\epsilon) \le \frac{1}{\gamma(\mathcal{L})}(\log(\epsilon^{-1})+\log |\Lambda|$), for $\gamma(\mathcal{L})$ possibly another constant depending on the Lindbladian. 

Heuristically, the reason for this is that the local structure of the interactions, both among the many particles and with the environment, may cause the effective dissipation to be local. In that case, the Lindbladian can be written as a sum of local jump operators, such that we can think of the thermalization of the whole system as a sum of roughly independent processes localized among regions of $\mathcal{O}(1)$ many particles. Since there are polynomially many such regions, the total convergence error $\epsilon$ should not be more than the sum of that of the individual regions. This ``divide and conquer" line of thought then suggests a convergence error $\epsilon \sim |\Lambda| e^{-\Omega (t)}$  \footnote{In fact, this is the scaling that one can trivially find when there are no interactions between all the particles.}. When this scaling holds, the mixing time grows at most as $\log |\Lambda|$ and we say the system displays \emph{rapid mixing}.

Rapid mixing is a defining feature of dissipative many-body dynamics, and comes along with a number of important consequences. The fact that an evolution has rapid mixing can be associated to properties of the correlations of the fixed point: for systems to reach a steady state quickly, it must be the case that the fixed points do not have features akin to long-range order. As such, the study of rapid mixing, both in the classical and quantum case, is very closely linked to the study of the correlation properties of their (thermal) fixed points. 

Along these lines, we know that dissipative evolutions with the rapid mixing property are stable under perturbations \cite{art:StabilityCubitt_2015}, and their fixed points have decay of correlations \cite{art:RapidMixingandDecayofCorrelationsKastoryano_2013}, display concentration properties \cite{art:De_PalmaRouzeConcentrationInequalities}, and equivalence of ensembles, among various other features associated with standard statistical ensembles. Additionally, rapid mixing signals the absence of dissipative phase transitions \cite{art:Diehl_2008,art:DissTheory} and rules out the usefulness of models as self-correcting quantum memories \cite{Rev:BrownSelf}. It is thus of great interest to understand when such property holds. 

While rapid mixing may be an intuitive feature of thermalizing dynamics, proving it is in general highly non-trivial. Nevertheless, progress has been made in recent years \cite{art:QuantumGibbsSamplers-kastoryano2016quantum,art:2localPaper,art:ApproxTensorizationBardet_2021,art:ImplicationsAndRapidTermalization-Cambyse,BardetCapelLuciaPerezGarciaRouze-HeatBath1DMLSI-2019,art:CompleteEntropicInequalities_GaoRouze_2022,art:chen2023fast}, mostly in the context of commuting interactions, through the concept of the MLSI (\textit{modified logarithmic Sobolev inequality}, see Section \ref{sec:cMLSI}) constant, $\alpha\equiv \alpha(\mathcal{L})>0$. This quantity directly yields the estimate
\begin{equation}\label{equ:rapidmixing}
t_\text{mix}(\epsilon) \leq \frac{1}{\alpha}\mathcal{O}\left(\log\frac{1}{\epsilon}+\log|\Lambda|\right),
\end{equation}
so that rapid mixing can be proven via lower bounds on $\alpha$.

For Davies evolutions of 1-dimensional systems with uniform geometrically-local, commuting, and translation invariant Hamiltonians, it was shown in \cite{art:ImplicationsAndRapidTermalization-Cambyse} that there exists a strictly positive MLSI constant $\alpha(\mathcal{L})=\Omega((\log|\Lambda|)^{-1})$ at any temperature. While this guarantees rapid mixing with a polylogarithmic scaling, it does not yet reach the optimal $\mathcal{O}(1)$ constant rate of exponential decay with time that is expected on physical grounds. So far, this optimal scaling was only known for on-site depolarizing noise \cite{art:QuantumConditionalEntropyCapel_2018,BeigiDattaRouze-ReverseHypercontractivity-2018,MullerHermesFrancaWolf-EntropyProduction-2016,MullerHermesFrancaWolf-DepolarizingChannels-2016} and, in a more general context, for the Schmidt generators (which are a less physically motivated thermalization process, see  \Cref{subsec:SchmidtCondExp}) of a system in hypercubic latices in dimensions $D\geq 2$ above a threshold temperature, and uniform nearest neighbour commuting Hamiltonian \cite{art:2localPaper}. 

\subsection{Summary of results}

In this paper, we prove \textit{rapid mixing} for the Davies dynamics of a large class of lattice models of commuting Hamiltonians. We show that in 1D the property of rapid mixing follows directly from the presence of a gap in the Davies generator (which can be proven from first principles), and that it can also be proven at high temperatures with a great degree of generality in higher degree lattices. To do this, we derive relations between different measures of correlations at the fixed point, including a novel notion that we term ``strong local indistinguishability", which is based on the max-relative entropy between the marginals of the fixed point and we find to be directly linked to the strategies for proving rapid mixing.    

Our first main result is for 1D systems, where we achieve the optimal $\Omega(1)$ scaling for the MLSI constant.
\begin{theorem}[Optimal rapid thermalization in 1D, informal]
In 1D, for Davies generators of commuting, local Hamiltonians, having a positive gap is equivalent to the existence of a system-size independent positive MLSI constant $\alpha(\mathcal{L}^D_\Gamma)=\Omega(1)_{|\Gamma|\to \infty}$. This yields optimal rapid mixing at all positive temperatures $\beta^{-1}$ for these models. 
\end{theorem}
This is a strict strengthening of the previous 1D results \cite{art:EntropyDecayOf1DSpinChain-Cambyse,art:ImplicationsAndRapidTermalization-Cambyse,art:QuantumGibbsSamplers-kastoryano2016quantum}  with an additional extension to the non translation-invariant setting due to \cite{kimura2024clustering}. The formal version of this result can be found in Theorem \ref{thm:mainGap} and Corollary \ref{thm:main1D}. Under the same assumption on the gap, we also get an, over the simple gap assumption, square-root improved mixing time for $2D$ lattices assuming the decay of correlations is strong enough.
\begin{theorem}[Sub-linear thermalization in 2D, informal]
In 2D, for Davies generators of commuting, nearest-neighbour Hamiltonians, having a positive gap and a sufficiently small correlation length is equivalent to the existence of a strictly positive square root decreasing MLSI constant $\alpha(\mathcal{L}^D_\Gamma)=\Omega(|\Gamma|^{-\frac{1}{2}})_{|\Gamma|\to \infty}$. This implies a mixing time that scales at worst with the square root of the system, up to a logarithmic correction. 
\end{theorem}

The formal version of this result can also be found in Theorem \ref{thm:mainGap}.
For higher dimensional lattices in the high temperature regime we also give strict improvement in the following.
\begin{theorem}[Rapid thermalization at high temperature, informal] Nearest neighbour, commuting potentials at sufficiently high temperature satisfy a MLSI with 
\begin{itemize}
    \item[1)] system-size independent constant $\alpha(\mathcal{L}^D_\Gamma)=\Omega(1)_{|\Gamma|_\to\infty}$ when on a sub-exponential graph, e.g. $\mathbb{Z}^D$ for any $D\in\mathbb{N}$, or
    \item[2)] log-decreasing constant $\alpha(\mathcal{L}^D_\Gamma)=\Omega((\log|\Gamma|)^{-1})_{|\Gamma|_\to\infty}$ when on an exponential graph, e.g. a $b$-ary tree $\mathbb{T}_b$ for $b>1$.
\end{itemize} 
In both of these cases the Davies dynamics displays rapid mixing.
\end{theorem}
For trees the bound on mixing times is novel within the quantum setting, while for hypercubic lattices it generalises the result of \cite{art:2localPaper} from the less physically motivated Schmidt dynamics to the Davies generators. The formal version of this result can be found in Theorem \ref{thm:mainTemp}. An overview of these mixing time results can also be found in \ref{tab:mainresultsinformal}.

\begin{table}[h]
    \centering
    \begin{tabular}{|c|c|c|c|}
    \hline Assumptions  & Lattice & Mixing results  & Previous results \\ \hline \hline
   any positive temperature & $\mathbb{Z}$ & \cellcolor{green} $\begin{array}{c}
         \text{optimal,}  \\
          t_{mix}=\mathcal{O}(\log|\Gamma|)
    \end{array} $      & \cellcolor{lime} \cite{art:EntropyDecayOf1DSpinChain-Cambyse}: $t_{mix}=\mathcal{O}(\log|\Gamma|^2)$ \\ \hline 
    gap, small $\xi$ & $\mathbb{Z}^2$ & \cellcolor{yellow}  $\begin{array}{c}
         \text{sub-linear,}  \\
          t_{mix}=\mathcal{O}(\sqrt{|\Gamma|}\log|\Gamma|)
    \end{array} $    & \cellcolor{orange} \eqref{def:VarianceMixingtime}: linear, $t_{mix}=\mathcal{O}(|\Gamma|)$  \\ \hline 
    high temperature & $\mathbb{Z}^D$,\footnote{and any 2-colorable subexponential graphs} & \cellcolor{green} $\begin{array}{c}
         \text{optimal}  \\
          t_{mix}=\mathcal{O}(\log|\Gamma|)
    \end{array} $ & \cellcolor{orange} \cite{art:QuantumGibbsSamplers-kastoryano2016quantum} linear, $t_{mix}=\mathcal{O}(|\Gamma|)$  \\
    \hline 
    high temperature, small $\xi^\prime$ & $\mathbb{T}_b$,\footnote{and general exponential graphs with finite growth constant} & \cellcolor{lime} $ \begin{array}{c}
         \text{rapid mixing,}  \\
          t_{mix}=\mathcal{O}(\textup{poly}\log|\Gamma|)
    \end{array} $  & \cellcolor{red}  \\  \hline 
\end{tabular}
    \caption{Summary of the main results of this paper in terms of mixing times of Davies generators associated to systems with nearest neighbour commuting interactions ordered by system lattice and assumptions required. In the $1D$ setting the interactions need not be nearest-neighbour. $\xi$ denotes the ($\mathbb{L}_\infty$) correlations decay length, , see def \ref{def:LinfinityClustering}, and $\xi^\prime$ the q$\mathbb{L}_1\to\mathbb{L}_\infty$-decay length, see def \ref{def:qL1toLinfty}. For formal statements of these results see Section \ref{sec:main}. The \sethlcolor{green} \hl{green} denotes optimal scaling; the \sethlcolor{lime}\hl{lime} not optimal, but still rapid thermalization; the \sethlcolor{yellow}\hl{yellow} denotes thermalization that is better than just the gap assumption, but not rapid; the\sethlcolor{orange} \hl{orange}, linear scaling achieved directly from gap; and the \sethlcolor{red}\hl{red} no prior results. } 
    \label{tab:mainresultsinformal}
\end{table}

A key ingredient in obtaining these results is by showing equivalence between several different notions of clustering, meaning that for fixed-size regions, exponential decay of one measure is equivalent to exponential decay of others. More concretely, we have the following implications.
\begin{theorem}[Equivalence of clustering notions (const. size), informal]\label{thm:equiv_decay_corr_informal}
For regions of fixed finite size, an exponential decay in the following notions of clustering of Gibbs states of geometrically-local, commuting Hamiltonians is equivalent, and implied by uniform gap of the Davies generator:
\begin{itemize}
    \item[1.] Uniform $\mathbb{L}_2$-clustering (\Cref{def:L2Clustering})
    \item[2.] Uniform decay of covariance ($\mathbb{L}_\infty$-clustering) (\Cref{def:LinfinityClustering})
    \item[3.] Uniform local indistinguishability (\Cref{def:local_indistinguishability})
    \item[4.] Uniform decay of mutual information (\Cref{def:mutual_information})
    \item[5.] Uniform strong local indistinguishability (\Cref{def:strong_local_indistinguishability})
    \item[6.] Uniform mixing condition (\Cref{def:mixing_condition})
\end{itemize}
\end{theorem}
These last two notions together $(5.,6.)$ also imply uniform q$\mathbb{L}_1\to\mathbb{L}_\infty$ clustering (\Cref{equ:q1toinfty_clustering}), which will be instrumental in the proof of the main result. For the formal definition of these notions see the respective definitions in \Cref{sec:static} and for a more detailed picture of the implications see \Cref{fig:enter-label21}. 
\begin{figure}[h]
    \centering
    \includegraphics[width=0.9\textwidth]{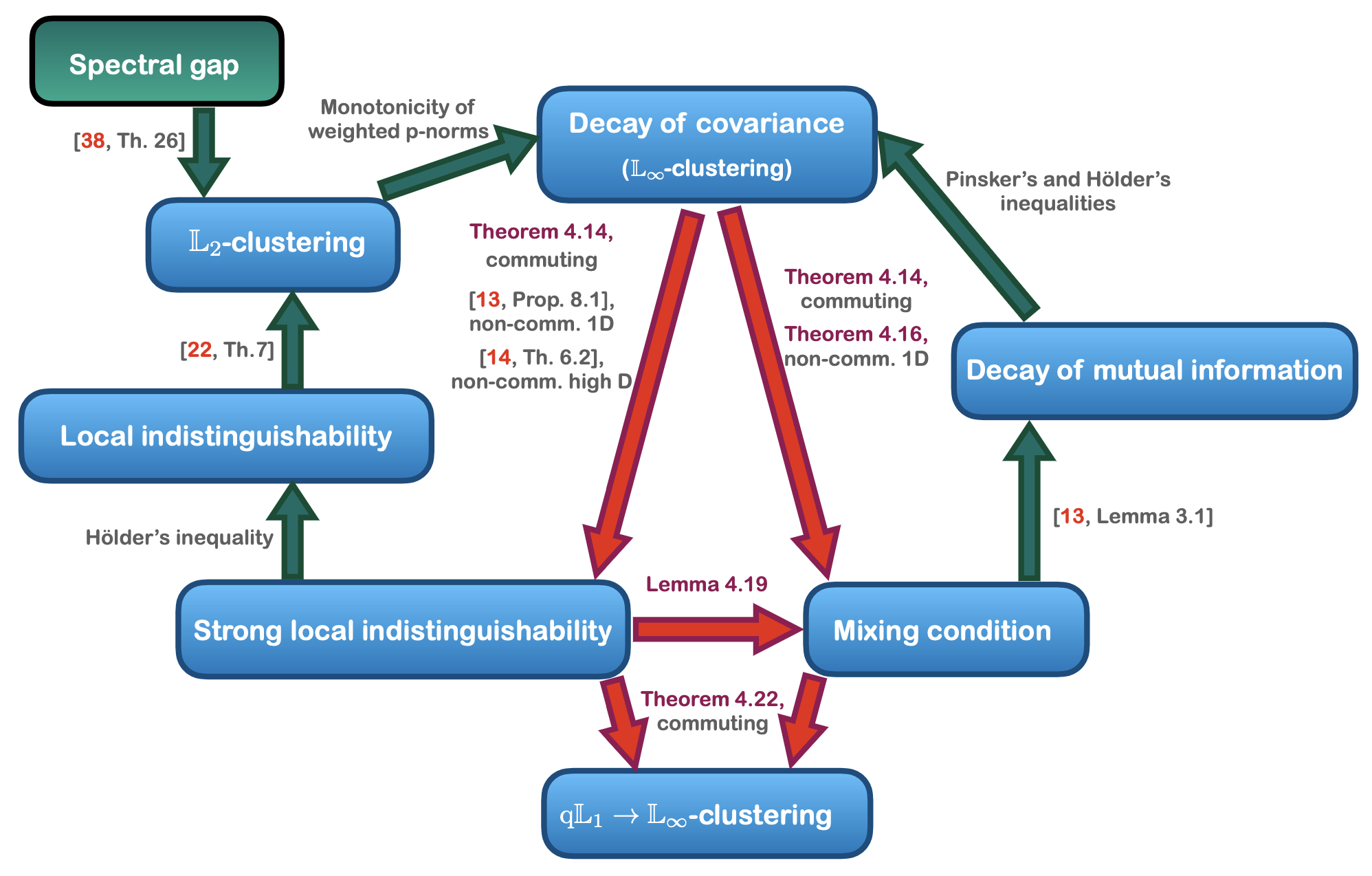}
    \caption{Relation between the notions of decay of correlations listed in \Cref{thm:equiv_decay_corr_informal}. The red arrows represent the new connections derived in this paper, whereas the green arrows signal previously known results.}
    \label{fig:enter-label21}
\end{figure}

In the 1D setting, \Cref{thm:equiv_decay_corr_informal} also holds for arbitrary regions, allowing us to go directly from gap to system-size invariant MLSI. However, for higher dimensional lattices the decay functions of $4.,5.,6.,$ assuming $1.,2.,$ or $3.$ have prefactors depending exponentially on the boundaries of the regions, which is why for rapid mixing of commuting systems on higher dimensional lattices we require the stronger assumption of high temperature.

\section{Preliminaries}\label{sec:prelim}

\subsection{Spin systems on graphs}
We now describe the graphs that underlie the interactions among the particles. A graph is a tuple $\Lambda = (V,E_V)$ of vertex set $V$ and edge set $E_V\subset V\times V$. A complete subgraph $\Gamma \subset\Lambda$ is a tuple $(G,E_G)$, where $G\subset V$ and $E_G$ contains all edges in $E_V$ which contain the vertices in $G$. Abusing notation slightly, we call complete subgraphs \textit{subsets}, writing $\Gamma\subset\Lambda$. For simplicity of notation we associate the graph with its vertex set. Hence we may write $x\in\Lambda$, $A\subset\Lambda$, or $x\in\Gamma$ for an $x\in V$, when the edge set $E_V$ of $\Lambda=(V,E_V)$ is clear from context. 

We define the size of a graph $\Lambda$, or of a subset $\Gamma\subset\Lambda$, denoted as $|\Lambda|, |\Gamma|$, respectively, as the number of vertices it contains. When emphasizing that $\Gamma$ is a finite subset of $\Lambda$, i.e. $|\Gamma|<\infty$, we write $\Gamma\subset\subset\Lambda$.
We write $CD\subset \Lambda$ for the complete subgraph containing all of the vertices of $C$ and $D$, so in this sense $CD = C\cup D$. Note that this does not require $C,D$ to be disjoint.
We call a subset of vertices $\Gamma\subset\Lambda$ connected, if for any two vertices $x,y\in\Gamma$ there exists a sequence of pairwise overlapping edges in $E_G$, such that the first overlaps with $x$ and the last with $y$. 

The graph \textit{distance} $d$ (on $\Lambda$) between two vertices $x\neq y\in\Lambda$ is defined as the minimal length of a connected subset of edges which overlap both with $x$ and $y$. We also set $d(x,x)=0 \, , \; \forall x\in\Lambda$. The length of a subset of edges is given by the number of edges it contains. The distance between two subsets $A,B\subset\Lambda$ is defined as the minimal graph distance between pairs of points in $A$ and $B$, respectively. It is denoted, with slight abuse of notation, with the same symbol $d$. 
We define the diameter of a set $A\subset\Lambda$ as diam$(A):=\sup_{x,y\in A}d(x,y)$. 

The graph has \textit{growth constant} $\nu>0$ defined as the smallest positive number such that, for any $m\in\mathbb{N}$, the number of connected subsets of size $m$ containing some edge, for any edge, is bounded by $\nu^m$:
\begin{align}
    n_m:= \sup_{e\in E_V}|\{F\subset E_V \text{ connected } |\; F\ni e, |F|=m \}| \leq \nu^m.
\end{align}
Note that any \textit{regular graph}, i.e. one where every vertex has the same number of neighbours as every other, has finite growth constant. For example, the growth constant of the $D$-dimensional hypercubic lattice ($\mathbb{Z}^D$) is bounded by $2D e$, where $e$ is Euler's number
\cite{art:LocalityofTemperature, art:GrowthConstant10.1214/ECP.v16-1612}. We say a graph is \textit{2-colorable} if there exists a labeling of the graph with labels 0 and 1, i.e. a map which assigns each vertex one label, such that adjacent vertices, i.e. ones which are connected by some edge, have different labels.

\begin{definition} For an infinite graph $\Lambda$ we define $N(l):=\sup_{x\in\Lambda}|B_l(x)|$, where $B_l(x):=\{v\in \Lambda \, |\, d(x,v)\leq l \}$ is the ball of radius $l$ around vertex $x$.
We call a graph \textit{sub-exponential} if there exists a $\delta\in(0,1)$ s.t. $N(l)\leq \exp(l^\delta)$ holds eventually, i.e. if $\log N(l)=\mathcal{O}(l^\delta)_{l\to\infty}$. Analogously, we call it \textit{exponential} if no such $\delta$ exists, i.e. if $\log N(l)=\Omega(l)_{l\to\infty}$.
\end{definition}
First note that all graphs with finite growth constant are in either of these two classes, since we can crudely bound $|B_l(x)|\leq \nu^l$ and hence $N(l)\leq \nu^l$. Hypercubic lattices are sub-exponential under this definition, whereas $b$-ary trees are exponential. 

We denote the complement of some set $A\subset\Lambda$ as $A^c:=\Lambda\setminus A$.
We will often also consider geometrically-$r$-local interactions, with $r>1$ some integer, on such graphs. For some fixed $r$ we define the boundary of a subset $A\subset\Lambda$, denoted with $\partial A$, to be all vertices in $\Lambda\setminus A$ that are within graph distance $r-1$ from vertices in A \begin{align}
    \partial A &:= \{x\in\Lambda\setminus A \, |\, d(x,A) < r\}, \\
    A\partial &:=A\cup \partial A.
\end{align}
It will be clear from context what $r$ and hence the set-boundary $\partial$ is.
Hence, for nearest neighbour interactions $(r=2)$, $\partial A$ coincides with the usual set-boundary. Whenever we consider $\Lambda =ABC$ with $B$ shielding $A$ from $C$, we denote by $\partial_A B$ the boundary of $B$ in $A$.
An important class of graphs considered here are hypercubic lattices of dimension $D\in\mathbb{N}$, $\Lambda=\mathbb{Z}^D$, with the graph distance equal to the Hamming distance.  Another example is the complete infinite $b$-ary tree $\mathbb{T}_b$, for some integer $b\geq 1$. These are loop-free, exponential, and two-colorable graphs. where each vertex has exactly $b$ neighbours. Each tree has one vertex, called the \textit{root}, from which the tree extends, and whose $b$ neighbours are called its children or leaves. Every other vertex has exactly $b-1$ children or leaves. 

\subsection{General Notation}
A quantum spin system on a finite graph $\Gamma = (V,E_V)$ is described by the Hilbert space \begin{align}
    \mathcal{H}_\Gamma := \bigotimes_{x\in V}\mathcal{H}_x,
\end{align} where each local Hilbert space $\mathcal{H}_x$ has dimension $d<\infty$, i.e.~describes a qudit system. Hence the global dimension of the system is $\dim(\mathcal{H}_\Gamma)=d^{|\Gamma|}$. We will only be considering finite-dimensional Hilbert spaces in this work.
We denote the algebra of bounded linear operators over $\mathcal{H}_\Gamma$ by $\mathcal{B}(\mathcal{H}_\Gamma)$ and the set of density operators with $\mathcal{D}(\mathcal{H}_\Gamma):=\{\rho\in\mathcal{B}(\mathcal{H}_\Gamma) \, | \,  \rho \geq 0, \Tr[\rho]=1 \}$. Note that this algebra is *-homeomorphic to $\mathcal{A}_{\Gamma*}$, the pre-dual of $\mathcal{A}_{\Gamma}=\mathcal{B}(\mathcal{H}_\Gamma)$ w.r.t. to the Hilbert-Schmidt inner product induced by the canonical trace on the finite-dimensional Hilbert space $\mathcal{H}_\Gamma$, i.e. the map $\langle X,Y\rangle=\Tr[X^*Y]$. 
We recall that we can associate to each normalized state (a positive, linear functional) its density operator representation, i.e. for $\omega\in\mathcal{A}_{\Gamma*}$ there exists a $\rho\in\mathcal{D}(\mathcal{H}_\Gamma)$, s.t. $\omega(X) =\Tr[\rho X]$, and the other way around. We denote the trace-class operators on a Hilbert space $\mathcal{H}$ with $\mathcal{B}_1(\mathcal{H})$. The norm on $\mathcal{B}(\mathcal{H})$ is the usual operator norm, denoted by $\|A\|\equiv\|A\|_{\infty}$ for $A\in\mathcal{B}(\mathcal{H})$. The norm on $\mathcal{D}(\mathcal{H})$ is the usual trace-norm, denoted by $\|\rho\|_1:=\Tr[|\rho|]$ for $\rho\in\mathcal{D}(\mathcal{H})$, where the trace on the full Hilbert space $\mathcal{H}_\Gamma$ is denoted as $\Tr[\cdot]$. Moreover, the partial trace on the Hilbert space corresponding to a region $A\subset\Gamma$ is denoted as $\tr_{A}[\cdot]: \mathcal{B}_1(\mathcal{H}_\Gamma)\to \mathcal{B}_1(\mathcal{H}_{\Gamma\setminus A})$.

We denote the identity operator on $\mathcal{H}$ as $\1\equiv\1_\mathcal{H}\in\mathcal{B}(\mathcal{H})$ and the identity map $\mathcal{B}(\mathcal{H})\to\mathcal{B}(\mathcal{H})$ as $\id\equiv\id_{\BH}$.
Given a linear map $\Phi :\mathcal{B}(\mathcal{H})\to\mathcal{B}(\mathcal{H})$ we denote its pre-dual with respect to the Hilbert-Schmidt inner product as $\Phi_*$. We call such a map $\Phi$ a unital CP map if it is completely positive and identity-preserving $\Phi(\1)=\1$.   Their pre-duals $\Phi_*:\mathcal{B}_1(\mathcal{H})\to \mathcal{B}_1(\mathcal{H})$ are completely positive, trace-preserving maps (CPTP or quantum channels). 

We denote the spectrum of an operator $A\in\mathcal{A}$ with $\text{spec}(A)$. 
We define the subregion of the graph on which the operator $A$ is non-trivially supported, denoted as $\supp(A)\subset\Gamma$ as the smallest subregion $X\subset\Gamma$ s.t. $A=A^\prime_X\otimes\1_{\Gamma\setminus X}$, for a suitable $A^\prime_X\in\mathcal{B}(\mathcal{H}_X)$.

We will employ the following \textit{``big-O''}-notation $\mathcal{O}(g(x))_{x\to\infty}$ when meaning that $f(x)=\mathcal{O}(g(x))$ for $x\to\infty$, i.e. to indicate in which limit the scaling $\mathcal{O}(g(x))$ holds for a function $f(x)$, and we use the $``big-\Omega"$ notation similarly.

\subsection{Weighted non-commutative $\mathbb{L}_{p,\sigma}$-spaces and inner products} \label{subsec:weightednorms}
We will make use of so-called (weighted) non-commutative $\mathbb{L}_p$ spaces in this work.
For a general overview and construction of such spaces on von Neumann algebras, see e.g. \cite{PISIER20031459}. 
In our case of finite-dimensional Hilbert spaces, the non-weighted ones are just the $p$-Schatten spaces on $\mathcal{H}$.
The non-commutative $p$-Schatten norm of $X\in \mathcal{B}(\mathcal{H})$ is defined as
\begin{align}
    &\|X\|_p:=\Tr[|X|^p]^\frac{1}{p} \hspace{3.3cm} 1\leq p < \infty, \\
    &\|X\|_\infty = \|X\|.
\end{align}
Given a full-rank state $\sigma\in\mathcal{D}(\mathcal{H})$, we define the \textit{weighted non-commutative} $\mathbb{L}_{p,\sigma}$ norm of $X$ as
\begin{align}
   \|X\|_{p,\sigma}&:=\Tr[|\sigma^{\frac{1}{2p}}X\sigma^{\frac{1}{2p}}|^p]^{\frac{1}{p}} \hspace{2cm}  1\leq p<\infty, \\
   \|X\|_{\infty,\sigma} &:= \|X\|_{\infty} \equiv \|X\|,
\end{align} respectively. These norms turn the $\mathbb{L}_{p,\sigma}$ spaces into Banach spaces for $p\in[1,\infty]$ and satisfy the usual Hölder-type inequality, Hölder duality, and monotonicity in $p$ for fixed $\sigma$, see e.g. \cite{olkiewicz1999hypercontractivity,art:QuantumGibbsSamplers-kastoryano2016quantum}.
Similarly, $\mathbb{L}_{2,\sigma}$ is a Hilbert space with respect to the \textit{KMS-inner product}
\begin{align}\label{eq:KMSprod}
    \langle X,Y\rangle_\sigma^{\KMS} := \Tr[\sqrt{\sigma}X^*\sqrt{\sigma}Y].
\end{align}
There exists a natural embedding $\Gamma_\sigma:\mathbb{L}_{1,\sigma}\to \mathbb{L}_{1}$ via
$
    \Gamma_\sigma(X):= \sqrt{\sigma}X\sqrt{\sigma}$.
Hence the weighted $p,\sigma$-norm can also be expressed as $\|X\|_{p,\sigma}=\|\Gamma_\sigma^{\frac{1}{p}}(X)\|_p$. 
For completeness we define the \textit{modular operator} of $\sigma$ here as
\begin{align}
    \Delta_\sigma(X):=\sigma X \sigma^{-1},
\end{align} and the \textit{modular group} of $\sigma$ as $\{\Delta_{\sigma}^{is}\}_{s\in\mathbb{R}}$. 
Additionally, the \textit{GNS-inner product} on $\mathcal{B}(\mathcal{H})$, for a finite dimensional Hilbert space $\mathcal{H}$ is
\begin{align}\label{eq:GNSprod}
    \langle X,Y\rangle_\sigma^{\GNS} := \Tr[\sigma X^*Y].
\end{align}

\subsection{Uniform families of Hamiltonians} 
In this work we consider families of many-body Hamiltonians $\{H_\Gamma\}_{\Gamma\subset\subset \Lambda}$, such that
\begin{align} \label{def:UnfiformHamiltonian}
    H_\Gamma=\sum_{X\subset \Gamma}\Phi_X
\end{align} for a given fixed \textit{potential} $\Phi:\{X\subset\Lambda\}\to\mathcal{A}_\Lambda:X\mapsto\Phi_X$.\footnote{Here $\mathcal{A}_\Lambda$ represents the closure of the algebra created by local operators on $\Lambda$, sometimes also called the algebra of quasi-local observables, when $|\Lambda|=\infty$ and $B(\mathcal{H}_\Lambda)$ when it is finite $|\Lambda|<\infty$.} Note that for each $X\subset\Lambda$, $\Phi_X$ is a self-adjoint operator acting only non-trivially on the sub-region $X$. 
The potential is called \textit{commuting} (on $\Lambda$) if for each $X,Y\subset\Lambda$, $\Phi_X$ and $\Phi_Y$ commute. 
It is said to have \textit{bounded interaction strength} $J:=\max_{X\subset\Lambda}\{\|\Phi_X\|\}$ and \textit{interaction range} $r:=\max\{\text{diam}(X) \, | \, X\subset\Lambda, \Phi_X\neq 0\}$, where $\text{diam}(X)$ stands for the diameter of region $X$ with respect to the graph distance. We will call potentials with interaction range $r$ \textit{geometrically-r-local}. When $r=2$, we will interchangeably use the terms geometrically-$2$-local and \textit{nearest-neighbour}.
We call this family $\{H_\Gamma\}_{\Gamma\subset\subset \Lambda}$ the \textit{ associated family of Hamiltonians} to $(\Lambda,\Phi)$. It is called a \textit{uniform} $J$-bounded, geometrically-$r$-local, commuting family if the potential is $J$-bounded, geometrically-$r$-local and commuting.\footnote{Hence the constants $J,r$ do not depend on the regions $\Gamma$, and explicitly on $|\Gamma|$, on which the local Hamiltonians are defined.} In this work, we will only consider such uniform families, unless explicitly stated otherwise.

The associated Gibbs state of the local Hamiltonian on $A\subset\Gamma$ at inverse temperature $\beta$ is denoted by
\begin{align}
    \sigma^A:=\frac{e^{-\beta H_A}}{\Tr[e^{-\beta H_A}]}\in\mathcal{D}(\mathcal{H}_\Gamma),
\end{align} while the reduced state onto some subregion $A\subset\Gamma$ is denoted by
\begin{align}
    \sigma_A:=\tr_{\Gamma\setminus A}\sigma^\Gamma\in\mathcal{D}(\mathcal{H}_A),
\end{align} where $\sigma\equiv\sigma^\Gamma=\sigma_\Gamma$.
Given a potential $\Phi$ on $\Lambda$ and some inverse temperature $\beta>0$ we call the family of Gibbs states $\{\sigma^\Gamma\}_{\Gamma\subset\subset\Lambda}$, where the Hamiltonian $H_\Gamma$ is given as in \eqref{def:UnfiformHamiltonian} w.r.t this $\Phi$, the \textit{family of Gibbs states associated to} $(\Lambda,\Phi,\beta)$.
We will employ the convenient notation $E_{X,Y}:=e^{-H_{XY}}e^{H_X+H_Y}$ for Araki's expansionals for two disjoint subsets $X,Y\subset\Lambda$ from \cite{art:Bluhm2022exponentialdecayof}. 

\subsection{Quantum Markov semigroups and Lindbladians} 
A \textit{quantum Markov semigroup} (QMS) is a strongly continuous one-parameter semigroup of unital CP maps $\{\Phi_t\}_{t\geq 0}:\mathcal{B}(\mathcal{H})\to\mathcal{B}(\mathcal{H})$. 
This is a family such that $\Phi_0=\id_{\BH}, \Phi_{s+t}=\Phi_s\circ\Phi_t \, , \; \forall s,t\geq 0,$ and $ \lim_{t\downarrow0} \|(\Phi_t-\id)(X)\|=0\, , \; \forall X\in\mathcal{B}(\mathcal{H})$. 
By the Hille-Yosida theorem there exists a densely defined generator, called the \textit{Lindbladian}
\begin{align}
    \mathcal{L}(X):=\lim_{t\downarrow0}\frac{1}{t}(\Phi_t-\id)(X),
\end{align} such that the semigroup is given as $\Phi_t= e^{t\mathcal{L}} \, , \; \forall t\geq 0$.
In our case of a finite-dimensional Hilbert space, the Lindbladian is defined on all of $\mathcal{B}(\mathcal{H})$ and its pre-dual on all of $\mathcal{B}_1(\mathcal{H})$. 

A QMS with generator $\Li$ gives the unique solution to the master equation $\frac{d}{dt}\rho(t)=\mathcal{L}(\rho(t))$. A state $\sigma$ is invariant (or stationary) if $\Phi_{t*}(\sigma)=\sigma$ for all $t\geq0$, which is equivalent to $\Li_*(\sigma)=0$. 
We call a QMS and its generator \textit{faithful} if the QMS admits a full-rank invariant state $\sigma\in\mathcal{D}(\mathcal{H})$ and \textit{primitive} if this state is unique.

We call a QMS and its generator \textit{reversible} or \textit{KMS-symmetric} w.r.t. a state $\sigma$ if the QMS is symmetric w.r.t the KMS-inner product and similarly for
\textit{GNS-symmetric}.
In the latter case, we also say that the QMS satisfies the \textit{detailed balance condition}, which is equivalent to
\begin{align}
    \Tr[\sigma X^*\Li(Y)] = \Tr[\sigma \Li(X)^*Y] \hspace{1cm} \forall X,Y\in\BH.
\end{align}
If a QMS is GNS-symmetric w.r.t a state $\sigma$, then this state is necessarily a stationary one.

Given a graph $\Lambda$ and a finite subset $\Gamma\subset\subset\Lambda$, we consider a family of Lindbladians $\Li_\Lambda=\{\Li_\Gamma\}_{\Gamma\subset\subset\Lambda}$, such that
\begin{align}\label{eq:locallind}
    \Li_\Gamma = \sum_{X\subset \Gamma}L_X,
\end{align} where $\{L_X\}_{X\subset\Lambda}$ is a fixed family of local Lindbladians, such that $\tilde{J}:=\sup_{X\subset\Lambda}\|L_{X*}\|_{1\to1,\text{cb}}<\infty$ and $L_{X*}=0$ whenever diam$(X)>\tilde{r}$.\footnote{Here $\|\Phi_*\|_{1\to1,\text{cb}}$ is the completely bounded $1\to1$ norm, i.e. $$\|\Phi_*\|_{1\to1,\text{cb}} := \sup_{n\in\mathbb{N}}\sup_{\rho\in\mathcal{D}(\mathbb{C}^n\otimes\mathcal{H})}\|(\id_n\otimes\Phi_*)(\rho)\|_1 .$$ } Hence we call these a \textit{uniform} \textit{ geometrically$-\tilde{r}-$local, $\tilde{J}-$bounded} family of (bulk) Lindbladians, in analogy with the Hamiltonian case. 
The family is called \textit{locally reversible}, if each element $\mathcal{L}_\Gamma$ is reversible (KMS-symmetric) w.r.t. the Gibbs state of any larger region $\sigma^\Pi$ for $\Gamma\subset\Pi$. 
Note that this implies that it also is \textit{frustration-free}, that is for any two finite subsets $\Gamma\subset\Pi\subset\subset\Lambda$, the stationary states of $\Li_\Pi$ are also stationary under $\Li_\Gamma$, i.e. ker($\Li_\Pi)\subset$ker($\Li_\Gamma$).


Note that if a QMS is GNS-symmetric, i.e. satisfies detailed balance, then it is also KMS-symmetric, i.e. reversible.
For a region $\Gamma\subset\subset\Lambda$, we can write the projection onto the fixed point subalgebra of $\Li_\Gamma$ as 
\begin{align} \label{def:FixedpoitProjection}
    E_\Gamma(\cdot):=\lim_{t\to\infty}e^{t\mathcal{L}_\Gamma}(\cdot).
\end{align}
It turns out that for a primitive, frustration-free uniform family, these projections are conditional expectations w.r.t the family of stationary states. See \Cref{sec:ConditionalExpectations} for more details.

Here we will be working with the \textit{Davies generators} $\Li^D_\Lambda=\{\Li^D_\Gamma\}_{\Gamma\subset\subset\Lambda}$, which is a physically motivated uniform family of Lindbladians associated to a uniform family of Hamiltonians. These are introduced in \Cref{subsec:Davies}.
In the setting we are considering, they are a uniform geometrically-bounded, locally GNS-symmetric 
family of Lindbladians which describe thermalization of a spin system.

We call a uniform family of Lindbladians which are locally GNS-symmetric 
and frustration-free w.r.t to a set of Gibbs states $\{\sigma^\Gamma\}_{\Gamma\subset\subset\Lambda}$ a \textit{quantum Gibbs sampler} of the system $(\Lambda,\Phi,\beta)$ when $\{\sigma^\Gamma\}_{\Gamma\subset\subset\Lambda}$ is the uniform family of Gibbs states associated to $(\Lambda,\Phi,\beta)$. The Davies generators (see \Cref{subsec:Davies}), the Heat-bath generators \cite{art:QuantumGibbsSamplers-kastoryano2016quantum} and the Schmidt generators \cite{art:bravyi2004commutative, art:2localPaper} are examples of quantum Gibbs samplers. 

\subsection{Conditional Expectations} \label{sec:ConditionalExpectations}

Given a von Neumann subalgebra $\mathcal{N}\subset\mathcal{B}(\mathcal{H})$, a \textit{conditional expectation} onto $\mathcal{N}$ is a completely positive, unital map $E_\mathcal{N}:\mathcal{B}(\mathcal{H})\to\mathcal{N}$, such that
\begin{align}
    &E_\mathcal{N}(X)=X \, , \; \forall  X\in\mathcal{N} \\
    &E_\mathcal{N}(aXb)= aE_\mathcal{N}(X)b \, , \; \forall  a,b\in\mathcal{N}, X\in\mathcal{B}(\mathcal{H}).
\end{align}
By complete positivity and unitality, it follows that the preadjoint of any conditional expectation with respect to the Hilbert-Schmidt inner product $E_{\mathcal{N}*}:\mathcal{N}_*\to\mathcal{B}(\mathcal{H})_*$ is a completely positive, trace-preserving map, i.e. a quantum channel.
Any \textit{conditional expectation} onto $\mathcal{N}$ for which there exists a full-rank state $\sigma\in\mathcal{D}(\mathcal{H})$ which satisfies
\begin{align}
    E_{\mathcal{N}*}(\sigma)=\sigma \Longleftrightarrow \Tr[\sigma E_\mathcal{N}(X)]=\Tr[\sigma X] \, , \; \forall  X\in\mathcal{B}(\mathcal{H}),
\end{align} is said to be \textit{with respect to the state} $\sigma$ 
\cite{art:2localPaper,art:EntropyDecayOf1DSpinChain-Cambyse}.
Let $E$ be a conditional expectation with respect to a full-rank state $\sigma$ onto $\mathcal{N}$, then from the definition it follows that it is self-adjoint with respect to the $\sigma-$KMS inner product, i.e.
\begin{align}
    \sigma^{\frac{1}{2}}E(X)\sigma^{\frac{1}{2}} = E_*(\sigma^{\frac{1}{2}}X\sigma^{\frac{1}{2}})
\end{align} holds for any $X\in\mathcal{B}(\mathcal{H})$. 

Furthermore, it can be shown that $E$ commutes with the modular automorphism group of $\sigma$, i.e.
\begin{align}
    \Delta_{\sigma}^{is}\circ E = E \circ \Delta_\sigma^{is} \, , \; \forall s\in\mathbb{R}.
\end{align}  
Moreover, given a *-subalgebra $\mathcal{N}\subset\mathcal{B}(\mathcal{H})$ and a full-rank state $\sigma\in\mathcal{D}(\mathcal{H})$, the existence of a conditional expectation w.r.t. $\sigma$ onto $\mathcal{N}$ is equivalent to the invariance of $\mathcal{N}$ under the modular automorphism group $\{\Delta_\sigma^{is}\}_{s\in\mathbb{R}}$. Furthermore, in the case that the  *-subalgebra $\mathcal{N}$ is invariant under the modular automorphism group of said full-rank state $\sigma$, this conditional expectation is uniquely determined by $\sigma$ \cite{art:TAKESAKI1972306, art:2localPaper}.

Conditional expectations with respect to some full-rank state $\sigma$ between finite-dimensional matrix algebras, as all the ones in this work, can be given in an explicit form, see e.g. \cite{art:EntropyDecayOf1DSpinChain-Cambyse}. Any finite dimensional subalgebra $\mathcal{N}\subset\mathcal{B(H)}$ is unitarily isomorphic to the algebra
\begin{align}
    \mathcal{N}=\bigoplus_{i=1}^n \mathcal{B}(\mathcal{H}_i)\otimes \mathbb{C}\1_{\mathcal{K}_i}, \text{ where } \mathcal{H} = \bigoplus_{i=1}^n\mathcal{H}_i\otimes\mathcal{K}_i.
\end{align}
Now there exist density operators $\{\tau_i\in\mathcal{D}(\mathcal{K}_i)\}_{i=1}^n$ and projections $\{P_i\in\mathcal{B}(\mathcal{H}_i)\}_{i=1}^n$, respectively, onto $\{\mathcal{H}_i\otimes\mathcal{K}_i\}$ such that
\begin{align}
    E_\mathcal{N}(X) = \bigoplus_{i=1}^n \tr_{\mathcal{K}_i}[P_iXP_i(\1_{\mathcal{K}_i}\otimes \tau_i)]\otimes \1_{\mathcal{K}_i} \ \Longleftrightarrow \ E_\mathcal{N*}(\rho) = \bigoplus_{i=1}^n\tr_{\mathcal{K}_i}[P_i\rho P_i]\otimes \tau_i,
\end{align} for $X\in\mathcal{B}(\mathcal{H})$ and $\rho\in\mathcal{D}(\mathcal{H})$ \cite{art:EntropyDecayOf1DSpinChain-Cambyse}. 
An important example of conditional expectations are the following.

\begin{definition}[Local Davies Lindbladian Projectors] \label{example:LocalLindbaldProj} 
Let $\Gamma$ be some finite graph. 
Let $\mathcal{L}_\Gamma =\{\Li_A\}_{A\subset\Gamma}$ 
be a uniform locally reversible family of Lindbladians with respective stationary states $\{\sigma^A\}_{A\subset\Gamma}$. The local Lindbladian projector associated with the family $\mathcal{L}_\Gamma$ on $A\subset\Gamma$ is given by
\begin{equation}
    E_A(X):=\lim_{t\to\infty}e^{t\mathcal{L}_A}(X)
\end{equation}
for $X\in\mathcal{B}(\mathcal{H}_\Gamma)$. Notice that each $E_A$ acts only non-trivially on $A^c$ and since $\Li_\Gamma$ is frustration-free, $E_A$ is a conditional expectation with respect to the stationary state $\sigma^\Gamma$ onto the subalgebra $\1_{A\partial}\otimes\mathcal{B}(\mathcal{H}_{(A\partial)^{\text{c}}})$.
\end{definition}
For a proof of these claims see e.g. \cite[Proof of Proposition 9]{art:QuantumGibbsSamplers-kastoryano2016quantum}. In this case the expectation value of any observable w.r.t. the invariant state on the full system is
\begin{align}
    \Tr[\sigma X]= \Tr[\sigma E_\Gamma(X)] = \Tr[\sigma E_A(X)] \; \; \forall A\subset\Gamma.
\end{align}
On the other hand, given a family of local conditional expectation $E_A: \mathcal{B}(\mathcal{H})\to \1_{A\partial}\otimes\mathcal{B}(\mathcal{H})_{(A\partial)^{\text{c}}}$ w.r.t. the same state $\sigma\in\mathcal{D}(\mathcal{H})$, then
\begin{align}
    \overline{\mathcal{L}}_A:= (E_A-\id_A),
\end{align} is a family of locally reversible, frustration-free Lindbladians with invariant state $\sigma$.

\subsection{The relative entropy and strong data processing} 

The \textit{Umegaki relative entropy} \cite{petz_quasi-entropies_1986} between two finite-dimensional quantum states given by their density operators $\rho,\sigma\in\mathcal{D}(\mathcal{H})$ is defined as
\begin{align}
   D(\rho\|\sigma) := \begin{cases}
   \Tr[\rho(\log\rho-\log\sigma)]
    & \text{if supp}(\rho)\subset\text{supp}(\sigma) \\ 
   \infty & \text{else} \end{cases},
\end{align} where the logarithm here is the natural logarithm to base $e$. 
Pinsker's inequality \eqref{equ:Pinsker} gives an upper bound on the trace-distance in terms of the relative entropy:
\begin{align}
    \|\rho-\sigma\|_1^2 \leq 2D(\rho\|\sigma).
    \label{equ:Pinsker}
\end{align}

The relative entropy satisfies the \textit{data processing inequality}, so that no quantum channel, i.e. CPTP map $\Phi_*$, can increase it between any two states,
\begin{align}
    D(\Phi_*(\rho)\|\Phi_*(\sigma))\leq D(\rho\|\sigma).
    \label{equ:DPI}
\end{align}
The core part of the entropic-inequalities approach to thermalization relies upon a strengthening of this inequality. We say a quantum channel $\Phi_*$ satisfies a non-trivial \textit{strong data processing} (sDPI)  with \textit{contraction coefficient} $\eta\equiv\eta(\Phi_*)<1$ if for any pair $(\rho,\sigma)$ of states, with $\rho\neq\sigma$, it holds that
\begin{align}
    D(\Phi_*(\rho)\|\Phi_*(\sigma))\leq \eta(\Phi_*)D(\rho\|\sigma).
    \label{equ:sDPI}
\end{align}
More formally, we define the contraction coefficient for a GNS symmetric QMS $\Phi_*\equiv\Phi_{t_0*}$ as
\begin{align}
    \eta(\Phi_*):=\sup_{\rho\in\mathcal{D}(\mathcal{H})}\frac{\inf_{\sigma\in\Sigma}D(\Phi_*(\rho)\|\Phi_*(\sigma))}{\inf_{\sigma\in\Sigma}D(\rho\|\sigma)},
\end{align} where $\Sigma$ is the set of \textit{stationary states} of $\Phi_*$. These are all the density operators which are left invariant under the action of the channel.\footnote{In the case of a general quantum channel (instead of QMS) one would need to replace $\Sigma$ by the decoherence-free subalgebra. If $\Phi_*(X)=\sum_{k}A_kXA_k^*$ is the Kraus representation of $\Phi_*$, then the decoherence-free subalgebra is $\Sigma :=\bigcap_{k\in\mathbb{N}}\mathcal{N}(\Phi_*^k)$, where $\mathcal{N}(\Phi_*):=\text{Alg}\{X\in\mathcal{B}(\mathcal{H})\, | \, [X,A_i^*A_j]=0 \ \forall i,j\}$} Assume we have some channel $\Phi_*$ which has a contraction coefficient $\eta(\Phi_*)<1$ and a unique invariant state $\sigma$, i.e. $\Phi_*(\sigma)=\sigma$. Then sDPI immediately induces an exponential decay of the relative entropy in the number of times the channel is applied.
\begin{align}
    D(\Phi_*^n(\rho)\|\sigma)=  D(\Phi_*^n(\rho)\|\Phi_*^n(\sigma)) \leq \eta^n D(\rho\|\sigma).
\end{align}
Since conditional expectations are, by definition, projections on closed-*-subalgebras (which are convex), the following chain rule holds for states $\rho,\sigma\in\mathcal{D}(\mathcal{H})$, whenever $E_{\mathcal{N}*}(\sigma)=\sigma$ \cite[Lemma 3.4]{art:ChainRuleForRelEntropyConExpJunge_2022}
\begin{align}
\label{def:relentChainRule}
    D(\rho\|\sigma)= D(\rho\|E_{\mathcal{N}*}(\rho))+D(E_{\mathcal{N}*}(\rho)\|\sigma).
\end{align}

For completeness we define the \textit{max-relative entropy} \cite{Datta_2009} between two finite dimensional quantum states given by their density operators $\rho,\sigma\in\mathcal{D}(\mathcal{H})$ as
\begin{align}
    D_\text{max}(\rho\|\sigma):= \begin{cases}
        \log\inf\{\lambda\in\R| \rho\leq \lambda\sigma \} = \log\|\sigma^{-\frac{1}{2}}\rho\sigma^{-\frac{1}{2}}\| & \text{if supp}(\rho)\subset\text{supp}(\sigma) \\ 
   \infty & \text{else}
    \end{cases},
\end{align}
where the inverses are taken as generalised inverses. It also satisfies the DPI and a triangle inequality and it holds that
\begin{align}
    D(\rho\|\sigma) \leq D_\text{max}(\rho\|\sigma).
\end{align}

The relative entropy also gives rise to the \textit{quantum mutual information} $I$, which measures correlations between two regions. Given a finite graph $\Gamma=ABC$, the mutual information of a state $\rho\in\mathcal{D}(\mathcal{H}_\Gamma)$ between the reduced state on the region $A$ and the one on the region $C$  is defined as
\begin{align}
    I_\rho(A:C):=D(\rho_{AC}\|\rho_A\otimes\rho_C).
\end{align} 
Analogously, the \textit{max-mutual information} $I_\text{max}$ of the state
$\rho\in\mathcal{D}(\mathcal{H}_\Gamma)$ on $\Gamma=ABC$ between the reduced state on the region $A$ and the one on the region $C$  is defined as
\begin{align}
    I_{\text{max},\rho}(A:C):=D_\text{max}(\rho_{AC}\|\rho_A\otimes\rho_C).
\end{align} 

\subsection{Relative entropy decay via the complete modified logarithmic Sobolev Inequality}\label{sec:cMLSI}
Assuming that our QMS has at least one full-rank invariant state $\sigma$ with respect to which it is GNS-symmetric (as is the case for all Davies maps), we can establish the sDPI \eqref{equ:sDPI}  for time-continuous QMS $\{e^{t\mathcal{L}}\}_{t\geq0}$. 

The way to do this is via a differential version of the strong data processing inequality \eqref{equ:sDPI} for the channel $\Phi_{t*}:=e^{t\mathcal{L}_*}$ in which we set $\eta(e^{t\mathcal{L}_*})=e^{-t\alpha}$, yielding
\begin{align}
   -\frac{d}{dt} D(e^{t\mathcal{L}_*}(\rho)\|E_*(\rho)) \big|_{t=0} =: \EP_{\mathcal{L}}(\rho)  \geq \alpha D(\rho\|E_*(\rho)).
   \label{equ:MLSI}
\end{align}
Here $E_*:=\lim_{t\to\infty}e^{t\mathcal{L}_*}$ is the projection onto the stationary states \eqref{def:FixedpoitProjection}, see also Example \ref{example:LocalLindbaldProj}.
We call this inequality \eqref{equ:MLSI} the \textit{modified logarithmic Sobolev inequality} (MLSI) and $\EP_\mathcal{L}(\rho)$ the \textit{entropy production} of the QMS $\{e^{t\mathcal{L}}\}_{t\geq 0}$.
The optimal constant $\alpha$ satisfying the MLSI is called the \textit{modified logarithmic Sobolev constant} (MLSI constant) $\alpha(\mathcal{L})$. It is hence given by
\begin{align}
    \alpha(\mathcal{L}):= \inf_{\rho\in\mathcal{D}(\mathcal{H})}\frac{\EP_\Li (\rho)}{D(\rho\|E_*(\rho))}.
    \label{def:MLSIconstat}
\end{align}
By integration and use of Gronwall's inequality it follows that any QMS $\{e^{t\Li}\}_{t\geq0}$ which satisfies the MLSI with strictly positive MLSI constant $\alpha\equiv\alpha(\mathcal{L})>0$ induces exponential convergence in relative entropy to its stationary states, i.e.
\begin{align}
    D(e^{t\Li_*}(\rho)\|E_*(\rho))\leq e^{-\alpha t}D(\rho\|E_*(\rho)).
\end{align}
One important way to establish the existence of such constants in the classical setting is to exploit its stability under tensorization. This allows us to describe the dynamics of large composite systems via their dynamics on small subregions. This is, however, often not straightforward in the quantum setting, i.e. if we have two QMS $\{e^{t\mathcal{L}}\}_{t\geq0}, \{e^{t\mathcal{K}}\}_{t\geq 0}$, then the joint evolution, given by $\{e^{t\Li}\otimes e^{t\mathcal{K}}\}_{t\ge 0}$ is not necessarily as quickly mixing as the slower individual one, i.e. $\alpha(\Li+\mathcal{K})\not\geq \min\{\alpha(\Li),\alpha(\mathcal{K})\}$. \cite{art:brannan2020complete}.
In order to recover the stability under tensorization we introduce the so-called \textit{complete MLSI} (cMLSI) and the \textit{cMLSI constant} 
\begin{align}
    \alpha_c(\mathcal{L}):=\inf_{n\in\mathbb{N}}\alpha(\mathcal{L}\otimes \id_n),
\end{align} where $\id_n:\mathcal{B}(\mathbb{C}^n)\to \mathcal{B}(\mathbb{C}^n)$ is the identity channel \cite{art:CompleteOrder}.
Hence, we say that the QMS $\{e^{t\Li}\}_{t\geq 0}$ satisfies the cMLSI if the QMS $\{e^{t\Li}\otimes\id_n\}_{t\geq 0}$ satisfies the MLSI for all ancillary systems of arbitrary dimension with the same constant.
In \cite{art:brannan2020complete} it was shown that for two QMS with commuting generators $\mathcal{L},\mathcal{K}$, respectively, it holds that
\begin{align}
    \alpha_c(\Li+\mathcal{K})\geq \min\{\alpha_c(\Li),\alpha_c(\mathcal{K})\}.
\end{align}

Next, the following important result from \cite{art:CompleteEntropicInequalities_GaoRouze_2022,art:gao2021geometric} guarantees the existence of positive cMLSI constants for a sufficiently large class of QMS.
\begin{theorem}[\cite{art:CompleteEntropicInequalities_GaoRouze_2022}]\label{thm:finiteregioncMLSI}
For any GNS-symmetric QMS $\{e^{t\mathcal{L}}\}_{t\geq0}$ on the algebra $\mathcal{B}(\mathcal{H})$ of bounded linear operators over some finite dimensional Hilbert space $\mathcal{H}$, we have $ \alpha_c(\mathcal{L})>0$, with $\alpha_c(\mathcal{L})= \Omega\left(\frac{\lambda(\mathcal{L})}{\log\dim\mathcal{H}}\right)$. 
In particular, for many-body quantum systems, this bound on the cMLSI constant gives a `trivial' lower bound that is decreasing at best as $\Omega(|\Gamma|^{-1})$ as long as the gap $\lambda(\mathcal{L}_\Gamma)$ is constant.
\end{theorem} 

Local existence of a strictly positive cMLSI constant is a promising starting point. However, on its own it does not give good bounds for systems in the thermodynamic limit. The strategy used in this work, when showing existence of a cMLSI constant with a better scaling than $\Omega(|\Lambda|^{-1})$, is to use approximate tensorization results to geometrically break down the lattice into finite-size parts. Then, we apply Theorem \autoref{thm:finiteregioncMLSI} to each of these small regions and put them back together into the whole lattice. Such an approach is sometimes called a ``divide-and-conquer strategy", or a global-to-local reduction.

\section{Dynamical properties: Local generators and conditional expectations}
\label{sec:dynamic}

In this work we consider two classes of dissipation dynamics. The first one, commonly known as Davies dynamics, is a physically motivated Quantum Markov Semigroup. It is typically used to model thermalization of finite-dimensional quantum systems weakly coupled to their environment. The main result of this paper is formulated with respect to it, so
we devote Section \ref{subsec:Davies} to the introduction of the Davies evolution and to some technical results concerning it.

The second class of semigroups we explore are the so called \textit{Schmidt generators}, see Section \ref{subsec:SchmidtCondExp} for its definition and the notation used. It serves as a mathematically more tractable model that will be useful in the proofs, even if it lacks a clear physical interpretation. We will use these dynamics as a proxy to derive rigorous bounds on the mixing time of the former through establishment of the modified logarithmic Sobolev inequality.

\subsection{Davies dynamics}\label{subsec:Davies} 
The main dynamics considered here are the Davies dynamics introduced in \cite{art:Davies1archbold_1983}, see also \cite{DaviesBook,art:Davies1}, \cite{art:OpenQuantumSystemsRivas_2012} for a more modern derivation, and \cite{art:QuantumGibbsSamplers-kastoryano2016quantum}, whose notation we will be mostly using, in the context of Gibbs sampling.
On a finite subsystem $\Gamma\subset\subset\Lambda$, its generators, the \textit{Davies Lindbladians} or \textit{generators} (in the Heisenberg picture) are given by
\begin{align}\label{def:globalDavies}
    \Li_\Lambda^D(X) = i[H_\Lambda,X] + \sum_{x\in\Lambda}\mathcal{L}^D_{x}(X) \, ,
\end{align} where $H_\Lambda$ is the lattice Hamiltonian and dissipative terms 
\begin{align}\label{def:dissipativeDavies}
    \Li_x^D(X) =\sum_{\omega,\alpha(x)} \chi_{\alpha(x)}(\omega)\left(S^*_{\alpha(x)}(\omega)XS_{\alpha(x)}(\omega)-\frac{1}{2}(S^*_{\alpha(x)}(\omega)S_{\alpha(x)}(\omega)X+XS^*_{\alpha(x)}(\omega)S_{\alpha(x)}(\omega))\right) \, .
\end{align} Here $\omega\in\text{spec}(H_\Lambda)-\text{spec}(H_\Lambda)$ are the Bohr frequencies, $\chi_{\alpha(x)}(\omega)$ a function satisfying the KMS-condition $\chi_{\alpha(x)}(-\omega)=e^{-\beta\omega}\chi_{\alpha(x)}(\omega)$ and $S_{\alpha(x)}(\omega)$ operators related to the interaction between our lattice and the thermal environment \cite{art:QuantumGibbsSamplers-kastoryano2016quantum}. 
If we assume that our system is described by a uniformly bounded, geometrically-$r$-local, commuting family of Hamiltonians $\{H_\Gamma\}_{\Gamma\subset\subset\Lambda}$, as in \eqref{def:UnfiformHamiltonian}, then the above generator reduces to a \textit{local Davies generator}
\begin{align}\label{def:localDavies}
    \Li^D_\Gamma(X) = i[H_\Gamma,X]+\sum_{x\in\Gamma}\mathcal{L}^D_{x}(X) \, ,
\end{align} for any $\Gamma\subset\subset\Lambda$. 
We call these a \textit{(family of) Davies generators} associated to \textit{ $(\Lambda,\Phi,\beta)$}, whenever the family of Hamiltonians that describe our system is the one associated to $(\Lambda,\Phi)$ and the inverse temperature of the environment is $\beta>0$. Note that these are not unique, and depend on the details of the environment through $\chi_{\alpha(x)}(\omega)$. However, they always correspond to a uniformly bounded, geometrically-local family of Lindbladians satisfying the following properties:

\begin{proposition}(\cite[Lemma 11]{art:QuantumGibbsSamplers-kastoryano2016quantum})
For some inverse temperature $\beta$, some finite graph $\Lambda$, a subset $\Gamma\subset\Lambda$ and $\Phi$ a uniformly bounded, geometrically-$r$-local, commuting potential on $\Lambda$, the associated local Davies generators defined in (\ref{def:globalDavies}-\ref{def:localDavies}) satisfy
\begin{enumerate}
    \item For any subset $\Gamma\subset\Lambda$, $\{e^{t\Li^D_\Gamma}\}_{t\geq 0}$ is a CP unital semigroup with generator $\Li^D_\Gamma$.
    \item The family $\Li^D=\{\Li^D_\Gamma\}_{\Gamma}$ is geometrically-local, in the sense that each individual term $\Li^D_x$ acts only nontrivially on the region $B_{\Tilde{r}}(x)$ for some fixed radius $r\leq \tilde{r}\leq 2r$.
    \item The family $\Li^D=\{\Li^D_\Gamma\}_{\Gamma}$ is locally reversible and satisfies detailed balance w.r.t the global Gibbs state $\sigma^\Lambda$.
    \item The family $\Li^D=\{\Li^D_\Gamma\}_{\Gamma}$ is frustration-free.
\end{enumerate}
\end{proposition}
Recalling \Cref{sec:ConditionalExpectations}, we see that $E^D_\Gamma(X):=\lim_{t\to\infty}e^{t\Li^D_\Gamma}(X)$ is a conditional expectation, called the \textit{Davies conditional expectation}.

\subsection{Dynamics generated by Schmidt conditional expectations}\label{subsec:SchmidtCondExp}

Although the Davies dynamics are our target due to their physicality, their corresponding conditional expectations are not always mathematically tractable. For this reason it will be very useful to consider the so-called Schmidt conditional expectation, as introduced in \cite{art:2localPaper} inpired by \cite{art:bravyi2004commutative}. As we show below, their conditional expectation are much more tractable, yet the dynamics they induce are closely related to the Davies. The construction and results in this section work for any 2-colorable graph of finite growth constant.

Let $\Lambda = (V,E_V)$ be a quantum spin system with bounded, nearest-neighbour, commuting potential $\Phi:(i,j)\mapsto h_{i,j}$, such that the Hamiltonians on finite sub-lattices $\Gamma=(G,E_G)\subset\subset\Lambda$ can be written as
\begin{align}
    H_\Gamma=\sum_{(i,j)\in E_G} h_{i,j},
\end{align} where each term $h_{i,j}$ acts only non-trivially on vertices $i$ and $j$.
The Gibbs state of $H\equiv H_\Gamma$ with inverse temperature $\beta$ is
\begin{align}
    \sigma = \frac{e^{-\beta H}}{\Tr[e^{-\beta H}]} = \frac{\prod_{\{i,j\}\in E_V}e^{-\beta h_{i,j}}}{\Tr[\prod_{\{i,j\}\in E_V}e^{-\beta h_{i,j}}]}.
\end{align}
Given some $A\subset\subset\Lambda$ we will define a suitable *-algebra $\mathcal{N}_A$ and conditional expectation $E_A^S$ onto it, whose pre-adjoint has the Gibbs state as an invariant state. For simplicity of notations, we do this for a singleton $A=\{a\}$. However, this construction works similarly for all $A\subset\subset\Lambda$.
Given some $a\in\Lambda$, we enumerate the sets
\begin{align}
    \partial\{a\}&:=\{x\in\Lambda|\dist(x,a)=1\} = \{b_i\}_{i\in I_a}, \\
    \partial\{b_i\}&:=\{x\in\Lambda|\dist(x,b_i)=1\} = \{c_{i,j}\}_{j\in J^{(i)}},\text{ s.t. }a=c_{i,0} \ \forall i.
\end{align} 
Hence $\partial(\partial\{a\})\setminus\{a\}=\{y\in\Lambda| \dist(a,y)=2\}=\{c_{i,j}\}_{i\in I_a,j\in J^{(i)}\setminus\{0\}}$. See Figure \ref{fig:Schmidt} for a graphical example of these definitions. 

\begin{figure}[h]
    \centering
    \includegraphics[scale=0.5]{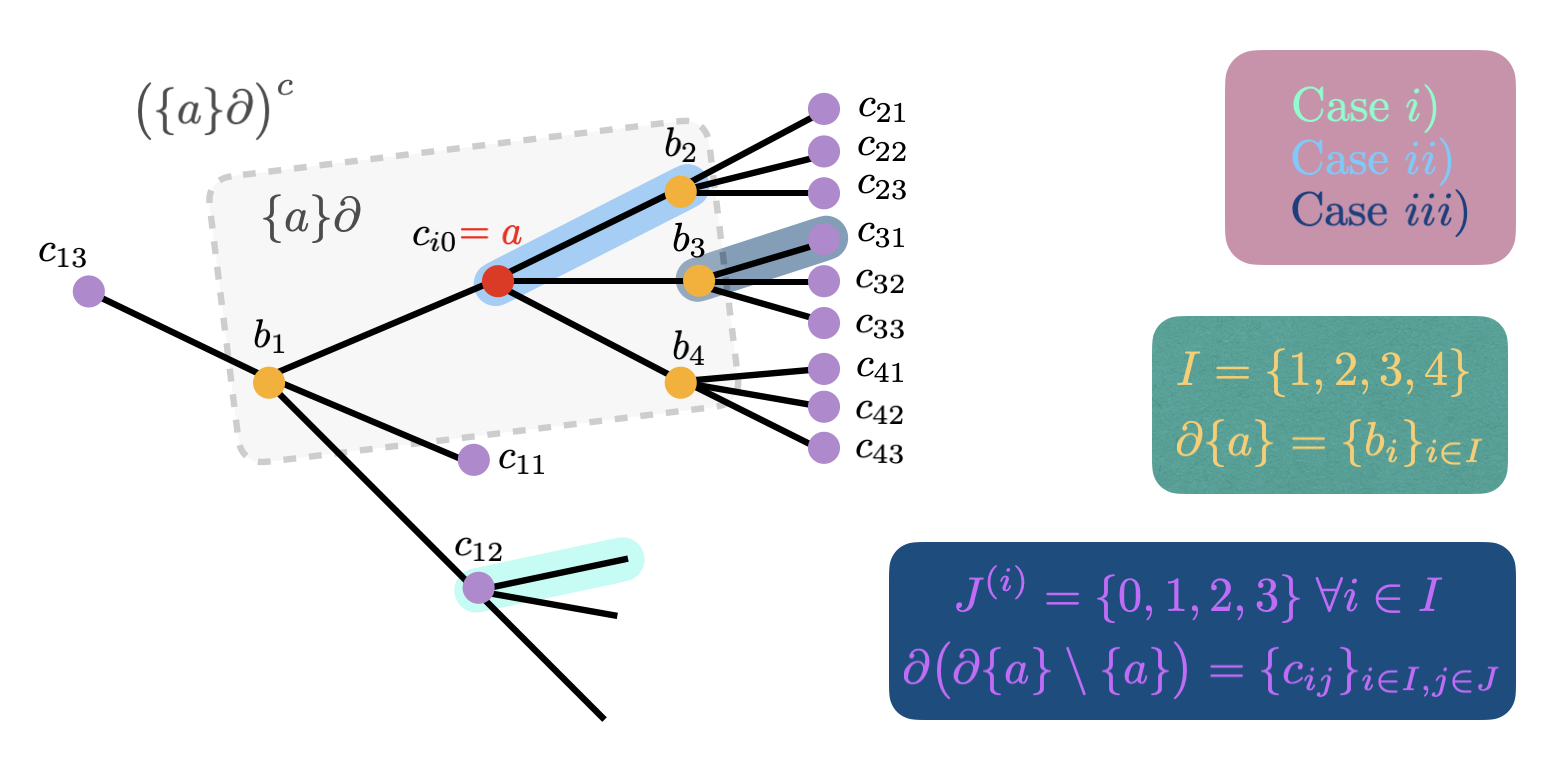}
    \caption{Simple example of the notation required for the local algebras $\mathcal{N}_{a}$ and thus the Schmidt conditional expectation $E^S_{a}$. Depicted is a small region of a 3-ary tree, with a (red) vertex labeled {\color{red} a}=$c_{i,0}$. Its neighbours (yellow) are labeled with {\color{RoyalYellow}$\{b_i\}_{i\in I}= \partial\{a\}$}, where $I\equiv I_a=\{1,2,3,4\}$. The next-nearest neighbours (purple) are  {\color{purplemod}$\partial(\partial\{a\})\setminus\{a\}=\{c_{i,j}\}_{i\in I,j\in J^{(i)}\setminus\{0\}}$}. The central vertex {\color{red}a} is logically the same as $c_{i,0}$ for each $i\in I$. 
    In the proof of proposition \ref{prop:ModularInvariance} the case i) considers for example the {\color{Turquise}turquiose shaded edge}, case ii) for example the {\color{SkyBlue}blue shaded edge} and case iii) for example the {\color{RoyalBlue}dark blue shaded edge}. The boundary between the subsets $\{a\}\partial$ and $(\{a\}\partial)^c$ is marked with a {\color{gray}gray dotted line}.}
    \label{fig:Schmidt}
\end{figure}

We will drop the index $a$ of $I\equiv I_a$, the labeling of all the neighbours of $a$ in the following. 
Now we Schmidt-decompose
\begin{align}
    e^{-\beta h_{b_ic_{ij}}}=\sum_{s}X^{j,s}_{b_i}\otimes X^s_{c_{ij}}
\end{align} for $i\in I$, where the operators $\left\{X^{j,s}_{b_i}\right\}_{j,s}\subset\mathcal{B}(\mathcal{H}_{b_i})$ and for $j\in J^{(i)}$, $\left\{X^{s}_{c_{ij}}\right\}_{s}\subset\mathcal{B}(\mathcal{H}_{c_{ij}})$.
Let us define the *-algebra $\mathscr{A}^j_{b_i}$ to be generated by all $\left\{X^{j,s}_{b_i}\right\}_{s}$, for all $i\in I$. \cite{art:bravyi2004commutative}.
\begin{proposition}
Any two non-identical of these algebras $\left\{\mathscr{A}^j_{b_i}\right\}_{i\in I,j\in J^{(i)}}$ commute.
\end{proposition}
\begin{proof}
Consider $\mathscr{A}^j_{b_i}$ and $\mathscr{A}^n_{b_m}$. If $i\neq m$, then the statement is obviously true, since their generators act on different Hilbert spaces $\mathcal{H}_{b_i}, \mathcal{H}_{b_m}$, respectively. If $i=m$, see that
\begin{align}
    &0= \left[e^{-\beta h_{b_ic_{ij}}},e^{-\beta h_{b_ic_{in}}} \right] = \left[\sum_s X_{b_i}^{j,s}\otimes X_{c_{ij}^s},\sum_r X_{b_i}^{n,r}\otimes X^r_{c_{in}}\right] \overset{j\neq n}{=} \sum_{s,r}X^s_{c_{ij}}\otimes X^r_{c_{in}} \otimes \left[X_{b_i}^{j,s} ,X_{b_i}^{n,r} \right] \\
    &\implies \left[X_{b_i}^{j,s} ,X_{b_i}^{n,r} \right]=0 \ \forall s,r ,
\end{align} where the last implication follows since $\{X^s_{c_{ij}}\}_s,\{X^r_{c_{in}}\}_r$ form a set of linear independent operators by Schmidt decomposition.
\end{proof}
Therefore these algebras and the underlying Hilbert spaces admit the following joint decomposition
\begin{align}
    \mathcal{H}_{b_i}:=\bigoplus_{\alpha_i}\bigotimes_{j\in J^{(i)}}\mathcal{H}_j^{\alpha_i}\otimes \mathcal{H}_c^{\alpha_i}=:\bigoplus_{\alpha_i}P^{\alpha_i}\mathcal{H}_{b_i},
\end{align} where $P^{\alpha_i}$ are orthogonal projectors such that
\begin{align}
    P^{\alpha_i}\mathcal{H}_{b_i}=\bigotimes_{j\in J^{(i)}}\mathcal{H}_j^{\alpha_i}\otimes \mathcal{H}_c^{\alpha_i},
\end{align} and $\{\mathcal{H}_j^{(\alpha_i)},\mathcal{H}_c^{(\alpha_i)}\}_{\alpha_i}$ are such that
\begin{align}
    \mathscr{A}^j_{b_i}=\bigoplus_{\alpha_i}\mathcal{B}(\mathcal{H}_j^{\alpha_i})\otimes \1_{\bigotimes_{k\in J^{(i)}\setminus\{j\}}\mathcal{H}_k^{\alpha_i}\otimes \mathcal{H}_c^{\alpha_i}} \hspace{1cm} \forall j\in J^{(i)}.
\end{align}
\begin{definition}
For any $\{a\}\in\Lambda$, define the *-subalgebra $\mathcal{N}_a:=\1_{\{a\}}\otimes \bigotimes_{i\in I}\bigotimes_{j\in J^{(i)}\setminus\{0\}}\mathscr{A}^j_{b_i}\otimes \mathcal{B}(\mathcal{H}_{(\{a\}\partial)})^c \subset \mathcal{B}(\mathcal{H}_\Lambda) $.
\end{definition}
\begin{proposition} \label{prop:ModularInvariance}
The modular group of the Gibbs state $\sigma$ leaves this algebra invariant, for any $a\in\Lambda$, i.e.
\begin{align}
    \Delta^{it}_{\sigma}(\mathcal{N}_a)\subset\mathcal{N}_a \hspace{1cm} \forall t\in\mathbb{R},
\end{align} where $\Delta_\sigma^{it}(X):=\sigma^{it}X\sigma^{-it}, \ \forall X\in\mathcal{B}(\mathcal{H}_\Lambda)$.
\end{proposition}
\begin{proof}
We can set $\beta=1$ without loss of generality. Let us fix $a\in\Lambda$. It is enough to show that $e^{ith_{kl}}\mathcal{N}_ae^{-ith_{kl}}\subset\mathcal{N}_a$ holds for any pair $(k,l)\in E_V$ (for examples of the following cases see also \Cref{fig:Schmidt}):

\vspace{0.2cm}

\noindent \underline{\textit{Case i)}} For $(k,l)\subset (\{a\}\partial)^c$, then this is obvious, since $\mathcal{N}_a\big|_{\mathcal{H}_{(\{a\}\partial)^c}}=\mathcal{B}(\mathcal{H}_{(\{a\}\partial)^c})$.

\vspace{0.2cm}

\noindent \underline{\textit{Case ii)}} For
$(k,l)\subset\{a\}\partial$, assume w.l.o.g. $k=a$, hence $l=b_i$ for some $i\in I$. Let $Y\in\mathcal{N}_a$ then via the spectral theorem $[e^{-ith_{ab_i}},Y]=0 \Leftrightarrow [e^{-h_{ab_i}},Y]=0$. It is enough to show this for the generators $Y=\1_{\{a\}}\otimes X^{j,s}_{b_i}$ for all $s$ and $j\in J^{(0)}\setminus\{0\}$, by closedness of the algebra.
\begin{align}
    \left[e^{-h_{ab_i}},\1_{\{a\}}\otimes X^{j,s}_{b_i}\right] = \left[\sum_r X_a^r\otimes X^{0,r}_{b_i},\1_{\{a\}}\otimes X^{j,s}_{b_i}\right] = \sum_r X_a^r \otimes \left[X_{b_i}^{0,r},X_{b_i}^{j,s}\right] = 0,
\end{align} $\forall s,j\in J^{(i)}\setminus\{0\}$ since the algebras $\mathscr{A}^0_{b_i}$ and $\mathscr{A}^j_{b_i}$ commute for $j\in J^{(i)}\setminus\{0\}$. 
\vspace{0.2cm}

\noindent \underline{\textit{Case iii)}} For $(k,l)\subset\partial\{a\}\cup(\{a\}\partial)^c$. W.l.o.g. $k=b_i$ for some $i\in I$, hence $l=c_{ij}$ for some $j\in J^{(i)}\setminus\{0\}$. Then $e^{-\beta h_{kl}}=e^{-\beta h_{b_ic_{ij}}}\otimes \1_{\mathcal{H}_{\{k,l\}^c}}=\sum_s X^{j,s}_{b_i}\otimes X_{c_{ij}}^s\otimes \1_{\mathcal{H}_{\{k,l\}^c}}\in\mathcal{N}_a$. Thus by the spectral theorem $e^{\pm ith_{kl}}\otimes \1 \in\mathcal{N}_a$ and hence by closedness of the algebra $e^{ith_{kl}}\mathcal{N}_ae^{-ith_{kl}}\subset\mathcal{N}_a$.
\end{proof}

Thus by Takesaki's theorem \cite{art:TAKESAKI1972306}, see also \cite[Proposition 10]{art:2localPaper} due to Proposition \ref{prop:ModularInvariance}, there exists a conditional expectation
\begin{align}
    E^S_{a} : \mathcal B(\mathcal{H}_\Lambda) \to \mathcal{N}_a,
\end{align} such that its pre-adjont has the Gibbs state $\sigma$, which is full-rank, as an invariant state, i.e. $E^S_{a*}(\sigma)=\sigma$.
This is called the \textit{Schmidt conditional expectation}.
\begin{remark}
The above construction works exactly the same for any subregion $A\subset\Lambda$ in place of $\{a\}\subset\Lambda$, yielding a *-subalgebra $\mathcal{N}_A$. Hence, we equally define the family of conditional expectations $\{E_A^S\}_{A\subset\Lambda}$ on $\mathcal{B}(\mathcal{H}_\Lambda)$. Similarly, for these it holds that $E^S_{A*}(\sigma)=\sigma$.
We can think of these conditional expectations as replacing any given observable on the local subset with the identity, in such a way that is consistent with the invariance of the Gibbs state under its pre-adjoint.
Hence, the family of Schmidt conditional expectations still has the desirable properties of the Davies expectations, i.e. that the Gibbs state is invariant, but their structure is easier to analyze since we can give an explicit expression for the conditional expectations which does not depend on system-environment couplings \cite{art:2localPaper}. 
\end{remark}
Before giving their explicit form, we highlight another important property of the Schmidt conditional expectations.
\begin{proposition}\label{prop:commutativityofSC}
For any two subsets $A_1,A_2\subset\Lambda$, such that $\dist(A_1,A_2)\geq 2$, the Schmidt conditional expectations $E^S_{A_1}$ and $E^S_{A_2} $ satisfy
\begin{align}
    E^S_{A_1}\circ E^S_{A_2} &=  E^S_{A_2}\circ E^S_{A_1} = E^S_{A_1\cup A_2}, \\
    E^S_{A_1}\circ E^S_{A_1 \cap A_2} &= E^S_{A_1 \cap A_2}\circ E^S_{A_1} = E^{S}_{A_1}.
\end{align}
\end{proposition}
This follows from the fact that the conditional expectation $E^S_A$ is a local map, acting only non-trivially on $A\partial$, and the following Lemma. 
\begin{lemma} \label{lem:SchmidAlgebra}
For any two subsets $A_1,A_2\subset\Lambda$, s.t. $\dist(A_1,A_2) \ge 2$, or such that one is a subset of the other it holds that 
\begin{itemize}
    \item[1)] $\mathcal{N}_{A_1}\cap\mathcal{N}_{A_1}$ = $\mathcal{N}_{A_1\cap A_2}$ .
    \item[2)] $\mathcal{N}_{A_1}\cup\mathcal{N}_{A_2} = \mathcal{N}_{A_1\cup A_2}$.
\end{itemize}
Here, $\mathcal{N}_{A_1}\cup\mathcal{N}_{A_2}$ denotes the *-algebra generated by $\mathcal{N}_{A_1}$ and $\mathcal{N}_{A_2}$. $\mathcal{N}_{A_1}\cap\mathcal{N}_{A_2}$ denotes the *-algebra generated by all elements in both $\mathcal{N}_{A_1}$ and $\mathcal{N}_{A_2}$.
\end{lemma}
\begin{proof}
The proof is elementary from the definition of the algebras $\mathcal{N}$ and may be found in Appendix \ref{sec:proof_Schmidt}.
\end{proof}

For a subset $A\subset\Lambda$, we call a set $(\alpha):=\{\alpha_i\}_{i\in I_A}$ a \textit{boundary condition}.

\begin{proposition}[Explicit Form of Schmidt conditional Expectation, \cite{art:2localPaper}]
For $A\subset\Lambda$, let $(\alpha):=\{\alpha_i\}_{i\in I_A}$ be a fixed boundary condition for the subset $A$ and denote $P^{(\alpha)}:=\bigotimes_{i\in I_A}P^{\alpha_i}$. We set
\begin{align}
\mathcal{H}^{(\alpha)}_{A_\inn}&:=\mathcal{H}_A\otimes\bigotimes_{i\in I_A}\mathcal{H}_0^{\alpha_i}\otimes \mathcal{H}_c^{\alpha_i} \equiv \mathcal{H}_A\otimes \mathcal{H}^{(\alpha)}_{\partial_\inn A}, \\
    \mathcal{H}^{(\alpha)}_{A_\out}&:= \mathcal{H}_{(A\partial)^c}\otimes\bigotimes_{i\in I_A}\bigotimes_{j\in J^{(i)}\setminus\{0\}}\mathcal{H}_j^{\alpha_i}\equiv \mathcal{H}_{(A\partial)^c}\otimes\mathcal{H}^{(\alpha)}_{\partial_\out A},
\end{align}
i.e. $\mathcal{H}_{\Lambda}=\bigoplus_{(\alpha)}{\mathcal{H}^{(\alpha)}_{A_\inn}}\otimes{\mathcal{H}^{(\alpha)}_{A_\out}} $, and write ${\tr_{\mathcal{H}_{A^{(\alpha)}_\inn}}} \equiv \tr_{A_\inn,(\alpha)}$, respectively, $\tr_{\mathcal{H}_{A_\out^{(\alpha)}}}=\tr_{A_\out,(\alpha)}$ for simplicity. 

Since every element of the algebra $\mathcal{N}_A$ is block diagonal w.r.t the sets $(\alpha)$, we can decompose the Schmidt conditional expectation $E^S_A$ and its pre-adjoint $E^S_{A*}$ along those blocks as well, yielding
\begin{align}
    E_A^S(X)&:=\bigoplus_{(\alpha)} E^{S,(\alpha)}_A(X), \\
    E_{A*}^S(\rho)&:=\bigoplus_{(\alpha)} E^{S,(\alpha)}_{A*}(\rho),
\end{align}
for any $X\in \mathcal{B}(\mathcal{H}_\Lambda)$ and $\rho\in\mathcal{D}(\mathcal{H}_\Lambda)$.  $E^{S,(\alpha)}_A$ and $E^{S,(\alpha)}_A$ have the following expressions \cite{art:2localPaper}, defined on the block Hilbert space $P^{(\alpha)}\mathcal{H}_{\Lambda}$
\begin{align}
    E^{S,(\alpha)}_A(X) &= P^{(\alpha)}\left({\tr_{A_\inn,(\alpha)}}[\tau_{A_\inn}^{(\alpha)}X]\otimes \1_{A_{\inn,(\alpha)}}\right)P^{(\alpha)}, \label{def:SchmidtCondExp1} \\
    E^{S,(\alpha)}_{A*}(\rho) &= {\tr_{A_\inn,(\alpha)}} [P^{(\alpha)}\rho P^{(\alpha)}]\otimes\tau^{(\alpha)}_{A_\inn}, \label{def:SchmidtCondExpDual}
\end{align} where the state $\tau_{A_\inn}^{(\alpha)}$ is given by
\begin{align} \label{def:tau}
    \tau^{(\alpha)}_{A_\inn} := \frac{1}{\Tr[...]}\tr_{A_\out,(\alpha)}[P^{(\alpha)}\sigma^{(A\partial)} P^{(\alpha)}] = \frac{1}{\Tr[...]}\tr_{\partial_\out A,(\alpha)}[P^{(\alpha)}\sigma^{(A\partial)} P^{(\alpha)}],
\end{align} where the prefactors $\frac{1}{\Tr[...]}$ contain the trace on $\mathcal{B}(\mathcal{H}^{(\alpha)}_{A_\inn})$ and ensure proper trace normalization $\Tr[\tau_{A_\inn}^{(\alpha)}]=1$, and the partial trace in the last expression traces out the Hilbert space $\mathcal{H}^{(\alpha)}_{\partial_\out A}=\bigotimes_{i\in I_A}\bigotimes_{j\in J^{(i)}\setminus\{0\}}\mathcal{H}_j^{\alpha_i} = \mathcal{H}^{(\alpha)}_{A_\out\setminus(A\partial)^c,}$.
\end{proposition}
It is easy to check that the expressions \eqref{def:SchmidtCondExp1} and \eqref{def:SchmidtCondExpDual} are dual to each other w.r.t. the Hilbert-Schmidt inner product on $\mathcal{B}(\mathcal{H}_\Lambda)$. 
The expression \eqref{def:tau} follows from the invariance of the local Gibbs states $\sigma^{(A\partial)}$ via $P^{(\alpha)}\sigma^{(A\partial)}P^{(\alpha)}=E^{S(\alpha)}_{A*}\left(P^{(\alpha)}\sigma^{(A\partial)}P^{(\alpha)}\right)=\tr_{A_\inn,(\alpha)}\left[P^{(\alpha)}\sigma^{(A\partial)}P^{(\alpha)}\right]\otimes\tau_{A_\inn}^{(\alpha)}$. Taking the partial trace of this expression on $\mathcal{H}_{A_\out}^{(\alpha)}$ gives the above expression.


\begin{remark}
One can think of the $\alpha_i$ as labeling the boundary conditions at site $i\in I_A$, and all $(\alpha):=\{\alpha_i\}_{i\in I_A}$ give a complete labeling of all boundary conditions of some subset $A$.
Hence, one can think of the effect of the Schmidt-conditional expectation on states as effectively replacing the state $\rho$ locally on $A$ with the Gibbs state $\sigma$, where the boundary conditions are set to $(\alpha)=\{\alpha_i\}_{i\in I_A}$.
\end{remark}
As in the Davies evolution, there exists a uniform family of Lindbladians for which the Schmidt conditional expectations are given by the local Lindbladian projectors.
For the Schmidt conditional expectation, the corresponding family of Lindbladians is
\begin{align}
    \Li_A^S(X):=\sum_{x\in A}E^S_{x}(X)-X.
\end{align}  We call them the (family of) \textit{Schmidt generators associated to $(\Lambda,\Phi,\beta)$}\footnote{Explicitly, when $\{\sigma^\Gamma\}_{\Gamma\subset\subset\Lambda}$ is the family of Gibbs states associated to $(\Lambda,\Phi,\beta)$ is the one used to construct these Schmidt conditional expectations. 
} 
\cite{art:2localPaper}. It is straightforward to see that the projection onto their kernel is given by the Schmidt conditional expectation. From the properties established above, they are uniform families of locally primitive, locally GNS-symmetric, frustration-free Lindbladians. Through this definition we immediately obtain additivity in the region. That is, $\mathcal{L}^S_A+\mathcal{L}^S_B=\mathcal{L}^S_{AB}$ for non-overlapping regions $A,B$. 

\subsection{Relating Davies and Schmidt dynamics}
An important observation is that, since the Schmidt and Davies families of conditional expectations almost have the same fixed-point algebras, we can relate the relative entropy distance of any given state to the fixed point subalgebra of one to the other: 
\begin{lemma}\label{lem:RelEntCondExpBound}
Let $X\subset\Lambda$, $\rho\in\mathcal{D}(\mathcal{H}_{\Lambda})$. For $E^S_{X*}(\rho)$ the Schmidt conditional expectation of $\rho$ and $E^D_{X*}(\rho)$ its Davies conditional expectation, it holds that
\begin{align} \label{lemma:SchmidtandDavies}
    D(\rho\|E^D_{X*}(\rho)) \leq D(\rho\|E^S_{X*}(\rho)) \leq D(\rho\|E^D_{X\partial*}(\rho)).
\end{align}
\end{lemma}

\begin{proof}
First recall that for some region $X\subset\Lambda$, $E_X^D$ is the projection onto the kernel of $\mathcal{L}^D_{X}$, and call it $\mathcal{F}^D_X:=$Fix($E^D_X$). It is also the projection onto the largest *-subalgebra of $\mathcal{B}(\mathcal{H}_{X^c})\otimes\1_X\subset\mathcal{B}(\mathcal{H}_\Lambda)$ which is invariant under the modular group of the Gibbs state $\sigma$, $\{\Delta^{it}_\sigma\}_{t\in\mathbb{R}}$ \cite{art:2localPaper,art:ApproxTensorizationBardet_2021}.
Now, $E^S_X$ is a projection onto, say $\mathcal{F}^S_X:=$Fix($E^S_X$). This is by construction a *-subalgebra of $\mathcal{B}(\mathcal{H}_{X^c})\otimes\1_X\subset\mathcal{B}(\mathcal{H}_\Lambda)$ invariant under the modular group $\{\Delta^{it}_\sigma\}_{t\in\mathbb{R}}$, see \Cref{subsec:SchmidtCondExp}. Thus $\mathcal{F}_X^S\subset\mathcal{F}_X^D$. This implies that
$D(\rho\|E^D_{X*}(\rho)) \leq D(\rho\|E^S_{X*}(\rho))$ for any state $\rho\in\mathcal{D}(\mathcal{H}_\Lambda)$, since
\begin{align}
    D(\rho\|E^S_{X*}(\rho))-D(\rho\|E^D_{X*}(\rho)) &= \Tr\big[\rho\underbrace{(\log E^D_{X*}(\rho)-\log E^S_{X*}(\rho))}_{\in\mathcal{F}^D_X}\big] \\ &= \Tr[E^D_{X*}(\rho)(\log E^D_{X*}(\rho)-\log E^S_{X*}(\rho))] \\&= D(E^D_{X*}(\rho)\|E^S_{X*}(\rho)) \geq 0.
\end{align}
Here we used that if $\omega_1,\omega_2$ are fixed points of some conditional expectation, then so is $\log\omega_1-\log\omega_2$.
Since $\mathcal{F}_{X\partial}^D\subset \1_{X\partial}\otimes B(\mathcal{H}_{(X\partial)^c})\subset\mathcal{F}_X^S$ holds by construction of the Schmidt conditional expectation, it equally follows from the calculation above that $D(\rho\|E^S_{X*}(\rho)) \leq D(\rho\|E^D_{X\partial*}(\rho))$. 
\end{proof}

This is a crucial lemma that allows us to analyse the MLSI for Davies generators in terms of entropic inequalities associated to Schmidt generators, which are easier to analyse. The Schmidt generators hence serves as a proxy QMS to the Davies. Such a comparison is a well known technique for classical Markov chains \cite{Martinelli1999}.

\section{Static properties: Clustering of correlations on Gibbs states}\label{sec:static}

In this section we study the static properties of quantum spin systems. These are used to measure how correlations decay between spatially separated regions on a Gibbs state.

We review various types of spatial clustering and spatial mixing properties ranging from the more commonly considered notions, such as exponential decay of covariance, to stronger ones, such as exponential decay of mutual information. Clustering properties and their refinement ``mixing conditions" are central to the geometric divide-and-conquer arguments that establish rapid mixing for spin systems, both in the quantum \cite{art:2localPaper, art:EntropyDecayOf1DSpinChain-Cambyse, art:SupperaditivityofQuentRelEntCapel_2018} and the classical setting  \cite{art:ClassicaTreesMartinelli_2004}. While the various notions are useful at different steps along the proofs, we will also show equivalence relations between them in the context of Gibbs states of commuting Hamiltonians. For our current purposes these will, however, only be interchangeable on low dimensional lattices, e.g. 1D and 2D hypercubic ones.
In the rest of this section we will assume that $\Lambda$ is some fixed graph with finite growth constant.

\subsection{Decay of correlations} 

The most commonly used quantifier of correlations on Gibbs states is the \textit{covariance}. Given two operators $X,Y \in \BH$ and a full rank state $\sigma \in \mathcal{D}(\mathcal{H})$, it is defined as
\begin{equation}\label{eq:covariance}
    \Cov_{\sigma}^{\KMS,\GNS}(X,Y) :=\langle X-\Tr[\sigma X]\1,Y-\Tr[\sigma Y]\1\rangle_\sigma^{\KMS,\GNS} \, ,
\end{equation}
where the inner product considered is either the KMS (see \eqref{eq:KMSprod}) or the GNS (see \eqref{eq:GNSprod}) one. When $X=Y$, the covariance reduces to the \textit{variance}
\begin{equation}\label{eq:variance}
     \Var^{\KMS,\GNS}_\sigma(X) := \Cov^{\KMS,\GNS}_\sigma(X,X) \, .
\end{equation}
This allows us to introduce the first notion of clustering, the \textit{exponential decay of covariance}, or simply \textit{exponential decay of correlations} (termed $\mathbb{L}_\infty$-clustering in \cite{art:2localPaper}). Since all conditions of clustering of correlations considered in this paper exhibit exponential decay, we will omit the reference to this on the names of the properties hereafter. We will define \textit{uniform} versions of the following notions of clustering. This guarantees that we have clustering independently of the sequence of sub-lattices $\Gamma\subset\subset\Lambda$, since the function $l\mapsto\epsilon_{X,Y}(l)$ and its decay rate do not depend on $\Gamma$. It may, however, still depend on the regions of $X$ or $Y$ or their boundaries. Instead, a non-uniform notion of clustering would be one in which we are only guaranteed the existence of a function $\epsilon_{\Gamma}(l)$ for each $\sigma^\Gamma$, with an explicit dependence on $\Gamma$. 

\begin{definition}[Uniform Decay of covariance] \label{def:LinfinityClustering} 
Let $\Phi$ be a geometrically local potential on $\Lambda$ and consider an inverse temperature $\beta>0$. We say that the family of Gibbs states $\{\sigma^\Gamma\}_{\Gamma\subset\subset\Lambda}$ associated to $(\Lambda,\Phi,\beta)$ 
satisfies \textit{uniform exponential decay of covariance} if there exist a function $(X,Y,l)\mapsto\epsilon_{X,Y}(l)$, exponentially decaying in $l$, with decay rate independent of $X,Y$ such that for any subregion $\Gamma\subset\subset\Lambda$ and any $A,C\subset\Gamma$ with $\dist(A,C)=l$, it holds that 
\begin{align}
    \Cov_{\sigma^\Gamma}^{\GNS}(f,g)\leq \|f\|\|g\| \epsilon_{A,C}(l),
\end{align} for any self-adjoint $f,g\in\mathcal{B}(\mathcal{H}_\Gamma)$ with support on $A,C$, respectively. Uniformity here refers to the fact, that this function does not explicitly depend on $\Gamma$. 
\end{definition}

This condition is also frequently rewritten in terms of 
\begin{align}
    \Cov_{\sigma^\Gamma}^{\GNS}(A,C):= \underset{\substack{
f \in \mathcal{B}(\mathcal{H}_A), \norm{f} \leq 1\\
g \in \mathcal{B}(\mathcal{H}_C), \norm{g} \leq 1}}{\sup}  \Cov_{\sigma^\Gamma}^{\GNS}(f,g), \, 
\end{align} 
so that $ \Cov_{\sigma^\Gamma}^{\GNS}(A,C)\le \epsilon_{A,C}(l)$. We will use this notation when we want to emphasise the systems between which we are studying correlations. Moreover, since we will mostly focus on the GNS covariance in the rest of the text, we will drop the superscript whenever this is the case. 

\begin{remark}
    The decay length of the function $\epsilon_{A,C}(l)$ is called the ($\mathbb{L}_\infty$)-\emph{correlation length} $\xi$, which is the standard decay rate of thermal two-point correlation functions, i.e. $-\log\epsilon_{A,C}(l)=\mathcal{O}\left(\frac{l}{\xi}\right)_{l\to\infty}$. It is by the uniformity assumption independent of $A,C$.
\end{remark}

The decay of covariance is a rather weak notion of clustering, hence sometimes referred to as \emph{weak clustering}. In 1D, it was proven to hold in infinite spin chains, for translation invariant interactions, first in the finite-range regime \cite{Araki1969} at every inverse temperature $\beta >0$, and subsequently at high enough temperature for exponentially-decaying interactions in \cite{PP2023} (with a critical temperature tending to zero when approaching the finite-range case). This was extended to the finite chain regime in \cite{art:Bluhm2022exponentialdecayof} and \cite{CMTW2023LocalityGibbs} for finite-range and exponentially-decaying interactions, respectively. More recently, \cite{kimura2024clustering} proved this decay for commuting Hamiltonians without the translation-invariant assumption, as well as slower decays in other regimes. The high dimensional result appears frequently in the literature, with proofs for the finite-range case in e.g. \cite{art:LocalityofTemperature} and for exponentially-decaying interactions in \cite{FU2015}. Additionally, in relation to dynamical properties of QMS, the exponential decay of correlations is known to hold for steady states of rapidly mixing QMS \cite{art:QuantumGibbsSamplers-kastoryano2016quantum, art:StabilityCubitt_2015, art:RapidMixingandDecayofCorrelationsKastoryano_2013}.

For completeness, we also introduce the very related property of $\mathbb{L}_2$-clustering.  

\begin{definition}[Uniform $\mathbb{L}_2$-clustering] \label{def:L2Clustering}
Let $\Phi$ be a geometrically local potential on $\Lambda$ and consider an inverse temperature $\beta>0$. We say that the family of Gibbs states $\{\sigma^\Gamma\}_{\Gamma\subset\subset\Lambda}$ associated to $(\Lambda,\Phi,\beta)$ 
satisfies \textit{uniform} $\mathbb{L}_2$-\textit{clustering} if there exists a function $(X,Y,l)\mapsto\epsilon_{X,Y}(l)$, exponentially decaying in $l$, with decay rate independent of $X,Y$ such that for any subregion $\Gamma\subset\subset\Lambda$ and any $A,C\subset\Gamma$ with $\dist(A,C)=l$, it holds that
\begin{align}
    \Cov_{\sigma^\Gamma}^{\GNS}(f,g)\leq \|f\|_{2,\sigma^\Gamma} \|g\|_{2,\sigma^\Gamma} \epsilon_{A,C}(l),
\end{align} for any self-adjoint $f,g\in\mathcal{B}(\mathcal{H}_\Gamma)$ with support on $A,C$, respectively. 
\end{definition}

This property immediately implies the previous exponential decay of correlations by the monotonicity of $\mathbb{L}_{p,\sigma}$-norms in $p$. Moreover, the $\mathbb{L}_2$-clustering is important in gapped primitive QMS.  In \cite[Corollary 27]{art:QuantumGibbsSamplers-kastoryano2016quantum}, it was proven that steady states of gapped primitive QMS satisfy $\mathbb{L}_2$-clustering, using the detectability lemma \cite{art:detectibilityLemmaaharonov2011quantum}. 

We now discuss the notion of local indistinguishability \cite{art:FiniteCorrelationLengthImpliesEfficient, art:LocalityofTemperature}, which quantifies the influence of the state at spatially separated regions in that of a fixed subregion $A$.

\begin{definition}[Uniform Local indistinguishability] \label{def:local_indistinguishability}
Let $\Phi$ be a geometrically local potential on $\Lambda$ and consider an inverse temperature $\beta>0$. We say that the family of Gibbs states $\{\sigma^\Gamma\}_{\Gamma\subset\subset\Lambda}$ associated to $(\Lambda,\Phi,\beta)$ 
satisfies \textit{uniform local indistinguishability} \cite{art:FiniteCorrelationLengthImpliesEfficient} 
if there exists a function $(X,Y,l)\mapsto\epsilon_{X,Y}(l)$, exponentially decaying in $l$, with decay rate independent of $X,Y$ such that for any subregion $\Gamma\subset\subset\Lambda$ and any $A,C\subset\Gamma$ with $\dist(A,C)=l$, it holds that
\begin{align}
    \norm{\tr_{BC}[\sigma^{ABC}] -\tr_{B}[\sigma^{AB}] }_1 \leq  \epsilon_{A,C}(l) \, .
\end{align} 

\end{definition}

In \cite{art:FiniteCorrelationLengthImpliesEfficient} it was shown that Gibbs states of geometrically-local, bounded, possibly non-commuting Hamiltonians, which satisfy uniform decay of covariance, also satisfy uniform local indistinguishability (see also \cite{onorati2023efficient}), where the exponentially decaying function in the latter $l\mapsto \epsilon_{A,C}(l)$ has an additional factor of $|\partial C|$, the boundary of the region that is `cut' away, compared to the exponential decay function of the former.
There it was also shown that if the Hamiltonian is commuting then the decay length of the function in the local indistinguishability can be controlled by the thermal correlation length $\xi$.
For a more recent extension of the latter to exponentially-decaying interactions, and an overview on the relation between these properties and the locality of the Hamiltonian, we refer the reader to \cite{CMTW2023LocalityGibbs}.

We now also introduce a stronger version of local indistinguishability, which will be crucial for our proofs later on. 

\begin{definition}[Uniform strong local indistinguishability] \label{def:strong_local_indistinguishability} 
Let $\Phi$ be a geometrically local potential on $\Lambda$ and consider an inverse temperature $\beta>0$. We say that the family of Gibbs states $\{\sigma^\Gamma\}_{\Gamma\subset\subset\Lambda}$ associated to $(\Lambda,\Phi,\beta)$ 
satisfies \textit{uniform strong local indistinguishability} if there exists a function $(X,Y,l)\mapsto\epsilon_{X,Y}(l)$, exponentially decaying in $l$, with decay rate independent of $X,Y$ such that for any subregion $\Gamma\subset\subset\Lambda$ and any $A,C\subset\Gamma$ with $\dist(A,C)=l$, it holds that
\begin{align}
    \norm{(\tr_{B}[\sigma^{AB}])^{-\frac{1}{2}}\tr_{BC}[\sigma^{ABC}](\tr_{B}[\sigma^{AB}])^{-\frac{1}{2}} - \1_A}_\infty \leq  \epsilon_{A,C}(l) \, .
\end{align}  
\end{definition}

We will show that the family of Gibbs states associated to $(\Lambda,\Phi,\beta)$ satisfies uniform strong local indistinguishability with an exponential dependence on the boundaries if it satisfies uniform decay of covariance whenever $\Lambda$ has finite growth constant, $\Phi$ is a bounded, geometrically-local, commuting potential. The requirement for the commuting property may be dropped in 1D systems (see Theorem \ref{thm:StrongLocalIndist1D}). For commuting Hamiltonians, we also give quantitative results in terms of the correlation length $\xi$ and inverse temperature $\beta$ in \Cref{sec:relation_correlations}.

\begin{definition}[Uniform mixing condition] \label{def:mixing_condition}
Let $\Phi$ be a geometrically local potential on $\Lambda$ and consider an inverse temperature $\beta>0$. We say that the family of Gibbs states $\{\sigma^\Gamma\}_{\Gamma\subset\subset\Lambda}$ associated to $(\Lambda,\Phi,\beta)$ 
satisfies the \textit{uniform mixing condition} 
if there exists a function $(X,Y,l)\mapsto\epsilon_{X,Y}(l)$, exponentially decaying in $l$, with decay rate independent of $X,Y$ such that for any subregion $\Gamma\subset\subset\Lambda$ and any $A,C\subset\Gamma$ with $\dist(A,C)=l$, it holds that
\begin{align}
    \norm{\sigma_{AC}(\sigma_A \otimes \sigma_C)^{-1} - \1_{AC}}_\infty \leq  \epsilon_{A,C}(l) \, .
\end{align} 
\end{definition}

For Gibbs states of geometrically-local, possibly non-commuting, bounded, and translation invariant Hamiltonians on a 1D lattice this condition was shown to follow qualitatively from uniform exponential decay of covariance in \cite[Proposition 8.1]{art:Bluhm2022exponentialdecayof} at any inverse temperature $\beta >0$, and subsequently in \cite{art:Bluhm2023decayofcorrelations} for arbitrary dimensions at high enough temperature. In \cite{art:ImplicationsAndRapidTermalization-Cambyse,art:EntropyDecayOf1DSpinChain-Cambyse}, this was crucial to establish the existence of a log-decreasing MLSI constant $\alpha=\Omega(\log|\Gamma|)^{-1}$ for commuting quantum spin chain systems . 

The last quantity to measure correlations explored in this section is the \emph{mutual information}. We recall that, given a bipartite space $\HH_A \otimes \HH_B$ and $\rho_{AB}$ a density matrix on it, it is defined as 
\begin{equation}
    I_\rho(A:B) := \Tr[\rho_{AB}(\log \rho_{AB} - \log \rho_A \otimes \rho_B)] = D(\rho_{AB}\|\rho_A\otimes\rho_B) \, .
\end{equation}

\begin{definition}[Uniform decay of mutual information] \label{def:mutual_information}
Let $\Phi$ be a geometrically local potential on $\Lambda$ and consider an inverse temperature $\beta>0$. We say that the family of Gibbs states $\{\sigma^\Gamma\}_{\Gamma\subset\subset\Lambda}$ associated to $(\Lambda,\Phi,\beta)$ 
satisfies \textit{uniform exponential decay of mutual information} 
if there exists a function $(X,Y,l)\mapsto\epsilon_{X,Y}(l)$, exponentially decaying in $l$, with decay rate independent of $X,Y$ such that for any subregion $\Gamma\subset\subset\Lambda$ and any $A,C\subset\Gamma$ with $\dist(A,C)=l$, it holds that
\begin{align}
    I_{\sigma^\Gamma}(A:C) \leq  \epsilon_{A,C}(l), 
\end{align} holds. 
\end{definition}

By a standard use of Pinsker's and Hölder's inequalities, the decay of the mutual information directly implies decay of the covariance in \Cref{def:LinfinityClustering} with the decay rate halved, since $$I_\sigma(A:C)=D(\sigma_{AC}\|\sigma_A\otimes\sigma_C)\geq\frac{1}{2} \|\sigma_{AC}-\sigma_A\otimes\sigma_C\|^2_1.$$ 
Additionally, by \cite[Lemma 3.1]{art:Bluhm2022exponentialdecayof}, the following inequality holds
\begin{align}\label{eq:ineq_mutualinfo_mixingcondition}
I_\sigma(A:C)\leq \|\sigma_{AC}(\sigma_A\otimes\sigma_C)^{-1}-\1_{AC}\|_\infty \, .
\end{align}
so that the mixing condition from \Cref{def:mixing_condition} also implies the usual decay of covariance ($\mathbb{L}_\infty$-clustering).  
We show in the next subsections that these implications can be reversed in the case of commuting, finite-range Hamiltonians. This allows us to conclude that to establish the mixing condition, or the mutual information decay for the above considered systems, it is enough to establish the decay of the covariance.


\subsection{A useful relation}

To simplify the notation in the rest of this section, we first define the following relation.

\begin{definition}[A strong similarity relation]\label{def:Relation}
Given a finite lattice $\Gamma$ and two states $\omega,\tau \in \mathcal{D}(\mathcal{H}_\Gamma)$ with the same support $\supp(\omega)=\supp(\tau)$, the relation $\sim$ is
\begin{equation}
    \omega \overset{\epsilon}{\sim} \tau :\Longleftrightarrow \|\omega^{\frac{1}{2}}\tau^{-1}\omega^{\frac{1}{2}}-\1\|\leq\epsilon<1,
\end{equation}
where the identity $\1\equiv \1_{\supp(\omega)}=\1_{\supp(\tau)}$ is on the support of the states. The inverse represents the generalized inverse, i.e. the inverse on the support times the projection onto it.
\end{definition} 
\begin{proposition}\label{prop:RelationDmax} 
    When restricted to states, it holds that the notion of similarity induced by this relation is equivalent to the one induced by the max-relative divergence. I.e. for $\omega,\tau\in\mathcal{D}(\mathcal{H})$ such that $\omega\overset{\epsilon}{\sim}\tau$ it holds that
    \begin{align}
        D_\text{max}(\omega\|\tau)\leq \log(1+\epsilon) \leq \epsilon = \|\omega^{\frac{1}{2}}\tau^{-1}\omega^{\frac{1}{2}}-\1\| \leq D_\text{max}(\omega\|\tau)\exp(D_\text{max}(\omega\|\tau)). 
    \end{align}
\end{proposition}
A proof of this can be found in Appendix \ref{app:proofs_relation}.
\begin{remark}
It directly follows that $\omega \overset{\epsilon}{\sim} \tau \implies D(\omega\|\tau)\leq \epsilon$ and by Hölder's inequality it follows that $\omega\overset{\epsilon}{\sim}\tau \implies \|\omega-\tau\|_1\leq \epsilon$, but the converse is in general not true. Hence, this relation quantifies a stronger form of similarity between a pair of states than closeness in relative entropy and $1$-norm. 
Despite the equivalence to the max-relative entropy, when restricting to states, it will often be simpler to work directly with the relation instead of the max-relative entropy, and thereby avoid a logarithm. 
\end{remark}

We will often have $\epsilon$ be some exponentially decaying function depending on the supports of $\omega,\tau$, as the clustering functions discussed above. Thus, by a slight abuse of notation, we may write $\omega \sim \tau$ when meaning that there exists \emph{some} exponentially decaying function $\epsilon(l)$ with $\omega \overset{\epsilon}{\sim} \tau$, where $l$ is a suitable parameter (most often a distance between regions).
The mathematical terminology \textit{relation} is justified as per the following proposition. 

\begin{proposition}[Properties of $\overset{\epsilon}{\sim}$]\label{prop:properties_relation}
Given a bipartite system $\mathcal{K}=\mathcal{H}\otimes \mathcal{H}^\prime$, we consider positive\footnote{In fact self-adjoint would be sufficient, but for simplicity we restrict ourselves to positive operators here.} $A,\tilde{A},B,C\in \mathcal{B}(\mathcal{H})$, $D,E\in\mathcal{B}(\mathcal{K})$, and $F,\tilde{F}\in \mathcal{B}(\mathcal{H}^\prime)$ and $P\in\mathcal{B}(\mathcal{H}^\prime)$ a projection. 
The above defined relation $\overset{\epsilon}{\sim}$ satisfies the following properties:
\begin{enumerate}
    \item[$1)$]  \textbf{Reflexivity:} $A\overset{0}{\sim} A$.  
    \item[$2)$] \textbf{Symmetry:}  $ A\overset{\epsilon}{\sim} B \implies B\overset{\epsilon(1-\epsilon)^{-1}}{\sim} A$, $B^{-1}\overset{\epsilon}{\sim} A^{-1}$. 
    \item[$3)$] \textbf{Transitivity:} $A\overset{\epsilon_1}{\sim} B, B\overset{\epsilon_2}{\sim} C \implies A\overset{\eta}{\sim} C $, where $\eta = \epsilon_1\epsilon_2+\epsilon_1+\epsilon_2$.
    \item[$4)$] \textbf{Tensor multiplicativity:}  $A\overset{\epsilon_1}{\sim} \Tilde{A}, F\overset{\epsilon_2}{\sim} \Tilde{F} \implies A\otimes F \overset{\eta}{\sim} \Tilde{A}\otimes\Tilde{F} $, where $\eta =\epsilon_1\epsilon_2+\epsilon_1+\epsilon_2$.
    \item[$5)$] \textbf{Locality preservation:}  $ D\overset{\epsilon}{\sim} E \implies \tr_{\mathcal{H}^\prime}(\1\otimes P)D(\1\otimes P) \overset{\epsilon}{\sim} \tr_{\mathcal{H}^\prime}(\1\otimes P)E(\1\otimes P)$. 
    \item[$6)$] \textbf{Normalization preservation:} $ D\overset{\epsilon}{\sim} E \implies \frac{\tr_{\mathcal{H}^\prime}(\1\otimes P)D(\1\otimes P)}{\Tr[(\1\otimes P)D(\1\otimes P)]} \overset{\epsilon(2+\epsilon)}{\sim} \frac{\tr_{\mathcal{H}^\prime}(\1\otimes P)E(\1\otimes P)}{\Tr[(\1\otimes P)E(\1\otimes P)]}$. 
\end{enumerate}
    \label{prop:UsefulRelation}
\end{proposition}
For notational simplicity, we may write $A\overset{\epsilon_1}{\sim}B\overset{\epsilon_2}{\sim}C$, implying transitivity, when we mean $A\overset{\epsilon_1}{\sim}B$, $B\overset{\epsilon_2}{\sim}C$. The following corollary can be derived as a consequence of the previous proposition.  

\begin{corollary} \label{cor:Relation}
 If $A_i \overset{\epsilon}{\sim} A_{i+1}$ for $i= 0,...,K-1$, then $A_0\overset{\eta}{\sim}A_K \text{ with } \eta = (1+\epsilon)^K-1$.
\end{corollary}
The above corollary plays an important role in the estimation of the mixing condition between separate regions composed of several connected components.
The proofs of Proposition \ref{prop:properties_relation} and Corollary \ref{cor:Relation} are elementary, but we include them for completeness in \Cref{app:proofs_relation}.

\begin{remark}
    Using the relation introduced in this subsection, we can rewrite the strong local indistinguishability of a state $\sigma^{ABC}$ as 
    \begin{align}
\tr_{BC}[\sigma^{ABC}] \overset{\epsilon_{A,C}(l)}{\sim} \tr_{B}[\sigma^{AB}] \Longleftrightarrow D_\textup{max}(\tr_{BC}[\sigma^{ABC}]\|\tr_{B}[\sigma^{AB}])\lesssim \epsilon_{A,C}(l)  \, ,
\label{def:StrongLocalIndisting}    
\end{align}
while the mixing condition between $A$ and $C$ can be expressed as 
\begin{align}
 \sigma_{AC} \overset{\epsilon_{A,C}(l)}{\sim} \sigma_A\otimes\sigma_C \Longleftrightarrow I_\textup{max}(A:C)\lesssim \epsilon_{A,C}(l) \, .
\label{def:StrongLocalIndisting}    
\end{align} 
\end{remark}

\subsection{Relation between measures of correlations}\label{sec:relation_correlations}

We are now in a position to prove the main result of this section: the equivalence between the seemingly weaker notion of decay of correlations provided by the covariance and stronger notions such as strong local indistinguishability and the mixing condition. Note that this result will only hold in full generality for Gibbs states of commuting Hamiltonians, and that the equivalence is up to exponential pre-factors that grow with the boundaries of the relevant regions. As such, the results are more useful for low-dimensional systems, such as in 1D and 2D, where these boundaries are not too large compared to the distance between regions.

\begin{theorem}[Implications of decay of covariance; commuting case] \label{thm:ddimStrongLocalIndist}
Let $\Phi$ be a bounded, geometrically$-r-$local, commuting potential and $\beta>0$. Then if the family of Gibbs states $\{\sigma^\Gamma\}_{\Gamma\subset\subset\Lambda}$ associated to $(\Lambda,\Phi,\beta)$ satisfies uniform exponential decay of covariance with decay rate $\xi$, it satisfies, 
\\\textbf{1)} Uniform strong local indistinguishability with decay length $\xi$. That is for each $ABC=\Gamma\subset\subset\Lambda$, where $A,B,C$ are disjoint and $l=\dist(A,C)>r$, it holds that
\begin{equation}
    \tr_{BC}[\sigma^{ABC}]\overset{\epsilon(l)}{\sim}\tr_B[\sigma^{AB}] \text{ with } \epsilon_{A,C}(l)= e^{\mathcal{O}(\beta\min\{|\partial A|,|\partial_A B|\})}\mathcal{O}(|\partial C|)\exp{\left(-\frac{l-r}{\xi}\right)}.
    \label{equ:StrongLocalIndistdD}
\end{equation} Here $\partial_AB=\partial B\cap A$.
\vspace{0.2cm}
\\ \textbf{2)} Uniform mixing condition with decay length $\xi$. That is for each $ABC=\Gamma\subset\subset\Lambda$, where $A,B,C$ are disjoint and $l=\dist(A,C)>3r$, $\sigma^\Gamma$ satisfies
\begin{equation} 
    \sigma_{AC}\overset{\eta(l)}{\sim}\sigma_A\otimes\sigma_C \text{ with }
    \eta_{A,C}(l)= e^{\mathcal{O}(\beta(|\partial A|+|\partial C|))}\exp{\left(-\frac{l-2r}{\xi}\right)}.
    \label{equ:StrongTensorization}
\end{equation}
In particular, by \eqref{eq:ineq_mutualinfo_mixingcondition}, it also satisfies 
uniform exponential decay of mutual information with
\begin{align}
    I_{\sigma^{ABC}}(A:C) \leq \eta_{A,C}(l) \, .
    \label{equ:decayofMutualInformation}
\end{align}
\end{theorem}

\begin{remark}
In the statement of the theorem we have an exponential decay w.r.t $\dist(A,C)$ but we also have a spurious pre-factor which is growing exponentially in the size of the boundary of $A$ for the strong local indistinguishability and $A$ and $C$ for the mixing condition. This is a consequence of the proof technique we are using and we expect it to be not physically tight. 
\end{remark}

We will later be using these clustering results on connected and growing sets $A,C$. For the 1 dimensional spin chains these prefactors are just constant and thus negligible. In the $2$-dimensional square lattice this exponentially increasing pre-factor can still be dominated by the decay in the distance if the decay length $\xi$ is short enough. 

Since by \cite{art:QuantumGibbsSamplers-kastoryano2016quantum} the existence of a gap in the QMS with fixed point $\sigma$ implies uniform exponential decay of covariance, we immediately have strong local indistinguishability, mixing condition, and exponential decay of the mutual information from the gap property. This implies that 1-dimensional quantum spin chains satisfy these properties at any temperature for geometrically-local, commuting, bounded Hamiltonians and in $D$-dimensional regular latices, although with exponential prefactors in the boundaries of local regions. 
With this in mind, implication 1) above is a strict strengthening of the local indistinguishability result in \cite{art:LocalityofTemperature,art:FiniteCorrelationLengthImpliesEfficient} for commuting Hamiltonians. Implication 2) can be viewed as an extension of the results in \cite{art:Bluhm2022exponentialdecayof} to any lattice with finite growth constant under the additional condition of commutativity of the Hamiltonian. 

Before proving Theorem \ref{thm:ddimStrongLocalIndist}, we note that for 1-dimensional quantum spin chains we can also establish the above results and implications for non-commuting Hamiltonians. In this case, 2) is the main result of \cite{art:Bluhm2022exponentialdecayof}, and 1) is as follows. 

\begin{theorem}[Strong local indistinguishability in 1D]\label{thm:StrongLocalIndist1D}
Let $\Phi$ be a geometrically-$r$-local and $J-$bounded potential on $\Z$ and $\beta>0$ such that the to $(\Z,\Phi,\beta)$ associated family of Gibbs states $\{\sigma^\Gamma\}_{\Gamma\subset\subset\Z}$ satisfies uniform exponential decay of covariance.
Then there exist constants $K,a>0$, such that for each interval $I=ABC\subset\subset\mathbb{Z}$, where $B$ shields $A$ away from $C$ and $2l:=|B|= \dist(A,C)>r$ it holds that
\begin{equation}
    \|(\tr_{BC}[\sigma^{ABC}])(\tr_{B}[\sigma^{AB}])^{-1}-\1\|\leq Ke^{-al}.
    \label{equ:StrongLocalIndist}
\end{equation}
Note that $K,a>0$ depend only on the interaction range $r$ and effective strength $\beta J$.
\end{theorem}

\Cref{thm:StrongLocalIndist1D} will be proven in \Cref{sec:proof_stronglocalind_1d}. Its proof is analogous to the one for \Cref{thm:ddimStrongLocalIndist}, using some additional technical prerequisites from \cite{art:Bluhm2022exponentialdecayof}.
For the latter we will first need the following technical Lemma.
Let us use the standard notation
\begin{equation}
  E_{A,B}^\beta:=\exp(-\beta H_{AB})\exp(\beta (H_A+H_B))  
\end{equation}
to denote Araki's expansionals. We omit the superscript with the inverse temperature $\beta >0$ when it is unnecessary or clear from the context.

\begin{lemma}
\label{lemma:ImportantCommutingLemma1}
Let $\Phi$ be a geometrically-$r$-local, $J$-bounded, commuting potential on a quantum spin system $\Lambda$ with finite growth constant $\nu$. Let $\Gamma=ABC\subset\subset\Lambda$, with $B$ shielding $A$ from $C$. Then, if we denote by $\partial_A B$ the boundary of $B$ in $A$, i.e. $\partial_A B := (\partial B) \cap A$, similarly $\partial_BA$ and $\partial_{A,B}:=\min\{|\partial_B A|,|\partial_A B|\}$ the following bounds hold with $K_{A,B}:=\exp({\mathcal{O}(\beta |\partial_{A,B}|)})$ independent of $l:=\dist(A,C)$.
\begin{enumerate}
    \item[$1)$] For every $\beta>0$, we have
    \begin{equation}
        \|E_{A,B}^{\pm1}\|\leq K_{A,B}  \, .
    \end{equation}
    \item[$2)$] For any strictly positive $Q\in\mathcal{B}(\mathcal{H}_\Lambda)$, we have 
    \begin{equation}
        \|Q^{\mp1}\|^{-1} \leq \|\tr_B[\sigma^BQ]^{\pm1}\| \leq \|Q^{\pm1}\| \, .
    \end{equation} 
    \item[$3)$] The following bounds also hold respectively
    \begin{equation}
        \left\lbrace \|\tr_B[\sigma^BE_{A,B}^{\pm1}]^{\pm1}\| , \,  \|\tr_B[\sigma^BE_{B,C}^{\pm1}]^{\pm1}\| , \,  \Tr[\sigma^{AB}E_{A,B}^{\pm1}]^{\pm1}  \right\rbrace \leq \left\{K_{A,B},K_{B,C},K_{A,B}\right\} \, .
    \end{equation}
    \begin{equation}
\|\tr_B[\sigma^BE_{A,B}^{\pm1}E_{AB,C}^{\pm1}]^{\pm1}\| \leq K_{A,B}K_{B,C} \, .
    \end{equation}
\end{enumerate}
The same inequalities hold when exchanging the order of $\sigma$ and the expansionals inside the partial traces. 
Note that the \text{big-$\mathcal O$} notation refers to the dependence in $\beta$ and the boundaries and omits dependence on $J,d,r,\nu$. 
\end{lemma}

The proof of this result is deferred to \Cref{sec:proof_lemma_commuting}.

\begin{remark}
Our proof of Lemma \ref{lemma:ImportantCommutingLemma1} requires the commutativity of the Hamiltonian, since we require $E_{A,B}\geq0$, see 2). If positivity could be proven without this assumption, 
or 3) directly some other way, we believe that we could establish \Cref{thm:ddimStrongLocalIndist} without the additional assumption of commutativity. 
\end{remark}
By the combined use of Lemma \ref{lemma:ImportantCommutingLemma1}, clever rewritings inspired by the proofs in \cite{art:Bluhm2022exponentialdecayof}, and repeated application of local indistinguishability, we can prove the main theorem of this section. 

\begin{proof}[Proof of Theorem \ref{thm:ddimStrongLocalIndist}]
We first note that the following holds: 
\begin{equation}
   \|\omega^{\frac{1}{2}}\tau^{-1}\omega^{\frac{1}{2}}-\1\|\leq \|\omega\tau^{-1}-\1\| ,  
\end{equation}
by e.g. \cite[Proposition IX.1.1]{Bhatia1997}.
Set $l:=\dist(A,C)>r$. Since we are assuming uniform exponential decay of covariance with correlation length $\xi$, we may write
\begin{equation}
    \text{Cov}_{\sigma}(A:C) \leq \tilde{K} \exp\left(-\frac{l}{\xi}\right)
\end{equation}
for some constant $\tilde{K}>0$. 
\vspace{0.2cm}

\noindent \textbf{1)} To show strong local indistinguishability \eqref{equ:StrongLocalIndistdD}, first note that
\begin{equation}
    \tr_{BC}[\sigma^{ABC}] =  e^{-H_A} \tr_{BC}[\sigma^{BC}E_{A,BC}]  \frac{Z_{BC}}{Z_{ABC}} \, , 
\end{equation}
with $Z_\chi :=\Tr[e^{-H_\chi}]$ for any $\chi \subset \Gamma$. Then, we start by rewriting
\begin{align}
    (\tr_{BC}[\sigma^{ABC}])(\tr_B[\sigma^{AB}])^{-1} = \tr_{BC}[\sigma^{BC}E_{A,BC}]\tr_{B}[\sigma^BE_{A,B}]^{-1}\lambda_{ABC}^{-1} \, ,
\end{align}
where 
\begin{equation}
  \lambda_{ABC}=\frac{Z_{ABC} Z_B}{Z_{AB} Z_{BC}} = \frac{\Tr[\sigma^{AB}E^{-1}_{A,B}]}{\Tr[\sigma^{ABC}E_{A,BC}^{-1}]}   \, .
\end{equation}
Note that whenever we omit the subscript in $\Tr[X]$, we are referring to a total trace in the subsystems where $X$ has non-trivial support.
Thus, we need to upper bound 
\begin{equation}
    \norm{ \tr_{BC}[\sigma^{BC}E_{A,BC}]\tr_{B}[\sigma^BE_{A,B}]^{-1}\lambda_{ABC}^{-1} - \1} \, .
\end{equation}
We do this by splitting the proof into several claims, proven independently.

\vspace{0.2cm}

\noindent \underline{\textit{Claim 1:}} $|\lambda_{ABC}^{\mp1}-1|$ is exponentially decaying in $l$ with decay length $\xi$. 

\vspace{0.2cm}

\noindent \underline{\textit{Proof:}} Considering first $\lambda_{ABC}$, we have
\begin{align}
    |\lambda_{ABC}-1|&= \frac{1}{\Tr[\sigma^{ABC}E_{A,BC}^{-1}]}\left|\Tr[\sigma^{AB}E_{A,B}^{-1}]-\Tr[\sigma^{ABC}E_{A,BC}^{-1}]\right| \\&\leq \|E_{A,BC}^{{+1}}\|\left|\Tr[\sigma^{AB}E_{A,B}^{-1}]-\Tr[\sigma^{ABC}E_{A,BC}^{-1}]\right|,
\end{align} 
where the last inequality follows from Lemma \ref{lemma:ImportantCommutingLemma1} 2). Now, set $B=B_1B_2$ with $B_1:=\partial A$, $B_2=B\setminus B_1$, s.t. $\dist(A,B_2)=r, \; \dist(B_1,C)\geq l-r$. We also set $A=A_0A_1$, where $A_1=\partial_AB$ and $A_0=A\setminus A_1$ Then $E_{A,BC}=E_{A,B}=E_{A_1,B_1}$. Therefore, 
\begin{align}
    \left|\Tr[\sigma^{AB}E_{A,B}^{-1}]-\Tr[\sigma^{ABC}E_{A,BC}^{-1}]\right| &= \left|\Tr_{AB_1}[\tr_{B_2}(\sigma^{AB}E_{A,B}^{-1})]-\Tr_{AB_1}[\tr_{B_2C}(\sigma^{ABC}E_{A,BC}^{-1})]\right| \\ &=\left|\Tr_{AB_1}[(\tr_{B_2}[\sigma^{AB}]-\tr_{B_2C}[\sigma^{ABC}])E_{A_1,B_1}^{-1}]\right| \\&\leq\|\tr_{B_2}[\sigma^{AB}]-\tr_{B_2C}[\sigma^{ABC}]\|_1\|E_{A_1,B_1}^{-1}\| ,
\end{align}
where the last line follows from H\"{o}lder's inequality, and where we write the subscript in $\Tr$ to emphasize the systems over which we are tracing out. 
By \cite[Theorem 5]{art:FiniteCorrelationLengthImpliesEfficient}, uniform exponential decay of correlations implies local indistinguishability with the same decay and an additional factor $|\partial C|$. So we have with Lemma \ref{lemma:ImportantCommutingLemma1} 1) that
\begin{align}
 \|(\tr_{B_2}[\sigma^{AB}]-\tr_{B_2C}[\sigma^{ABC}])\|_1\|E_{A_1,B_1}^{-1}\|  &\leq K_{A,B} |\partial C|\tilde{K}\exp(\frac{-1}{\xi}\dist(B_1,C)) \\&=K_{A,B} \tilde{K} |\partial C|\exp\left(-\frac{l-r}{\xi}\right).
\end{align}
This allows us to conclude 
\begin{equation}
   |\lambda_{ABC}-1|\leq K_{A,B}^2 \tilde{K}|\partial C|\exp\left(-\frac{l-r}{\xi}\right) \, . 
\end{equation}
The same follows for $|\lambda^{-1}_{ABC}-1|$ analogously or by application of the geometric series to the above. 
This concludes the proof of Claim 1. 
Now we can rewrite
\begin{align}
 \norm{ \tr_{BC}[\sigma^{BC}E_{A,BC}]\tr_{B}[\sigma^BE_{A,B}]^{-1}\lambda_{ABC}^{-1} - \1} \, 
\end{align}
to
\begin{align}
    \|(\tr_{BC}[\sigma^{ABC}])(\tr_{B}[\sigma^{AB}])^{-1}-\1\|   &=\norm{ \tr_{BC}[\sigma^{BC}E_{A,BC}]\tr_{B}[\sigma^BE_{A,B}]^{-1}\lambda_{ABC}^{-1} - \1}  \\
   & \leq \|\tr_{BC}[\sigma^{BC}E_{A,BC}]\|\|(\tr_{B}[\sigma^{B}E_{A,B}])^{-1}\||\lambda_{ABC}^{-1}-1| \\
   & \quad +\|(\tr_{B}[\sigma^{B}E_{A,B}])^{-1}\| \|\tr_{BC}[\sigma^{BC}E_{A,BC}]-\tr_{B}[\sigma^{B}E_{A,B}]\|.
\end{align}

\noindent \underline{\textit{Claim 2:}} $\|\tr_{BC}[\sigma^{BC}E_{A,BC}]-\tr_{B}[\sigma^{B}E_{A,B}]\|$ is exponentially decaying in $l$ with decay rate $\xi$. 
\vspace{0.2cm}

\noindent \underline{\textit{Proof:}}
Again set $B=B_1B_2$, $B_1:=\partial A$, $B_2:=B\setminus B_1$, and split $A$ into $A_1:= \partial_A B$ and $A_0 := A \setminus A_1$. Thus $\dist(B_1,C)\geq l-r$ and $E_{A,BC}=E_{A,B}=E_{A_1,B_1}$. Then, by local indistinguishability
\begin{align}
    \|\tr_{BC}[\sigma^{BC}E_{A,BC}]-\tr_{B}[\sigma^{B}E_{A,B}]\| &= \|\tr_{B_1}[(\tr_{B_2}[\sigma^{B_1B_2}]-\tr_{B_2C}[\sigma^{B_1B_2C}])E_{A_1,B_1}]\| \\[1mm] 
    &\leq  \|E_{A_1,B_1}\|d^{|A_1|}\, \widetilde{K}\, |\partial C|\exp\left(-\frac{\dist(B_1,C)}{\xi}\right)\\
    &\leq  K_{A,B}d^{|\partial_AB|} \widetilde{K}  |\partial C| \exp\left(-\frac{l-r}{\xi}\right).   
\end{align} 
This concludes the proof of Claim 2.

\vspace{0.2cm}
\noindent Putting everything together, we get the desired result,
\begin{align}
   & \|(\tr_{BC}[\sigma^{ABC}])(\tr_{B}[\sigma^{AB}])^{-1}-\1\|\\
   &\hspace{3cm}\leq  K_{A,B}^4|\partial C|\tilde{K}\exp\left(-\frac{l-r}{\xi}\right)+K_{A,B}^2d^{|\partial_A B|} K^2\widetilde{K}|\partial C|\exp\left(-\frac{l-r}{\xi}\right) \\ &\hspace{3cm}= \exp{\mathcal{O}(\beta|\partial_{A,B}|)}\mathcal{O}(|\partial C|)\exp\left(-\frac{l-r}{\xi}\right).
\end{align}

\noindent \textbf{2)} Assume $l\geq 3r$. To prove the mixing condition/strong tensorization \eqref{equ:StrongTensorization}, following the steps above, or similarly those of \cite[Corollary 8.3]{art:Bluhm2022exponentialdecayof}, we can rewrite
\begin{align}
   \|\sigma_{AC}(\sigma_A\otimes\sigma_C)^{-1}-\1\| &\leq K_{A,B}K_{B,C} \|\tr_{B}[\sigma^{B}E_{A,B}E_{AB,C}]\|{|\lambda_{ABC}-1|} \\
   & \hspace{0.2cm}  + K_{A,B}K_{B,C}  {\|\tr_{BC}[\sigma^{BC}E_{A,BC}]\tr_{AB}[\sigma^{AB}E_{AB,C}]-\tr_{B}[\sigma^BE_{A,B}E_{AB,C}]\|}\\
   &\le K_{A,B}^4K_{B,C}^2\widetilde{K}|\partial C|\exp\left(-\frac{l-r}{\xi}\right)\, \\
   &\hspace{0.2cm}+K_{A,B}K_{B,C}{\|\tr_{BC}[\sigma^{BC}E_{A,BC}]\tr_{AB}[\sigma^{AB}E_{AB,C}]-\tr_{B}[\sigma^BE_{A,B}E_{AB,C}]\|},
\end{align}
where the second inequality follows from Claim 1 as well as Lemma \ref{lemma:ImportantCommutingLemma1}.
\vspace{0.2cm}

\noindent \underline{\textit{Claim 3:}} $\|\tr_{BC}[\sigma^{BC}E_{A,BC}]\tr_{AB}[\sigma^{AB}E_{AB,C}]-\tr_{B}[\sigma^BE_{A,B}E_{AB,C}]\|$ is exponentially decaying in $l$ with correlation length $\xi$. 
\vspace{0.2cm}

\noindent \underline{\textit{Proof:}} Set $B=B_1B_2B_3$ with $B_1:=\partial A, B_3:=\partial C, B_2:=B\setminus(B_1\cup B_3)$, then $\dist(B_1,B_3)\geq l-2r$ and $E_{A,BC}=E_{A,B}=E_{A,B_1}$ and $E_{AB,C}=E_{B,C}=E_{B_3,C}$ and consequently
\begin{align}
    &\|\tr_{BC}[\sigma^{BC}E_{A,BC}]\tr_{AB}[\sigma^{AB}E_{AB,C}]-\tr_{B}[\sigma^BE_{A,B}E_{AB,C}]\| \\[1mm]  & \hspace{4cm} = \|\tr_{BC}[\sigma^{BC}E_{A,B_1}]\tr_{AB}[\sigma^{AB}E_{B_3,C}]-\tr_{B}[\sigma^BE_{A,B_1}E_{B_3,C}]\| \\[1mm]   &\hspace{4cm}  \leq \underbrace{\|\tr_{B}[\sigma^BE_{A,B_1}E_{B_3,C}]-\tr_B[\sigma^BE_{A,B_1}]\tr_B[\sigma^BE_{B_3,C}]\|}_{\text{(I)}} \\& \hspace{4cm} +\underbrace{\|\tr_B[\sigma^BE_{A,B_1}]\tr_B[\sigma^BE_{B_3,C}]-\tr_{BC}[\sigma^{BC}E_{A,B_1}]\tr_{AB}[\sigma^{AB}E_{B_3,C}]\|}_{\text{(II)}}.
\end{align}
Next, we bound (I) and (II) separately. To bound (II) we use that, by the proof of \textit{Claim 2},
\begin{align}
    \|\tr_{B}[\sigma^BE_{A,B_1}]-\tr_{BC}[\sigma^{BC}E_{A,B_1}]\|&\leq \|E_{A,B_1}\|\|\tr_{B_2B_3}[\sigma^{B}]-\tr_{B_2B_3C}[\sigma^{BC}]\|_1 \\ &\leq K_{A,B}d^{|\partial_AB|}|\partial C|\widetilde{K}\exp{\left(-\frac{1}{\xi}\dist(B_1,C)\right)}, \\[1mm]
     \|\tr_{B}[\sigma^BE_{B_3,C}]-\tr_{AB}[\sigma^{AB}E_{B_3,C}]\|&\leq \|E_{B_3,C}\|\|\tr_{B_1B_2}[\sigma^{B}]-\tr_{AB_1B_2}[\sigma^{AB}]\|_1\\ &\leq K_{B,C}d^{|\partial_CB|}|\partial A|\widetilde{K}\exp{\left(-\frac{1}{\xi}\dist(A,B_3)\right)}. 
\end{align}
Then, putting both bounds together, \begin{align}
    \text{(II)} &\leq \|\tr_{B}[\sigma^BE_{A,B_1}] \| \|\tr_{B}[\sigma^BE_{B_3,C}]-\tr_{AB}[\sigma^{AB}E_{B_3,C}]\| \\[1mm] & \quad + \|\tr_{B}[\sigma^BE_{B_3,C}]\| \|\tr_{B}[\sigma^BE_{A,B_1}]-\tr_{BC}[\sigma^{BC}E_{A,B_1}]\| \\[1mm] &\leq K_{A,B}K_{B,C}d^{|\partial_CB|}|\partial A|\tilde{K}\exp{\left(-\frac{l-r}{\xi}\right)}+K_{B,C}K_{A,B}d^{|\partial_AB|}|\partial C|\tilde{K}\exp{\left(-\frac{l-r}{\xi}\right)} \\ &=\exp(\mathcal{O}(\beta|\partial A|+\beta|\partial C|))\mathcal{O}(|\partial A|+|\partial C|)\exp\left(-\frac{l-r}{\xi}\right). 
\end{align}

To bound (I), we use that exponential decay of covariance directly implies
\begin{equation}
  \|\sigma_{AC}-\sigma_{A}\otimes\sigma_{C}\|_1\leq \tilde{K}\exp\left(-\frac{\dist(A,C)}{\xi}\right)  
\end{equation}
for $\sigma\equiv\sigma^{ABC}$, by Hölder duality. Thus,
\begin{align}
    \text{(I)} &= \|\tr_{B_1B_3}[\sigma_{B_1B_3}(E_{A,B_1}\otimes E_{B_3,C})]-\tr_{B_1}[\sigma_{B_1}E_{A,B_1}]\otimes \tr_{B_3}[\sigma_{B_3}E_{B_3,C}]\| \\[1mm] &= \|\tr_{B_1B_3}[(\sigma_{B_1B_3}-\sigma_{B_1}\otimes\sigma_{B_3})(E_{A,B_1}\otimes E_{B_3,C})] \| \\ & \hspace{-0.95cm}\overset{\text{Proof of Claim 2}}{\leq} \|E_{A,B_1}\otimes E_{B_3,C}\| \|\sigma_{B_1B_3}-\sigma_{B_1}\otimes\sigma_{B_3}\|_1d^{|\partial_AB|+|\partial_CB|} \\[1mm] &\leq K_{A,B}K_{B,C}\widetilde{K}\exp\left(-\frac{\dist{B_1,B_3}}{\xi}\right) = \exp{\mathcal{O}(\beta|\partial A|+|\partial C|)}\exp\left(-\frac{l-2r}{\xi}\right).
\end{align}
Therefore, we conclude the proof of \textit{Claim 3}, since putting the bounds above together yields
\begin{align}
   & \|\tr_{BC}[\sigma^{BC}E_{A,BC}]\tr_{AB}[\sigma^{AB}E_{AB,C}]-\tr_{B}[\sigma^BE_{A,B}E_{AB,C}]\| \\
    & \hspace{3.6cm} = \exp(\mathcal{O}(\beta|\partial A|+\beta|\partial C|))\mathcal{O}(|\partial A|+|\partial C|)\tilde{K}\exp\left(-\frac{l-2r}{\xi}\right) \, .
\end{align}
 Finally, we can obtain the following bound for the mixing condition 
\begin{align}
    \|\sigma_{AC}(\sigma_{A}\otimes\sigma_{C})^{-1}-\1\|\leq \exp{\mathcal{O}(\beta(|\partial A|+|\partial C|))}\mathcal{O}(|\partial A|+|\partial C|)\exp\left(-\frac{l-2r}{\xi}\right).
\end{align}
The result on the mutual information follows directly from the one for the mixing condition by \eqref{eq:ineq_mutualinfo_mixingcondition}. 
\end{proof}
We can also get the mixing condition directly from strong local indistinguishability. 

\begin{lemma}[Strong local indistinguishability implies mixing conditon]
Let  $\Lambda$ be a graph with finite growth constant. Define $S_\Lambda(l):=\sup_{x\in\Lambda}|\{v\in\Lambda|\dist(x,v=l)\}|$ the surface area of the maximal $l$-sphere in $\Lambda$. Note that we have
\begin{align}
    &\Z^D: D\textup{-dim hypercubic lattice: } S_{\Z^D}(l) = \mathcal{O}(l^{D-1}), \\
    &\Lambda \textup{ sub-exponential: } S_{\Lambda}(l) = \exp(o(l)), \\
    &\mathbb{T}_b: b\textup{-ary tree: } S_{\mathbb{T}_b}(l) = b^l.
\end{align}
Now if $\{\sigma^\Gamma\}_\Gamma$ is a family of states which satisfy uniform strong local indistinguishability with decay function $l\mapsto\epsilon_{A,C}(l)$, i.e. each $\sigma\equiv\sigma^\Gamma$ satisfies $\|\tr_{BC}\sigma^{ABC}(\tr_{B}\sigma^{AB})^{-1}-\1\|\leq \epsilon_{A,C}(l)$, then there exists a suitable function $l\mapsto\eta_{A,C}(l)$ such that this family also satisfies the uniform mixing condition as $\|\sigma_{AC}(\sigma_A\otimes\sigma_C)^{-1}-\1\|\leq\eta_{A,C}(l)$ satisfying the following implications.
\begin{align}
    \epsilon_{A,C}(l) = \exp\left(\mathcal{O}(|\partial A|)\right)\mathcal{O}(|\partial C|)\exp\left(-\frac{l}{\xi}\right) &\implies \eta_{A,C}(l) = \exp\left(\mathcal{O}(|\partial A|+|\partial C|)\right)S_\Lambda\left(\frac{l}{3}\right)^3\exp\left(-\frac{l}{3\xi}\right), \\
    \epsilon_{A,C}(l) = \mathcal{O}(|\partial A|,|\partial C|)\exp\left(-\frac{l}{\xi}\right) &\implies \eta_{A,C}(l) = \mathcal{O}\left(\poly(|\partial A|,|\partial C|)\right)S_\Lambda\left(\frac{l}{3}\right)^3\exp\left(-\frac{l}{3\xi}\right).
\end{align} 
\end{lemma}
\begin{remark} {
    In the case where strong local indistinguishability follows from decay of covariance, this derivation provides no improvement in scaling over \Cref{equ:StrongLocalIndistdD}. However, whenever one can assume strong local indistinguishability with a polynomial prefactor on subexponential, respectively, polynomial graphs, then this directly gives us the mixing condition with a sub-exponential, respectively, polynomial prefactor. } 
\end{remark}
\begin{proof}
Fix some $\Gamma\subset\subset\Lambda$ and some $A,C\subset\Gamma$ such that $\dist(A,C)=l>3r$. W.l.o.g. we assume $l$ is divisible by 3, else use $l-1$ or $l-2$. Split $B:=\Gamma\setminus(AC)=B_1B_2B_3$ into 3 regions, where $B_1:=\partial_{\frac{l}{3}}A:=\{x\in B|\dist(x,A)\leq \frac{l}{3}\}$ and $B_3:=\partial_{\frac{l}{3}}C$. Then $\dist(A,B_2),\dist(B_1,B_3),\dist(B_2,C)\geq\frac{l}{3}$. Now by the transitivity and tensor multiplicativity of the relation $\sim$ (see proposition \ref{prop:properties_relation}) we have that
\begin{align}
    \sigma_{AC} = \tr_{B_1B_2B_3}\sigma^{AB_1B_2B_3C} \overset{\epsilon_1}{\sim} \tr_{B_1B_3}\sigma^{AB_1}\otimes\sigma^{B_3C} =  \overset{\epsilon_2\epsilon_3+\epsilon_2+\epsilon_3}{\sim}\tr_{BC}\sigma^{ABC}\otimes\tr_{AB}\sigma^{ABC} = \sigma_A\otimes \sigma_C.
\end{align}
Using the crude bounds $|B_1|\leq |\partial A|f_\Lambda(\frac{l}{3})$ and $|B_3|\leq |\partial C|f_\Lambda(\frac{l}{3})$, and the assumptions on the functions $\epsilon_i$ arising from strong local indistinguishability directly yields the claim. Note that the $A,C$ in the mixing condition are not necessarily the same regions $A,C$ as in the strong local indistinguishability.
\end{proof}

\subsection{q$\mathbb{L}_1\to\mathbb{L}_\infty$-clustering from decay of covariance or temperature}
Here we consider a notion of clustering introduced in \cite[Definition 8]{art:2localPaper} as q$\mathbb{L}_1\to\mathbb{L}_\infty$-clustering, where it is instrumental in implying rapid thermalization of the Schmidt dynamics. As such, this notion of clustering will be key in our proofs of rapid mixing.

This is in principle a more abstract notion of clustering defined w.r.t the family of invariant states of a family of Lindbladians with certain properties. In this section we will only consider the family of Lindbladians $\{\mathcal{L}^S_\Gamma\}_{\Gamma\subset\subset\Lambda}$ to be the one corresponding to the Schmidt conditional expectations $\{E^S_X\}_{X\subset\Gamma}$ introduced in \ref{subsec:SchmidtCondExp}. In this case 
\begin{definition}[q$\mathbb{L}_1\to\mathbb{L}_\infty$-clustering]\label{equ:q1toinfty_clustering}
The uniform family of primitive, reversible, and frustration-free Schmidt generators $\{\mathcal{L}^S_{\Gamma}\}_{\Gamma\subset\subset\Lambda}$ with fixed points $\{\sigma^\Gamma\}_{\Gamma\subset\subset\Lambda}$ satisfies \textit{uniform} q$\mathbb{L}_1\to\mathbb{L}_\infty$-\textit{clustering of correlations} if 
there exists a function $(C,D,l)\mapsto\eta_{C,D}(l)$, exponentially decaying in $l$, with decay rate independent of $C,D$ such that for two overlapping subregions $C,D\subset\Gamma$ with $l:=\dist(C\setminus D,D\setminus C)$, it holds that
\begin{align}
    \max_{\alpha=\{\alpha_i\}_{i\in I_{C\cup D}}} \|E^{S,(\alpha)}_C\circ E^{S,(\alpha)}_D-E^{S,(\alpha)}_{C\cup D}: \mathbb{L}_1(\tau_{C\cup D}^{\alpha})\to \mathbb{L}_\infty\| \leq \eta_{C,D}(l)\equiv \eta(l),
    \label{def:qL1toLinfty}
\end{align} where $\{E^S_X\}_{X\subset\subset\Lambda}$ are the Schmidt conditional expectations and $(\alpha)$ are the boundary conditions of the subset $CD:=C\cup D\subset\Lambda$. 
\end{definition}

Recall that the Schmidt conditional expectations can only sensibly be defined for nearest neighbour interacting systems, hence this notion of clustering can only exist on systems with geometrically-$2$-local potentials.

In the following, we show that this notion of clustering is implied by uniform decay of covariance via strong local indistinguishability, 
which constitutes another result of independent interest from this work.  

\begin{theorem}[Decay of covariance is equivalent to q$\mathbb{L}_1\to\mathbb{L}_\infty$-clustering]\label{thm:weaktostrongclustering}
Let $\Lambda$ be a 
graph with finite growth constant and let $\Phi$ be a bounded, commuting, nearest-neighbour potential on $\Lambda$.  Then if, for some $\beta>0$, the family of Gibbs states $\{\sigma^\Gamma\}_{\Gamma\subset\subset\Lambda}$ associated to $(\Lambda,\Phi,\beta)$ satisfies uniform exponential decay of covariance, then the Schmidt generators associated to $(\Lambda,\Phi,\beta)$ 
satisfy uniform q$\mathbb{L}_1\to\mathbb{L}_\infty$-clustering as in \eqref{def:qL1toLinfty}, i.e.
\begin{center}
   uniform $\epsilon(l)$-decay of covariance $ \implies $ $\eta_{C,D}(l)$-q$\mathbb{L}_1\to\mathbb{L}_\infty$-clustering, 
\end{center}
 where $\eta_{C,D}(l) = \textup{exp}(\mathcal{O}(|\partial (C\setminus D)|+|\partial (D\setminus C)|))\mathcal{O}(1)\epsilon(l)$.
The $\Longleftarrow$ implication and its proof can be found in \cite{art:2localPaper}. That implication also follows from  \Cref{thm:mainGap} in the case $\Lambda=\mathbb{Z}$.
\end{theorem}

An example of subsets $C,D$ can be found in \Cref{fig:Tree}. 
The exponential scaling of the prefactors of the decay in the q$\mathbb{L}_1\to\mathbb{L}_\infty$-clustering comes directly from Theorem \ref{thm:ddimStrongLocalIndist}, so if one had strong local indistinguishability with a linear (or polynomial) dependence on the boundary regions, then the decay in the q$\mathbb{L}_1\to\mathbb{L}_\infty$-clustering would also only be at worst polynomial in the boundaries of $C\setminus D$ and $D\setminus C$.

Such a decay is already known to hold at high enough temperatures. Before proving \Cref{thm:weaktostrongclustering}, we also present it in the context of our work here, as it will be a basis for our second main result.

\begin{theorem}[q$\mathbb{L}_1\to\mathbb{L}_\infty$-clustering from high temperature; Theorems 6,7, Proposition 2 in \cite{art:2localPaper}]\label{thm:clusteringFromHighTemp}
Let $\Lambda$ be a graph with finite growth constant $\nu$ and $\Phi$ a uniformly $J$-bounded, commuting, geometrically-2-local potential on it. Then if the temperature $\beta^{-1}$ is high enough, as $0\leq\beta<\frac{1}{10eJ\nu}$, the family of Schmidt generators associated to $(\Lambda,\Phi,\beta)$ satisfies uniform q$\mathbb{L}_1\to\mathbb{L}_\infty$-clustering as in \eqref{def:qL1toLinfty} with decay function $$\eta_{C,D}(l)\leq K|C\cup D|\exp\left(-\frac{l}{\xi^\prime}\right),$$
where $l=\dist(C\setminus D,D\setminus C)$ and $K,\xi^\prime>0$ are some fixed constants independent of the local regions $C,D$, or $\Gamma$.    
\end{theorem}
Note that in this so-called the ``high temperature" regime we have a linear dependence of the pre-factor on the size of the local regions, as compared to an exponential one in Theorem \ref{thm:weaktostrongclustering}.
Although in \cite{art:2localPaper,Harrow_2020} they only explicitly consider hyper-cubic lattice, the proofs there works equally for lattices with bounded degree, including trees. 

\begin{proof}[Proof of \Cref{thm:weaktostrongclustering}]
We show that, under the assumptions of the theorem, the strong local indistinguishability and mixing condition imply 
\begin{align}
  \max_{(\alpha)}  \|E^{S,(\alpha)}_C\circ E^{S,(\alpha)}_D-E^{S,(\alpha)}_{C\cup D}: \mathbb{L}_{1,\tau^{C\cup D}_{\alpha}}\to \mathbb{L}_\infty\| \leq \eta(l)
\end{align} explicitly for the Schmidt conditional expectations. 
Since both of these are, by Theorem \ref{thm:ddimStrongLocalIndist}, implied by uniform decay of correlations, this proves the theorem. 
Here $(\alpha)$ is a boundary condition of the subset $CD:=C\cup D\subset\Lambda$ and $\eta(l)$ is exponentially decaying with decay length $\xi$. We split the proof into two steps: First, we establish that what we need to show is algebraically equivalent to a statement $\sigma_1\overset{\eta(l)}{\sim}\sigma_2$ for two states $\sigma_1,\sigma_2$, and we then employ results of \Cref{thm:ddimStrongLocalIndist} to prove this statement from the assumptions of the theorem.

Before we continue, let us first establish some nomenclature for the subregions we are considering. We split the region $CD$ into the following disjoint subsets $\tilde{l}_\out\tilde{l}_{\inn}El_\out l_\inn F r_\inn r_\out G \tilde{r}_\inn \tilde{r}_\out$, where $C=ElF, C_\inn=\tilde{l}_{\inn}ElFr_\inn, D=FrG, D_\inn = l_\inn FrG \tilde{r}_\inn$. The projectors $P^{(\alpha_{\tilde{l}})},P^{(\gamma_l)},P^{(\beta_r)},P^{(\alpha_{\tilde{r}})}$ act on the regions $\tilde{l},l,r,\tilde{r}$ respectively. For a graphical representation of this see \Cref{fig:Tree}.

\begin{figure}[H]
    \centering
   \includegraphics[width=0.7\linewidth]{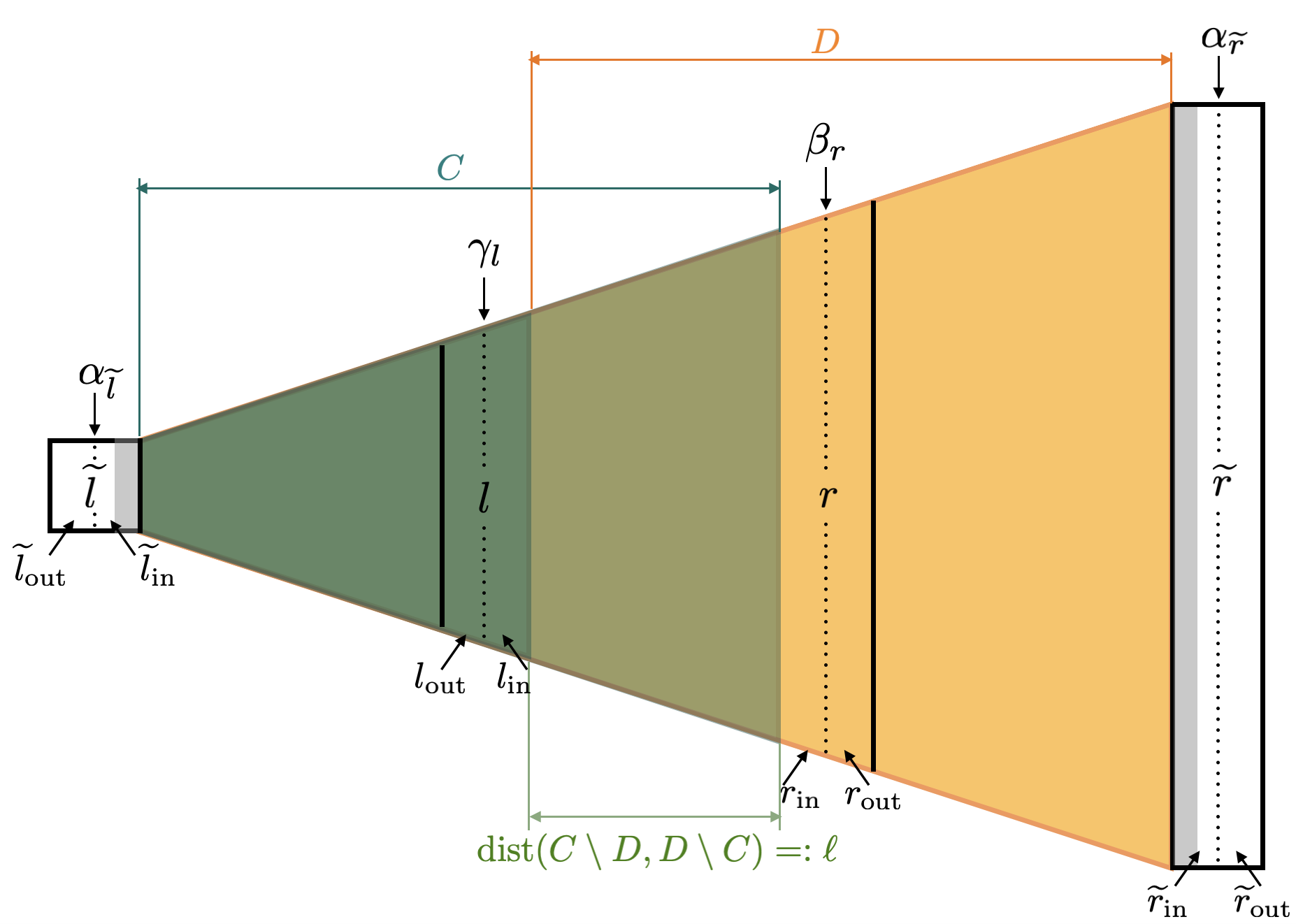}
    \caption{Partition of a subregion $CD$ of a tree into two overlapping subregions {\color{darkgreen} C} and {\color{RoyalYellow} D}. We have $E:=C\setminus(D\partial), F:=C\cap D, G:=D\setminus(C\partial)$. The splitting of the boundary Hilbert spaces corresponding to a boundary site $\{b_i\}$ in the boundary of a region $A\in\{E,F,G\}$ into $\bigotimes_{i\in I_A}P^{\alpha_i}\mathcal{H}_{b_i}=\mathcal{H}^{\alpha_i}_{\partial_\inn A}\otimes\mathcal{H}^{\alpha_i}_{\partial_\out A} =\left(\mathcal{H}^{\alpha_i}_0\otimes\mathcal{H}_c^{\alpha_i}\right)\otimes\left(\bigotimes_{j\in J^{(i)}\setminus\{0\}}\mathcal{H}_j^{\alpha_i}\right)$ is represented by a dotted line. Hence e.g. the Hilbert space of the region $\tilde{l}_\inn$ is $\mathcal{H}_{\partial^{\text{left}}_\inn E}^{(\alpha_{\tilde{l}})}$ and of $\tilde{l}_\out$ it is $\mathcal{H}_{\partial^{\text{left}}_\out E}^{(\alpha_{\tilde{l}})}$. Here the superscript \textit{left}, refers to the part of $\partial_\inn E$, which is located in the geometric region $\tilde{l}$, in the figure the left boundary of the region $E$. Respectively, this is the same with the other boundaries. We fix the boundary conditions $(\alpha_{\tilde{l}})$ on region $\tilde{l}$, $(\gamma_l)$ on $l$, $(\beta_r)$ on $r$, and $(\alpha_{\tilde{r}})$ on $\tilde{r}$. } 
    \label{fig:Tree}
\end{figure}

Recall the notation for the Schmidt conditional expectation established in \Cref{subsec:SchmidtCondExp} and equation \eqref{def:tau}.
Fix a boundary condition $(\alpha)=(\alpha_{\tilde{l}},\alpha_{\tilde{r}})=\{\alpha_i\}_{i\in I_{CD}}$ for $CD$, where $(\alpha_{\tilde{l}}):=\{\alpha_i\}_{i\in I_{\partial C\setminus D}}$ labels the boundary of $C$ not in $D$ and $(\alpha_{\tilde{r}}):=\{\alpha_i\}_{i\in I_{\partial D\setminus C}}$ labels the boundary of $D$ not in $C$. Similarly, denote with $(\beta_r):=\{\beta_i\}_{i\in I_{\partial C\cap D}}$ and $(\gamma_l):=\{\gamma_i\}_{i\in I_{\partial D\cap C}}$ the boundaries of $C$ in $D$ and of $D$ in $C$ respectively. For visualization, see also \Cref{fig:Tree}.
Let $0\leq X\equiv X^{(\alpha)}\in \mathbb{L}_{1,\tau^{(\alpha)}_{CD_{\inn}}}$ such that $\|X\|_{1,\tau^{(\alpha)}_{CD_\inn}}=1$. Set
\begin{align}
    N\equiv N^{(\alpha_{\tilde{r}})}:=E_D^{S,(\alpha_{\tilde{r}})}(X)=\bigoplus_{(\gamma_l)}E_D^{S,(\gamma_l,\alpha_{\tilde{r}})}(X) \in \mathcal{B}(\mathcal{H}_{D^{(\alpha_{\tilde{r}})}_\out})\otimes\1_{D^{(\alpha_{\tilde{r}})}_\inn}.
\end{align} By construction it holds that $P^{(\alpha_{\tilde{r}})}NP^{(\alpha_{\tilde{r}})}=N$. Recall that since $D\subset CD$, it follows that $E^S_{CD}\circ E_C^S=E^S_{CD}$ and hence
\begin{align}
    &\Big(E_C^{S,(\alpha)}\circ E_D^{S,(\alpha)}-E_{CD}^{S,(\alpha)}\Big)(X) \\ & \hspace{2cm}= \left(\bigoplus_{(\beta_r)}E_C^{S,(\alpha_{\tilde{l}},\beta_r)}-E_{CD}^{S,(\alpha_{\tilde{l}},\alpha_{\tilde{r}})}\right)(N) \\ & \hspace{2cm}= \bigoplus_{(\beta_r)}\left(P^{(\alpha_{\tilde{l}})}P^{(\beta_r)}\left(\tr_{C_\inn,(\alpha_{\tilde{l}},\beta_r)}\left[\tau_{C_\inn}^{(\alpha_{\tilde{l}},\beta_r)}N\right]\otimes \1_{C_\inn}\right)P^{(\alpha_{\tilde{l}})}P^{(\beta_r)}\right) \\ & \hspace{2.4cm}-  P^{(\alpha_{\tilde{l}})}P^{(\alpha_{\tilde{r}})}\left(\tr_{CD_\inn,(\alpha_{\tilde{l}},\alpha_{\tilde{r}})}\left[\tau_{CD_\inn}^{(\alpha_{\tilde{l}},\alpha_{\tilde{r}})}N\right]\otimes \1_{CD_\inn}\right)P^{(\alpha_{\tilde{l}})}P^{(\alpha_{\tilde{r}})} \\ &\hspace{2cm}=P^{(\alpha_{\tilde{l}})}\left[\bigoplus_{(\beta_r)}P^{(\beta_r)}\left(\tr_{C_\inn,(\alpha_{\tilde{l}},\beta_r)}\left[\tau_{C_\inn}^{(\alpha_{\tilde{l}},\beta_r)}P^{(\alpha_{\tilde{r}})}(N\otimes\1_{D_\inn})P^{(\alpha_{\tilde{r}})}\right]\otimes \1_{C_\inn}\right)P^{(\beta_r)} \right. \\ &\hspace{2.4cm}- \left. \left(\bigoplus_{(\beta_r)}P^{(\beta_r)}\right)P^{(\alpha_{\tilde{r}})}\left(
    \tr_{CD_\inn,(\alpha_{\tilde{l}},\alpha_{\tilde{r}})}\left[\tau_{CD_\inn}^{(\alpha_{\tilde{l}},\alpha_{\tilde{r}})}(N\otimes\1_{D_\inn})\right]\otimes\1_{CD_\inn}
    \right)P^{(\alpha_{\tilde{r}})} \right] P^{(\alpha_{\tilde{l}})}.
\end{align}
In the first line here we used the definition of $X$, in the second the explicit expressions for the Schmidt conditional expectation from \Cref{subsec:SchmidtCondExp}. Then in the third we factor out a common $P^{\alpha_{\tilde{l}}}$ and introduce an identity $\1=\bigoplus_{(\beta_r)}P^{(\beta_r)},$ which commutes with the term just after. We also employ the fact that $N\equiv P^{\alpha_r}N\otimes\1_{D_\inn}P^{\alpha_r}$. Therefore, we can write:
\begin{align}
    &\Big(E_C^{S,(\alpha)}\circ E_D^{S,(\alpha)}-E_{CD}^{S,(\alpha)}\Big)(X)\\
&\quad \quad =\left(P^{(\alpha_{\tilde{l}})}P^{(\alpha_{\tilde{r}})}\bigoplus_{(\beta_r)}P^{(\beta_r)}\right)\tr_{C_\inn\setminus D_\inn,(\alpha_{\tilde{l}})}\Bigg[\Big(\underbrace{\tr_{C_\inn\cap D_\inn,(\beta_r)}\left[\tau_{C_\inn}^{(\alpha_{\tilde{l}},\beta_r)}\right]}_{=:\sigma_2^{(\alpha_{\tilde{l}},\beta_r)}\equiv \sigma_2}-\underbrace{\tr_{D_\inn,(\alpha_{\tilde{r}})}\left[\tau_{CD_\inn}^{(\alpha_{\tilde{l}},\alpha_{\tilde{r}})}\right]}_{=:\sigma_1^{\alpha_{\tilde{l}}}\equiv\sigma_1}\Big)N\Bigg]\times \\
    &\quad \quad \quad \quad \left(P^{(\alpha_{\tilde{l}})}P^{(\alpha_{\tilde{r}})}\bigoplus_{(\beta_r)}P^{(\beta_r)}\right) \equiv \bigoplus_{(\beta_r)}P^{(\alpha_{\tilde{l}},\beta_r,\alpha_{\tilde{r}})}\left[\tr_{C_\inn\setminus D_\inn}\left[(\sigma_2-\sigma_1)N\right]\right]P^{(\alpha_{\tilde{l}},\beta_r,\alpha_{\tilde{r}})},
\end{align}
where $\sigma_2\equiv\sigma_2^{(\alpha_{\tilde{l}},\beta_r)}:=\tr_{C_\inn\cap D_\inn,(\beta_r)}\left[\tau_{C_\inn}^{(\alpha_{\tilde{l}},\beta_r)}\right]$ and $\sigma_1\equiv\sigma_1^{(\alpha_{\tilde{l}})}:=\tr_{D_\inn,(\alpha_{\tilde{r}})}\left[\tau_{CD_\inn}^{(\alpha_{\tilde{l}},\alpha_{\tilde{r}})}\right]$.
Here, we factored out the projections and rearrange the partial traces suitably. The last equality is then just introducing a simplifying notation. 
For simplicity we suppress the boundary conditions $(\alpha_{\tilde{l}})$ index on the states. 
Hence
\begin{align}
    \left\|\left(E_C^{S(\alpha)}\circ E_D^{S(\alpha)}-E_{CD}^{S(\alpha)}\right)(X)\right\|&\leq \max_{(\beta_r)}\|\Tr_{C_\inn\setminus D_\inn,(\alpha_{\tilde{l}})}\left[(\sigma_2-\sigma_1)N\right]\| \\ &= \max_{(\beta_r)}\left|\Tr_{C_{\inn}\setminus D_{\inn}, (\alpha_{\tilde{l}})}[(\sigma_2-\sigma_1)N]\right| \\ &= \max_{(\beta_r)}|\Tr_{C_{\inn}\setminus D_{\inn}, (\alpha_{\tilde{l}})}[(\1-\sigma_1^{-\frac{1}{2}}\sigma_2\sigma_1^{-\frac{1}{2}})(\sigma_1^\frac{1}{2}N\sigma_1^\frac{1}{2})]| \\ &\leq \max_{(\beta_r)}\|(\1-\sigma_1^{-\frac{1}{2}}\sigma_2\sigma_1^{-\frac{1}{2}})(\sigma_1^\frac{1}{2}N\sigma_1^\frac{1}{2})]\|_1 \\ &\leq \max_{(\beta_r)}  \|\sigma_1^{-\frac{1}{2}}\sigma_2\sigma_1^{-\frac{1}{2}}-\1\|_\infty \|N\|_{\mathbb{L}_{1,\sigma_1}},
\end{align} 
where the equality in the second line follows since $N=E_D^{S(\gamma_l,\alpha_{\tilde{r}})}(X)$ is the identity on the complementary Hilbert space to $(\mathcal{H}_{C_{\inn}\setminus D_{\inn}}^{(\alpha_{\tilde{l}})})$. In the last inequality we used Hölder and the definition of the $\mathbb{L}_{1,\sigma_1}$ norm.
By definition of $N$, we have
\begin{align}
    \|N\|_{\mathbb{L}_{1,\sigma_1}} &= \|\sigma_1^{\frac{1}{2}}N\sigma_1^{\frac{1}{2}}\|_1 \overset{N\geq0}{=} \Tr[\sigma_1N] = \Tr[\tau_{CD_\inn}^{(\alpha_{\tilde{l}},\alpha_{\tilde{r}})}N] = \Tr[\tau_{CD_\inn}^{(\alpha)}E_D^{S(\alpha)}(X)] \\ &= \Tr[E^{S(\alpha)}_{D*}(\tau_{CD_\inn}^{(\alpha)})X]=\Tr[\tau_{CD_\inn}^{(\alpha)}X] \overset{X\geq 0}{=} \|X\|_{\mathbb{L}_{1,\tau_{CD_\inn}^{(\alpha)}}} = 1.
\end{align}
Hence, to prove the theorem, we need to establish that $\sigma_1\overset{\epsilon}{\sim}\sigma_2$ for any boundary condition $(\alpha_{\tilde{l}},\beta_r,\alpha_{\tilde{r}})$. We will do it with $l\mapsto\epsilon(l)$ an exponentially decreasing function in $l=\dist(C\setminus D,D\setminus C)$ with decay length $\xi$. The two states can be written as
\begin{align}
    \sigma_1 &:=\tr_{D_\inn,(\alpha_{\tilde{r}})}\left[\tau_{CD_\inn}^{(\alpha_{\tilde{l}},\alpha_{\tilde{r}})}\right] = \frac{1}{\Tr[...]}\tr_{\partial_\out (CD)\cup D_\inn,(\alpha_{\tilde{r}})}\left[P^{(\alpha_{\tilde{l}})}P^{(\alpha_{\tilde{r}})}\sigma^{(CD\partial)}P^{(\alpha_{\tilde{l}})}P^{(\alpha_{\tilde{r}})}\right] \\ &= \frac{1}{\Tr[...]}\tr_{\tilde{l}_\out l_\inn D\tilde{r}}\left[P^{(\alpha_{\tilde{r}})}P^{(\alpha_{\tilde{l}})}\sigma^{(CD\partial)}P^{(\alpha_{\tilde{l}})}\right] \in \mathcal{B}(\mathcal{H}_{C_\inn\setminus D_\inn}^{(\alpha_{\tilde{l}})}), \\
    \sigma_2 &:=\tr_{F_\inn,(\beta_r)}\left[\tau_{C_\inn}^{(\alpha_{\tilde{l}},\beta_r)}\right] = \frac{1}{\Tr[...]}\tr_{\partial_\out (C)\cup F_\inn,(\beta_r)}\left[P^{(\alpha_{\tilde{l}})}P^{(\beta_r)}\sigma^{(C\partial)}P^{(\alpha_{\tilde{l}})}P^{(\beta_r)}\right] \\ &= \frac{1}{\Tr[...]}\tr_{\tilde{l}_\out l_\inn Fr}\left[P^{(\beta_r)}P^{(\alpha_{\tilde{l}})}\sigma^{(C\partial)}P^{(\alpha_{\tilde{l}})}\right] \in \mathcal{B}(\mathcal{H}_{C_\inn\setminus D_\inn}^{(\alpha_{\tilde{l}})}).
\end{align}
Both are full-rank states on $\mathcal{H}_{C_\inn\setminus D_\inn}^{(\alpha_{\tilde{l}})}$ and thus have the same support.
The intuition now is as follows: since we have a Gibbs state that satisfies exponential decay of covariance, strong local indistinguishability, and the mixing condition (see \Cref{thm:ddimStrongLocalIndist}), these two states should be approximately the same in the bulk, where we compare them. This is because they might only differ significantly on $D\setminus C$, but we look at them on $C \setminus D$. 

From the assumptions of the Theorem we have that the family of Gibbs states satisfy decay of covariance and hence by \Cref{thm:ddimStrongLocalIndist} also strong local indistinguishability and the mixing condition. We start by applying the mixing condition from \eqref{equ:StrongTensorization} to each of the states, which guarantees the existence of two exponentially decaying functions $\epsilon_1,\epsilon_2$ in $\dist(l,\tilde{r})>\dist(C\setminus D,D\setminus C)$ and $\dist(l,r)=\dist(C\setminus D,D\setminus C)$, respectively, s.t.
\begin{align}
    &\tr_{l_\inn D}\sigma^{CD\partial} \overset{\epsilon_1}{\sim} \tr_{l_\inn D\tilde{r}}\sigma^{CD\partial}\otimes\tr_{\tilde{l}ElD}\sigma^{CD\partial}  \\&\overset{\autoref{prop:UsefulRelation}}{\implies} 
    \tr_{\tilde{r}}P^{(\alpha_{\tilde{r}})}\tr_{l_\inn D}\sigma^{CD\partial}P^{(\alpha_{\tilde{r}})} \overset{\epsilon_1}{\sim} \tr_{l_\inn D\tilde{r}}\sigma^{CD\partial}\Tr[P^{(\alpha_{\tilde{r}})}\sigma^{CD\partial}P^{(\alpha_{\tilde{r}})}], \\
    &\tr_{l_\inn F}\sigma^{C\partial} \overset{\epsilon_2}{\sim} \tr_{l_\inn Fr}\sigma^{C\partial}\otimes\tr_{\tilde{l}ElF}\sigma^{C\partial}  \\&\overset{\autoref{prop:UsefulRelation}}{\implies} 
    \tr_rP^{(\beta_r)}\tr_{l_\inn F}\sigma^{C\partial}P^{(\beta_r)} \overset{\epsilon_2}{\sim} \tr_{l_\inn Fr}\sigma^{C\partial}\Tr[P^{(\beta_r)}\sigma^{C\partial}P^{(\beta_r)}].
\end{align} Here the implication follows from Proposition \ref{prop:UsefulRelation} 4) applied to the projections $P^{(\alpha_{\tilde{r}})},P^{(\beta_r)}$ respectively. Now, by strong local indistinguishability (\autoref{thm:ddimStrongLocalIndist}) there exists an exponentially decaying function $\epsilon_3$ in $\dist(l,r)=\dist(C\setminus D,D\setminus C)$, s.t.
\begin{align}
    \tr_{l_\inn D\tilde{r}}\sigma^{CD\partial} = \tr_{l_\inn FrG\tilde{r}}\sigma^{CD\partial} \overset{\epsilon_3}{\sim} \tr_{l_\inn Fr} \sigma^{C\partial}
\end{align}
and from Theorem \ref{thm:ddimStrongLocalIndist} one gets that uniform decay of covariance with decay function, say $\epsilon(\dist)=Ke^{-\frac{\dist}{\xi}}$, implies that $\epsilon_1\propto\exp{\mathcal{O}(|\tilde{r}|+|\partial(El)|)}\epsilon(\dist)$, $\epsilon_2\propto\exp{\mathcal{O}(|r|+|\partial(El)|)}\epsilon(\dist)$, and $\epsilon_3\propto\exp{\mathcal{O}(|\partial(El)|)}\mathcal{O}(|r|)\epsilon(\dist)$.
Hence by transitivity and symmetry of the strong similarity relation, see \Cref{prop:UsefulRelation} 1) and 2), it follows that 
\begin{align}
    \tr_{l_\inn D\tilde{r}}\left[P^{(\alpha_{\tilde{r}})}\sigma^{CD\partial}P^{(\alpha_{\tilde{r}})}\right] &\overset{\epsilon_1}{\sim} \tr_{l_\inn D\tilde{r}}\sigma^{CD\partial}\Tr[P^{(\alpha_{\tilde{r}})}\sigma^{CD\partial}P^{(\alpha_{\tilde{r}})}] \overset{\epsilon_3}{\sim} \tr_{l_\inn Fr} \sigma^{C\partial}\Tr[P^{(\alpha_{\tilde{r}})}\sigma^{CD\partial}P^{(\alpha_{\tilde{r}})}] \\
    &\overset{\epsilon_2(1-\epsilon_2)^{-1}}{\sim} \tr_{l_\inn Fr}\left[P^{(\beta_r)}\sigma^{C\partial}P^{(\beta_r)}\right]\frac{\Tr[P^{(\alpha_{\tilde{r}})}\sigma^{CD\partial}P^{(\alpha_{\tilde{r}})}]}{\Tr[P^{(\beta_r)}\sigma^{C\partial}P^{(\beta_r)}]},
\end{align} 
and hence
\begin{align}
    \widetilde{\sigma_1}:= \frac{\tr_{l_\inn D\tilde{r}}\left[P^{(\alpha_{\tilde{r}})}\sigma^{CD\partial}P^{(\alpha_{\tilde{r}})}\right]}{\Tr[P^{(\alpha_{\tilde{r}})}\sigma^{CD\partial}P^{(\alpha_{\tilde{r}})}]} \overset{\tilde{\eta}}{\sim} \frac{\tr_{l_\inn Fr}\left[P^{(\beta_r)}\sigma^{C\partial}P^{(\beta_r)}\right]}{\Tr[P^{(\beta_r)}\sigma^{C\partial}P^{(\beta_r)}]} =: \widetilde{\sigma_2},
\end{align} where $\tilde{\eta}$ is exponentially decreasing in $\dist(C\setminus D, D\setminus C)$ with decay length $\xi$ and prefactor scaling as $\exp(\mathcal{O}(|\partial (C\setminus D)|+|\partial D\setminus C)|))$, since $\epsilon_1,\epsilon_2,\epsilon_3$ are. By Proposition \ref{prop:UsefulRelation} $4^\prime)$ it follows that this strong similarity of the states also holds within each block $(\alpha_{\tilde{l}})$:
\begin{align}
    \sigma_1= \frac{1}{\Tr[...]}\tr_{\tilde{l}_\inn}\left[P^{\alpha_{\tilde{l}}}\widetilde{\sigma_1}P^{\alpha_{\tilde{l}}}\right] \overset{\eta:=\tilde{\eta}(2+\tilde{\eta})}{\sim} \frac{1}{\Tr[...]}\tr_{\tilde{l}_\inn}\left[P^{\alpha_{\tilde{l}}}\widetilde{\sigma_2}P^{\alpha_{\tilde{l}}}\right] = \sigma_2.
\end{align} This establishes the bound for any boundary condition $(\alpha_{\tilde{l}},\beta_r,\alpha_{\tilde{r}})$, which concludes the proof.
Note that this also implies that the exponential decay rate of $\eta$ is the same as the one of the decay of covariance we assumed, which is the correlation length of the Gibbs state $\xi$.
\end{proof}

\section{Main results}\label{sec:main}

In this section we give sufficient conditions for rapid thermalization of Davies evolutions with unique fixed point the Gibbs state of a nearest-neighbour, commuting Hamiltonian. 

The results essentially show that a suitable q$\mathbb{L}_1\to\mathbb{L}_\infty$ decay (Definition \ref{equ:q1toinfty_clustering}), together with some geometric features of the lattice, implies the existence of a constant or log decreasing cMLSI constant $\alpha(\mathcal{L}^D_\Gamma)$, which directly implies rapid mixing. The main technical lemma is as follows.

\begin{lemma}\label{thm:technicalMainThm}
    Let $\Lambda$ be a 2-colorable graph with finite growth constant, $\Phi$ a uniformly bounded, nearest-neighbour, commuting potential, and $\beta>0$ some inverse temperature. If the family of Schmidt generators associated to $(\Lambda,\Phi,\beta)$ satisfies an exponentially decaying q$\mathbb{L}_1\to\mathbb{L}_\infty$ decay function $\eta_{C,D}(l)$ in $l=\dist(C\setminus D,D\setminus C)$ whenever picking the regions $C,D$ convex and such that 
    \begin{itemize}
        \item[\phantom{a}]  \underline{\text{Case 1:}} $\diam(C),\diam(D)=\mathcal{O}(l^{2})$, 
        \item[\phantom{a}] \underline{\text{Case 2:}} $\diam(C),\diam(D)=\mathcal{O}(l)$, 
    \end{itemize}
 then a family of Davies generators associated to $(\Lambda,\Phi,\beta)$ satisfies the MLSI with
    \begin{itemize}
        \item[\phantom{a}] \underline{\text{Case 1:}} a system size independent MLSI constant $\alpha(\mathcal{L}^D_\Gamma)=\Omega(1)_{|\Gamma|\to\infty}$,
        \item[\phantom{a}] \underline{\text{Case 2:}} a linear in system diameter decaying MLSI constant $\alpha(\mathcal{L}^D_\Gamma)=\Omega((\diam(\Gamma))^{-1})_{|\Gamma|\to\infty}$.
    \end{itemize}
\end{lemma}
The proof is deferred to  \Cref{subsec:proofs}. 
The intuition behind separating these two cases is as follows.  The bound on the MLSI requires certain choices of coarse-grainings of the lattice into overlapping regions $C,D$. The choice of how large those regions are as compared to their overlap is limited both by the geometry of the lattice, and the decay of $\eta_{C,D}(l)$. \underline{Case 1} is the one where it is possible to choose the overlap to be a vanishing fraction of the regions, leading to a better scaling, while in \underline{Case 2} the overlap is a finite fraction of the region. The diameter is the relevant quantity here, because the decay of correlations is a function of the one dimensional distance between two subregions and not the number of sites in the overlap in general.

In the next two subsections, Theorem \ref{thm:technicalMainThm} is then applied to two types of decay of $\eta_{C,D}(l)$, and for each of them we consider two different lattice geometries, corresponding to the two cases of the theorem. \underline{\text{Case 1:}} will apply to the 1D quantum spin chain under gap and among others the $D$-dimensional hypercubic lattices at high temperature. On the other hand, \underline{\text{Case 2:}} will apply to the 2D hypercubic lattice under the gap condition and $b$-ary trees at high enough temperature. In particular:

\begin{itemize}
    \item In \Cref{subsec:Gap} we consider the case where the q$\mathbb{L}_1\to\mathbb{L}_\infty$ decay is implied by the existence of a unformly bounded strictly positive gap $\lambda(\mathcal{L}^D_\Gamma)=\Omega(1)_{|\Gamma|\to\infty}$ of the Davies Lindbladians established via the  clustering results \Cref{thm:ddimStrongLocalIndist} and \Cref{thm:weaktostrongclustering}, and \cite[Corollary 27]{art:QuantumGibbsSamplers-kastoryano2016quantum}.  This shows rapid thermalization with a constant decay rate of commuting 1D quantum systems here. This is visualised in \Cref{fig:Visualization}. We also show that the same argument can also be applied to 2D systems. There, we obtain a better bound on the thermalization time than one would naïvely get from the gap, but still too slow to be rapidly mixing.
    
\item In \Cref{subsec:Temp}, we consider the high temperature setting, where the q$\mathbb{L}_1\to\mathbb{L}_\infty$ is guaranteed to decay fast, see \Cref{thm:clusteringFromHighTemp}. This setting notably includes $D$-dim hypercubic lattices and $b$-ary trees, for which we hence show rapid mixing at high enough temperatures. 
\end{itemize}

\begin{figure}
    \centering
    \includegraphics[width=\linewidth]{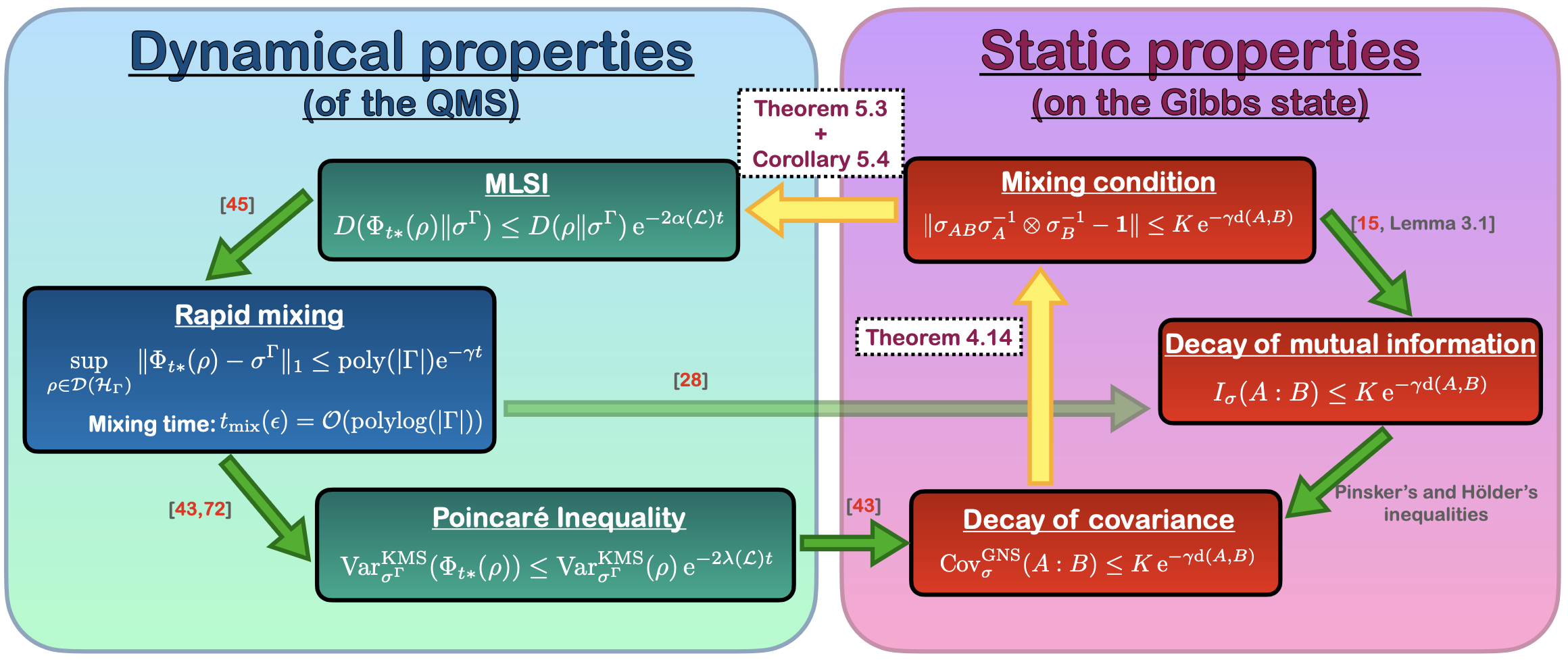}
    \caption{Relation between the main different notions of clustering (static properties on Gibbs states) and their connection to thermalization of Davies QMS (dynamical properties) in 1D systems. For those systems, we show that system size invariant gap is equivalent to a system size independent MLSI constant via the visualised chain of implications. Yellow arrows signify implications proved in this work.}
    \label{fig:Visualization}
\end{figure}

\subsection{Rapid Thermalization from Gap}\label{subsec:Gap}

Here we show two direct applications of \Cref{thm:technicalMainThm}. The first is that for 1D quantum spin chains with commuting nearest-neighbour Hamiltonians, having a gap is a sufficient (and necessary) condition for rapid thermalization. We also show how in 2D we obtain an improvement over previous results on mixing. Before stating both results, we recall what was previously established in Section \ref{sec:static}. 

Let $\Lambda$ be a graph with finite growth constant, $\Phi$ a uniformly bounded, nearest-neighbour, commuting potential and $\beta$ a suitable inverse temperature, such that the family of Davies Lindbladians associated to $(\Lambda,\Phi,\beta)$ has a positive spectral gap, i.e. $\inf_{\Gamma\subset\subset\Lambda}\lambda(\mathcal{L}^D_\Gamma)>0$ . Then, for any $\Gamma \subset \subset \mathbb{Z}$ or $\mathbb{Z}^2$ and $C,D \subset \Gamma$ with $C \cap D \neq \emptyset$,  by \cite[Corollary 27]{art:QuantumGibbsSamplers-kastoryano2016quantum}, \Cref{thm:ddimStrongLocalIndist}, and \Cref{thm:weaktostrongclustering} we have q$\mathbb{L}_1\to\mathbb{L}_\infty$-decay with decay rate 
\begin{align}\label{eq:etaexp}
    \eta_{\text{gap},C,D}(l)=\exp{\mathcal{O}(|\partial(C\setminus D)|+|\partial(D\setminus C)|)}\mathcal{O}\left(\exp\left(-\frac{\dist(C\setminus D,D\setminus C)}{\xi}\right)\right), 
\end{align}where $\xi$ is the thermal correlation length of the family of Gibbs states associated to $(\Lambda,\Phi,\beta)$. 
Together with this scaling, two judicious choices of regions $C,D$ (one for each case in Theorem \ref{thm:technicalMainThm}) and their corresponding overlap allows us to obtain the following proposition. 
\begin{proposition} \label{prop:important2}
    It holds that $\eta_{\text{gap},C,D}(l)$ is exponentially decaying in $l$ when picking $C,D$ convex with $\dist(C\setminus D,D\setminus C)=l$ and such that 
        \begin{itemize}
        \item[\phantom{a}] \underline{\text{Case 1:}} For $\Lambda=\mathbb{Z}$, we have $\diam(C),\diam(D)=\mathcal{O}(l^2)$  ,
        \item[\phantom{a}] \underline{\text{Case 2:}} For $\Lambda=\mathbb{Z}^2$, and the thermal correlation length $\xi$ is $\xi < \xi_{\text{max}}(\Phi,\beta)= \mathcal{O}(1)$, we have $\diam(C),\diam(D)=\mathcal{O}(l)$
    \end{itemize}
\end{proposition}
\begin{proof}
    For $\Lambda=\mathbb{Z}$, the boundary of convex sets is $2$, since they are intervals of the chain. Since this is a constant number, the decay holds independently of the correlation length $\xi$. 
   For $\Lambda=\mathbb{Z}^2$, the boundary of convex sets scales linearly in their diameter. Hence so do $|\partial(C\setminus D)|,|\partial(D\setminus C)|$ and thus, if the correlation length $\xi$ is small enough to dominate over the linear growth in the exponent, we get exponential decay of $l\mapsto\eta_{\text{gap},C,D}(l)$.
\end{proof}

Together with \Cref{thm:technicalMainThm} this directly yields the main result of the section.

\begin{theorem} \label{thm:mainGap}
The Davies generator $\Li^D_\Lambda:=\{\mathcal{L}^D_\Gamma\}_{\Gamma\subset\subset\Lambda}$ corresponding to a uniform, nearest-neighbour, commuting family of Hamiltonians acting on the locally-finite dimensional quantum spin system $\mathcal{H}_\Lambda$  satisfies a MLSI with constant 
\begin{enumerate}
    \item[$1)$] $\alpha(\mathcal{L}^D_\Gamma) = \Omega(1)_{|\Gamma|\to\infty}$ independent of system size, when $\Lambda=\Z$ is the spin-chain, or
    \item[$2)$] $\alpha(\mathcal{L}^D_\Gamma) = \Omega\left((\sqrt{|\Gamma|})^{-1}\right)_{|\Gamma|\to\infty}$ square-root decreasing in system size, when $\Lambda=\mathbb{Z}^2$ and the correlation length of the fixed point (Gibbs state) is $\xi < \xi_{\text{max}}(\Phi,\beta)= \mathcal{O}(1)$,
\end{enumerate}
whenever the family of thermal states satisfies uniform exponential decay of covariance or when the Davies Lindbladians are uniformly gapped, i.e.  $\inf_{\Gamma\subset\subset\Lambda}\lambda(\mathcal{L}^D_\Gamma)>0$, which implies the former.
\end{theorem}

The square root decay of the MLSI constant in the 2-dimensional setting in Theorem \ref{thm:mainGap} does not imply rapid mixing, since it only guarantees a polynomial thermalization time. However, it still improves the bound on the mixing time over the one we would get only from having a gap, i.e. we get $t_\text{mix}(\epsilon)=\mathcal{O}(\sqrt{|\Gamma|}\log|\Gamma|,\log\frac{1}{\epsilon})$, as opposed to an arbitrary $t_\text{mix}(\epsilon)=\mathcal{O}(\text{poly}(\vert \Gamma \vert),\log\frac{1}{\epsilon})$.

Concerning the 1D result, we first of all note that
we do not require $r=2$, i.e. nearest-neighbour interactions assumption, since by coarse-graining, we can always map a geometrically-$r$-local 1D system to a nearest neighbour one.\footnote{This, however, only works in the 1D case. The problem in higher dimensions is that the coarse grained graph is no longer 2-colorable or 2-colorable under further coarse-graining.}
This result has important implications in the 1D setting, 
where Theorem \ref{thm:mainGap} not just gives rapid mixing of the Davies dynamics but also an optimal scaling of the cMLSI constant from either gap or exponential decay of covariance.

A priori, assuming gap is the stronger assumption, since in \cite{art:QuantumGibbsSamplers-kastoryano2016quantum}, it was proved that existence of a system-size independent gap implies $\mathbb{L}_2$-clustering, and thus exponential decay of covariance (or $\mathbb{L}_\infty$-clustering). 
However, for the 1D systems considered here,\footnote{In fact, this holds more generally for uniformly bounded, commuting, geometrically-local systems on any graph with a finite growth constant.} it is known that rapid thermalization implies existence of a striclty positive gap of the generator \cite{art:temme2015fast,art:QuantumGibbsSamplers-kastoryano2016quantum}. This means that in 1D we have shown equivalence of strictly positive gap of generators, decay of covariance of the invariant state, MLSI with uniformly strictly positive constant, and rapid thermalization. 
Therefore, at least in 1D quantum spin chains,
this answers in the affirmative an open question from \cite{art:QuantumGibbsSamplers-kastoryano2016quantum}, namely whether for commuting systems on $\mathbb{Z}$, existence of exponential decay of covariance is sufficient to prove existence of a spin-system size invariant strictly positive spectral gap for $\mathcal{L}^D_\Gamma$. This is due to the fact that, as a combination of \Cref{thm:ddimStrongLocalIndist} and \Cref{thm:mainGap}, exponential decay of covariance of a thermal state of a uniform geometrically-local, commuting Hamiltonian implies a strictly positive and system size invariant MLSI constant, which implies gap. We refer the reader to Figure \ref{fig:Visualization} for a more clear depiction of these implications. 
We thus highlight the strength of our result for the 1D setting in the following corollary.


\begin{corollary}[1D constant MLSI at any temperature and locality]\label{thm:main1D}
Let $\Lambda=\mathbb{Z}$ be the one-dimensional quantum spin chain endowed with a uniformly $J$-bounded, geometrically-$r$-local, commuting potential $\Phi$ for some $r\in\mathbb{N}$ and let $\beta>0$ be any inverse temperature. 
There exists a strictly positive constant $\alpha$ independent of the spin-chain length that lower bounds the MLSI constant of any element of the Davies generators associated to $(\mathbb{Z},\Phi,\beta)$, i.e. $\alpha(\Li^D_\Gamma)=\Omega(1)_{|\Gamma|\to\infty}$. Hence 
for any finite $\Gamma\subset\subset\Lambda$ we have 
\begin{align}
    D(\rho_t\|\sigma^\Gamma) \leq e^{-t\alpha(\mathcal{L}_\Gamma^D)}D(\rho_0\|\sigma^\Gamma),
\end{align}
and thus also rapid thermalization with mixing time $t_\text{mix}(\epsilon)=\mathcal{O}\left(\log|\Gamma|,\log\frac{1}{\epsilon}\right)$.
\end{corollary}

This follows from the recent result in \cite{kimura2024clustering}, where it is shown that exponential decay of covariance holds for such systems with correlation length $\xi=\exp\mathcal{O}(\beta)$  and any finite temperature $\beta>0$. Alternatively, for translation-invariant systems, a uniform lower bounded gap is known by \cite{art:QuantumGibbsSamplers-kastoryano2016quantum}. 
We do not require $r=2$, i.e. nearest-neighbour interactions, since by coarse-graining, we can always map a geometrically-$r$-local 1D system to a nearest neighbour one.\footnote{This, however, only works in the 1D case. The problem in higher dimensions is that the coarse grained graph is no longer 2-colorable or 2-colorable under further coarse-graining.}

Corollary \ref{thm:main1D} is a strict improvement over the best previous results \cite{art:EntropyDecayOf1DSpinChain-Cambyse,art:ImplicationsAndRapidTermalization-Cambyse}, where $\alpha(\Li^D_\Gamma)=\Omega(\log|\Gamma|^{-1})$ is logarithmically decreasing in the system size for translation-invariant Hamiltonians. With our result, we obtain the same system-size independence as with the so-called LSI constant in the classical setting, which is known to be optimal.

\subsection{Rapid Thermalization from High Temperature}\label{subsec:Temp}

In this section we consider rapid thermalization as a consequence of being at high temperature, as given in Theorem \ref{thm:clusteringFromHighTemp}. Being at high temperature means that we can assume uniform-q$\mathbb{L}_1\to\mathbb{L}_\infty$-clustering with decay function 
\begin{align}
    \eta_{C,D;\textup{thermal}}(l) = \mathcal{O}(|C\cup D|)\exp\left(-\frac{\dist(C\setminus D,D\setminus C)}{\xi^\prime}\right). \label{Condition2}
\end{align}
See \cite{art:2localPaper}[Theorem 6,7, Proposition 2] for the proof.  The fact that we now have a linear dependence $\mathcal{O}(|C\cup D|)$ of the pre-factor as opposed to the exponential one in Eq. \eqref{eq:etaexp} means that we can define more favourable choices of regions $C,D$ in the proofs, thus yielding stronger results. This is of particular importance for $D>1$, where we prove rapid mixing under the high temperature assumption.

We now present the following condition on the decay length $\xi^\prime$, which will feature in the main theorem below.
\begin{definition}[Geometric condition on decay length] 
For an infinite graph $\Lambda$, recall $N(l):=\sup_{x\in\Lambda}|B_l(x)|$, where $B_l(x):=\{v\in\Lambda|\dist(x,v)\leq l\}$ is the ball around $x$ of radius $l$. 
We require the decay length $\xi^\prime$ of the uniform-q$\mathbb{L}_1\to\mathbb{L}_\infty$-clustering to satisfy 
\begin{align}
    \xi^\prime < \frac{l}{2\log N(l)}
    \label{def:treecondition}
\end{align} eventually in $l$, i.e. for $l\geq l_0$ for some $l_0\in\mathbb{N}$. 
\end{definition}

The decay length $\xi^\prime$ is not just a function of the graph alone, but also of the potential $\Phi$ and the inverse temperature $\beta$. If condition \eqref{def:treecondition} holds, then $N(l)\exp\left(-\frac{l}{\xi^\prime}\right)$ (and in turn, $\eta_{C,D}(l)$) is at least exponentially decaying in $l$.

Recall that hypercubic lattices are sub-exponential graphs, whereas $b-$ary trees are exponential.  Hence for infinite hypercubic lattices of dimension $D$ we have that $N(l)\propto l^D$ and thus this condition is trivially fulfilled for any $\xi^\prime>0$. For $b$-ary trees we have that $N(l)=\sum_{k=0}^lb^k=\frac{b^{l+1}-1}{b-1}\propto b^l$, and hence this condition becomes
$\xi^\prime<\frac{1}{2\log b}$, which is an implicit condition on the temperature $\beta^{-1}$.  Since it holds trivially for $\beta=0$, it should also hold for a small enough constant $\beta>0$.

Having established this, we now have the following proposition, similarly to Proposition \ref{prop:important2}. 
\begin{proposition} \label{prop:important3}
    The function $\eta_{C,D;\textup{thermal}}(l)$ is exponentially decaying in $l$ when choosing $C,D$ convex such that 
    \begin{itemize}
         \item[\phantom{a}] \underline{\text{Case 1:}} For $\Lambda=\mathbb{Z}^D$ a $D-$dim hypercubic lattice, $\diam(C),\diam(D)=\mathcal{O}(l^{2})$.  \footnote{This also holds for any subexponential graph, when requiring that we still have exponential decay when on can choose $\diam(C),\diam(D)=\mathcal{O}(l^{1+\delta})$ for any $0<\delta\leq 1$. The proof still goes through under this slightly weakend 'Case 2' condition, hence the main result directly extends from hypercubic to subexponential graphs.}  
           \item[\phantom{a}] \underline{\text{Case 2:}}   For $\Lambda$ exponential graph, $\diam(C),\diam(D)=\mathcal{O}(l)$ and the the decay length $\xi^\prime$ is small enough (see e.g. Condition \eqref{def:treecondition}).
    \end{itemize}
\end{proposition}
\begin{proof} Let $\Lambda$ be a $D$-dim hyper-cubic graph. Then $|C\cup D|$ scales at worst as $(\diam(C)+\diam(D))^D= \mathcal{O}(l^{2D})$, so $\eta_{C,D;\textup{thermal}}(l)\leq\mathcal{O}(l^{2D})\exp(-l\xi^{-1})$. This is eventually exponentially decaying.  \\
For exponential graphs there exists a finite growth constant, say $b$, such that we have that $|C\cup D|\leq b^{\diam(C)+\diam(D)}=\exp(\mathcal{O}(l))$, and thus $\eta_{C,D;\textup{thermal}}(l)\leq\exp(\mathcal{O}(l)-l\xi^{\prime -1})$ is exponentially decaying when $\xi^\prime$ is small enough. In the proof of Lemma \ref{lemma:OmegacMLSI} we will only be requiring exponential decay for sets $C,D$ such that $\diam(C\cup D)=2l=2\dist(C\setminus D,D\setminus C)$. In this case we have that condition \eqref{def:treecondition} is sufficient to guarantee the claimed exponential decay, as it holds that $|C\cup D|\leq b^{\diam(C\cup D)}\leq\exp(2l\log b)$, and thus $\eta_{C,D;\textup{thermal}}(l)\leq\exp(2l\log b-l\xi^{\prime -1})$.
\end{proof}

This allows us to apply Theorem \ref{thm:technicalMainThm} directly. In this section, as opposed to Proposition \ref{prop:important2} above, \underline{Case 1} will involve generic sub-exponential lattices, and \underline{Case 2} will be exponential trees. With those geometries, and given the choice of regions that Proposition \ref{prop:important3} allows for, the arguments from Sec. \ref{subsec:proofs} below allow us to obtain the following main result in the high temperature setting.

\begin{theorem}[Davies MLSI and rapid mixing from high temperature] \label{thm:mainTemp}
Let $\Lambda$ be a 2-colorable graph with finite growth constant $\nu$ and $\Phi$ a uniformly $J-$bounded, nearest-neighbour, commuting potential on it. If the temperature is high enough, i.e. $\beta^{-1}>10eJ\nu$, the Davies generators $\Li^D_\Lambda:=\{\mathcal{L}^D_\Gamma\}_{\Gamma\subset\subset\Lambda}$ associated to $(\Lambda,\Phi,\beta)$ acting on the locally-uniformly-finite dimensional quantum spin system $\mathcal{H}_\Lambda$ satisfy a MLSI with constant 
\begin{enumerate}
    \item[$1)$] $\alpha(\mathcal{L}^D_\Gamma) = \Omega(1)_{|\Gamma|\to\infty}$ independent of system size, when $\Lambda$ is a sub-exponential graph, such as all $\mathbb{Z}^D$, or
    \item[$2)$] $\alpha(\mathcal{L}^D_\Gamma) = \Omega\left((\log|\Gamma|)^{-1}\right)_{|\Gamma|\to\infty}$ logarithmically decreasing in system size, when $\Lambda$ is an exponential graph and the correlation length of the fixed point (Gibbs state) satisfies condition \eqref{def:treecondition}. For tree graphs, this condition is $\xi^\prime<(2\log b)^{-1}$.
\end{enumerate}
\end{theorem}
This means that in both cases the dynamics generated by the Davies generators are rapidly mixing. Importantly, the result also holds when we have a polynomial dependence on $|C\cup D|$ instead of the linear prefactor of $\eta_{C,D;\textup{thermal}}(l)$. The requirement on the high temperature comes solely through \Cref{thm:clusteringFromHighTemp}, hence if one could prove \eqref{Condition2} from some other starting point, such as having a gap, then one would equally obtain Theorem \ref{thm:mainTemp}. 

The result constitutes a strict extension of the main Theorem of \cite{art:2localPaper}, since we are able to prove rapid thermalization for Davies dynamics, as opposed to for the more artificial setting of Schmidt dynamics. For the sub-exponential graph such as hypercubic lattices, the constant bound on the MLSI yields an optimal scaling of the thermalization time of the semi-group. It should be possible to extend it to wider ranges of temperatures, as long as the right notion of decay of correlations holds.

The result for trees is, to the authors knowledge, the first of its kind in the quantum setting. Classically it is known, however, that the exponential decay rate of the relative entropy towards the equilibrium is tree-size independent \cite{art:ClassicaTreesMartinelli_2004}. We expect this to also hold in the quantum case, since the proof there hinges upon a very analogous condition on the temperature as the one here, which is implicit through $\xi^\prime<(2\log b)^{-1}$.

\subsection{Proof of main results}\label{subsec:proofs}

In this section we prove \Cref{thm:technicalMainThm}. The structure is as follows. We first reproduce an important result from \cite{art:2localPaper} that tells us that exponential q$\mathbb{L}_1\to\mathbb{L}_\infty$ decay implies an approximate tensorization statement in \Cref{thm:ApproxTensorizationForOmega}. This result, the work in \cite{art:2localPaper}, and Lemma \ref{lemma:SchmidtandDavies} are the ingredients we need to establish the main result \ref{thm:mainTemp} for quantum systems on hypercubic latices of dimension $D>1$. 
The  main part of the proof in \Cref{subsec:DivideandConquer} then consists of a geometric argument where we apply the aforementioned approximate tensorization to show the main result separately for \underline{Case 1} and \underline{Case 2}. This is done by proving the result for the two prototypical instances of both cases: quantum spin chains and $b$-ary trees, respectively. The derivations then generalise straightforwardly to all the instances of Cases 1 and 2. 

Denote with $\eta(l)\equiv\eta_{C,D}(\dist(C\setminus D,D\setminus C))$ the decay function of the q$\mathbb{L}_1\to\mathbb{L}_\infty$-clustering of the family of Schmidt generators associated to $(\Lambda,\Phi,\beta)$.
In what follows, it is important that, by assumption, in \underline{Case 1} we have that $\eta(l)$ is exponentially decaying when picking $C,D$ convex s.t. $\diam(C),\diam(D)=\mathcal{O}(l^2)$ and in \underline{Case 2} that $\eta(l)$ is exponentially decaying when picking $C,D$ convex s.t. $\diam(C),\diam(D)=\mathcal{O}(l)$. The two cases, as well as their consequences as MLSIs for the various models, are depicted in Figure \ref{fig:proof_mainthm}. 

\begin{figure}[H]
    \centering
    \includegraphics[width=\linewidth]{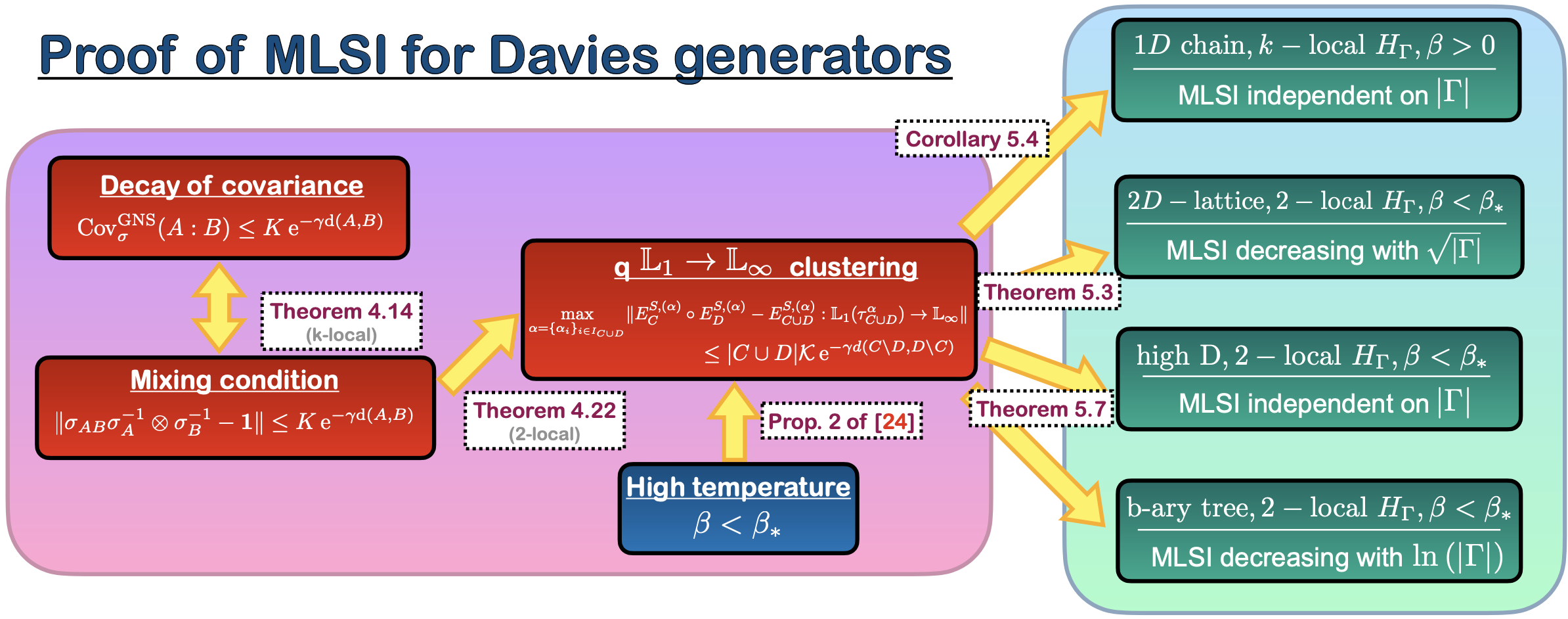}
    \caption{Sketch of the proof of \Cref{thm:technicalMainThm}. In a first part, we recall the results from Section \ref{sec:static} that showed the equivalence between clustering conditions in the case of nearest neighbour interactions. In a second part, at the right side of the diagram, we derive the various cases of \Cref{thm:technicalMainThm} as a combination of several lemmas in a geometric argument to derive positive MLSI constants for several classes of Hamiltonians.}
    \label{fig:proof_mainthm}
\end{figure}

\subsubsection{Approximate tensorization for an almost classical state $\omega$}\label{subsec:ApproxTensorization}

In this step of the proof, we connect the exponential decay from uniform q$\mathbb{L}_1\to\mathbb{L}_\infty$-clustering from Definition \ref{def:LinfinityClustering} with a decomposition of a relative entropy of a conditional expectation on some region with respect to that of smaller overlapping regions, which we refer to as ``approximate tensorization".

Given a 2-colorable graph $\Lambda$ and a fixed 2-coloring, we denote the set of vertices with labels 0 as $\Lambda_0$ and the set of vertices with labels 1 as $\Lambda_1=\Lambda\setminus\Lambda_0$.

\begin{definition}\label{def:OmegaState} 
Given the 2-colorable graph $\Lambda$ with 2-coloring $\Lambda = \{\Lambda_0,\Lambda_1\}$, denote a finite subgraphs $\Gamma\subset\subset\Lambda$ with induced 2-colorings $\{\Gamma_i:=\Gamma\cap\Lambda_i\}_{i=0}^1=:\{\Gamma_0,\Gamma\setminus \Gamma_0\}$. For a quantum state $\rho\in\mathcal{D}(\mathcal{H}_\Gamma)$, we define 
\begin{equation}\label{def:omegastate}
    \omega := E_{\Gamma_0*}^S(\rho) = (\bigcirc_{a\in\Gamma_0}E_{a*}^S)(\rho),
\end{equation}
where the second equality follows from Proposition \ref{prop:commutativityofSC}. 
\end{definition}
The importance of this ``semiclassical" state is that in \cite{art:2localPaper} it was shown to satisfy the following approximate tensorization statement from the q$\mathbb{L}_1\to\mathbb{L}_\infty$-decay we assume here. 


\begin{theorem}[Approximate tensorization for $\omega$]\label{thm:ApproxTensorizationForOmega}
Let $(\Lambda,\Phi,\beta)$ be as assumed in this section. Fix some $\Gamma\subset\subset\Lambda$ and given some $\rho\in\mathcal{D}(\mathcal{H}_\Gamma)$ set $\omega:=E^S_{A*}(\rho)$ as defined above. Let $C,D\subset\Gamma$ be two convex subsets, s.t. for $\#\in\{R:=C\cup D$, $\partial C, \partial D, \partial R\}$ we have $ \#\cap \Lambda_0=\emptyset,$ and $l:=\dist(C\setminus D,D\setminus C)>1$. Then it holds that
\begin{equation}\label{equ:approximateTensorization}
    D(\omega\|E_{C\cup D*}^S(\omega)) \leq \frac{1}{1-2\eta_{C,D}(l)}[D(\omega\|E_{C*}^S(\omega))+D(\omega\|E_{D*}^S(\omega))],
\end{equation}
where $\eta(l)$ is the decay function of the uniform q$\mathbb{L}_1\to\mathbb{L}_\infty$-clustering. 
\end{theorem}
The proof of this approximate tensorization statement, given either \Cref{thm:weaktostrongclustering} or \Cref{thm:clusteringFromHighTemp}, is just an application of \cite[Theorem 8]{art:2localPaper}. 
This theorem will allow us to execute what is sometimes called a ``divide-and-conquer" strategy for the state $\omega$. That is, approximate tensorization allows us to bound the relative entropy distance between $\omega$ and the projection onto its fixed point on some lattice region $CD$ by the ones on the smaller regions $C$ and $D$, up to a suitable factor that depends on the decay of correlations of the fixed point. We will then iterate this procedure to end up with finite size regions, on which we can bound the relative entropies by a local MLSI constant, and finally generalize this to arbitrary states $\rho$.

From this we will show that, as long as we have exponential decay in this approximate tensorization argument for convex regions $C,D$ with the right diameter, as per the two cases, we can establish a lower bound on the MLSI. In \textit{\underline{Case 1}}, with diameter $\mathcal{O}(l^2)$, we can establish a constant lower bound on the MLSI constant $\alpha(\mathcal{L}^D_\Gamma)=\Omega(1)_{|\Gamma|\to\infty}$ and for \textit{\underline{Case 2}}, where the convex regions $C,D$ have diameters $\mathcal{O}(l)$, we can establish a logarithmic (for trees at high temperature) and square root (for 2D lattices also at high temperature) lower bound on the MLSI constant $\alpha(\mathcal{L}^D_\Gamma)=\Omega((\diam|\Gamma|)^{-1})_{|\Gamma|\to\infty}$. To get there, we will require a repeated application of Theorem \ref{thm:ApproxTensorizationForOmega} with appropriate choices of regions $C,D$, as we show in the next subsection.


\subsubsection{Geometric argument}\label{subsec:DivideandConquer}

We now construct the aforementioned geometric divide-and-conquer-strategy employing a Cesaro averaging over different choices of regions $C,D$. 
The target is an upper bound on the relative entropy between the state $\omega$ and its image under the Schmidt conditional expectation on the whole $\Gamma$. This bound is based on Theorem \ref{thm:ApproxTensorizationForOmega}, and will be expressed in terms of the sum of the relative entropies between $\omega$ and its image under the Schmidt conditional expectation on smaller \textit{coarse-grained} sets of a fixed size, called $\{R_k\}_k$.

For the geometric argument, we will explicitly define the construction of the relevant coarse-grained-sets for the $b$-ary tree, for $b\in\mathbb{N}$. The $1$-D case is then covered as the $b=1$ case in trees. Since the Cesaro averaging and clustering works analogously for $b$-ary trees and $D$-dim hypercubic lattices, we will for simplicity only conduct it for the former. For the definition of the subsets and further details on the $D$-dimensional hypercubic construction see Section \ref{sec:hypercube} and also \cite{art:2localPaper}.

Denote with $K:=\{k|x_k\in\Gamma_0\}$ the index set of $\Gamma_0=\Gamma\cap\Lambda_0$. 
The $\{R_k\}_{k}$ will be of a sufficiently large but finite minimal size and s.t. $E^S_{R_k*}\circ E^S_{\Gamma_0*}=E^S_{\Gamma_0*}\circ E^S_{R_k*}$ holds, which is always the case if their boundaries satisfy $\partial R_k\cap\Gamma_0=\emptyset$. As $\Gamma_0$ is the union of single vertices each with distance 2 from each other, each $x_j\in\Gamma_0$ is either an element of $R_k$, or $\dist(x_j,R_k)=2$, so the commutation of the conditional expectations follows from Proposition \ref{prop:commutativityofSC}. For trees, we define the subsets as follows \footnote{Note that the constructions are not unique.}. 

\begin{definition}[Construction of coarse-grained sets for $b$-ary tree]
Denote the infinite $b$-ary tree with $\mathbb{T}_b$. Denote with $B_{x,l}$ the subtree rooted at site $x\in\mathbb{T}_b$ of height $l$. 
We define the following set of subsets
\begin{align}\label{def:R-k-sets}
    \{R_k\}_{k\in K}:=\{B_{x_k,l_0}\cap\Gamma\}_{k\in K}=\{B_{x,l_0}\cap\Gamma\}_{x\in \Gamma_0},
\end{align}where $l_0\in2\mathbb{N}$ is a a suitably large constant to be fixed later on. The $\{R_k\}_k$ form a coarse-graining into subtrees based on each vertex of the same label (i.e. in $\Gamma_0$) of finite fixed height $l_0$. We consider the cMLSI constant of our evolution on these sets, and extend these via a Cesaro averaging argument to the whole lattice.
\end{definition}

Through these sets, we define the following quantity with respect to the state $\omega$ from \eqref{def:omegastate}
\begin{align}
    D_R(\omega):=\sum_{R_k\subset R}D(\omega\|E^S_{R_k*}(\omega)),
\end{align} where the sets $\{R_k\}_{k\in K}$ are the coarse grained sets as specified above. Observe that $D_R(\omega)$ is monotonically increasing in $R$, i.e. if $A\subset B\subset\Gamma$ are two subregions, then $D_A(\omega)\leq D_B(\omega)$ by positivity of the relative entropy. It is also additive up to boundary terms, i.e. if $A,B\subset\Gamma$ are two disjoint subregions, then 
\begin{align}\label{eq:additiverels}
D_{AB}(\omega)=D_A(\omega)+D_B(\omega)+\sum_{\underset{R_k\cap A\neq \emptyset}{R_k\cap B\neq \emptyset}}D(\omega\|E^S_{R_k*}(\omega)).
\end{align}
Next we define a function $L\mapsto C(L):\mathbb{N}\to\mathbb{R}$, such that $C(L)$ is the smallest number for which 
\begin{align}\label{def:C-Funktion}
    D(\omega\|E^S_{\Gamma*}(\omega)) \leq C(L)D_{\Gamma}(\omega),
\end{align} where $L:=\height(\Gamma)$. It follows from this definition that $C(L)$ is monotonically non-decreasing. The main target of this section is thus to upper bound $C(L)$. These bounds are given in the following key lemma, which is proven by inductively applying the approximate tensorization and averaging suitably over the choices of partitions $C,D$. It is the main ingredient of the geometric part of the proof. 
\begin{lemma}\label{lemma:OmegacMLSI} For sub-trees of $\mathbb{T}_b$ of the form $\Gamma=B_{x_j,L}$ for some $j\in K$, \eqref{def:C-Funktion} holds with  
\begin{enumerate}
       \item[\phantom{a}] \underline{\text{Case 1:}}  $C(L)=\mathcal{O}(1)_{|\Gamma|\to\infty}$ is uniformly upper bounded by $C(\infty)<\infty$.
      \item[\phantom{a}] \underline{\text{Case 2:}}  $C(L)=\mathcal{O}(L)_{L\to\infty}=\mathcal{O}(\diam|\Gamma|)_{|\Gamma|\to\infty}=\mathcal{O}(\log|\Gamma|)_{|\Gamma|\to\infty}$.
\end{enumerate}
\end{lemma}

This lemma can be viewed as an approximate tensorization statement on many regions for a state of the form of $\omega$. That is, it upper bounds the relative entropy between it and its Schmidt conditional expectation on the whole lattice $\Gamma$ in terms of the relative entropies of it and the Schmidt conditional expectations of the fixed finite size regions $\{R_k\}_{k\in K}$. \\



In order to prove it, we will make a repeated use of the approximate tensorization result Theorem \ref{thm:ApproxTensorizationForOmega}. First breaking up a region $B_{x_k,L}$ into the following regions, then breaking these up in a recursive way until we reach a decomposition in terms of the initial defined $\{R_k\}_k$. To apply it, we now proceed with the choice of regions $C,D$ to be used in the result. 
For a set $B_{x_k,L}$ define 
\begin{align}
    C_k^{\tilde{l}} := B_{x_k,\tilde{l}}  \, , \hspace{1cm} D_k^{\tilde{l},l} := \bigcup_{\underset{\dist(x_m,x_k)=\tilde{l}-l}{m\in K}} B_{x_m,L+l-\tilde{l}} \, .
\end{align}
Hence we cover the subtree $B_{x_k,L}$ of height $L$ by a subtree with the same root of height $\tilde{l}$, called $C_k^{\tilde{l}}$, and a union of disjoint subtrees of height $L+l-\tilde{l}$, called $D_k^{\tilde{l},l}$, s.t. their overlap has height $l$ and is s.t. we can apply the approximate tensorization result \Cref{thm:ApproxTensorizationForOmega} with the function $\eta_{C_k^{\tilde{l}},D_k^{\tilde{l},l}}(l)$.  

Importantly we require that each of these sets \textit{begins} and \textit{ends} with some vertices of the same index 0, i.e. in $\Gamma_0$.
In formulae, this is $C_k^{\tilde{l}}\cup D_k^{\tilde{l},l}=B_{x_k,L}$ and $\dist\big(C_k^{\tilde{l}}\setminus D_k^{\tilde{l},l},D_k^{\tilde{l},l}\setminus C_k^{\tilde{l}}\big)=l$ for all $0\leq\tilde{l}\leq L$.
Hence each set $B_{x_k,L}$ has the family of non-trivial partitions $\{C_k^{\tilde{l}},D_k^{\tilde{l},l}\}_{\tilde{l}=1}^{L-1}$ and for each of these it holds, due to \Cref{thm:ApproxTensorizationForOmega}, that
\begin{align}
    D\big(\omega \big\|E^S_{B_{x_k,L}*}(\omega)\big) \leq \frac{1}{1-2\eta_{C_k^{\tilde{l}},D_k^{\tilde{l},l}}(l)}\left[D\Big(\omega \Big\|E^S_{C_k^{\tilde{l}}*}(\omega)\Big)+D\Big(\omega \Big\|E^S_{D_k^{\tilde{l},l}*}(\omega)\Big)\right], \label{equ:approxtens}
\end{align}
For an example of these regions see \Cref{fig:Partition}.

\begin{figure}[h]
    \centering
    \includegraphics[width=0.6\linewidth]{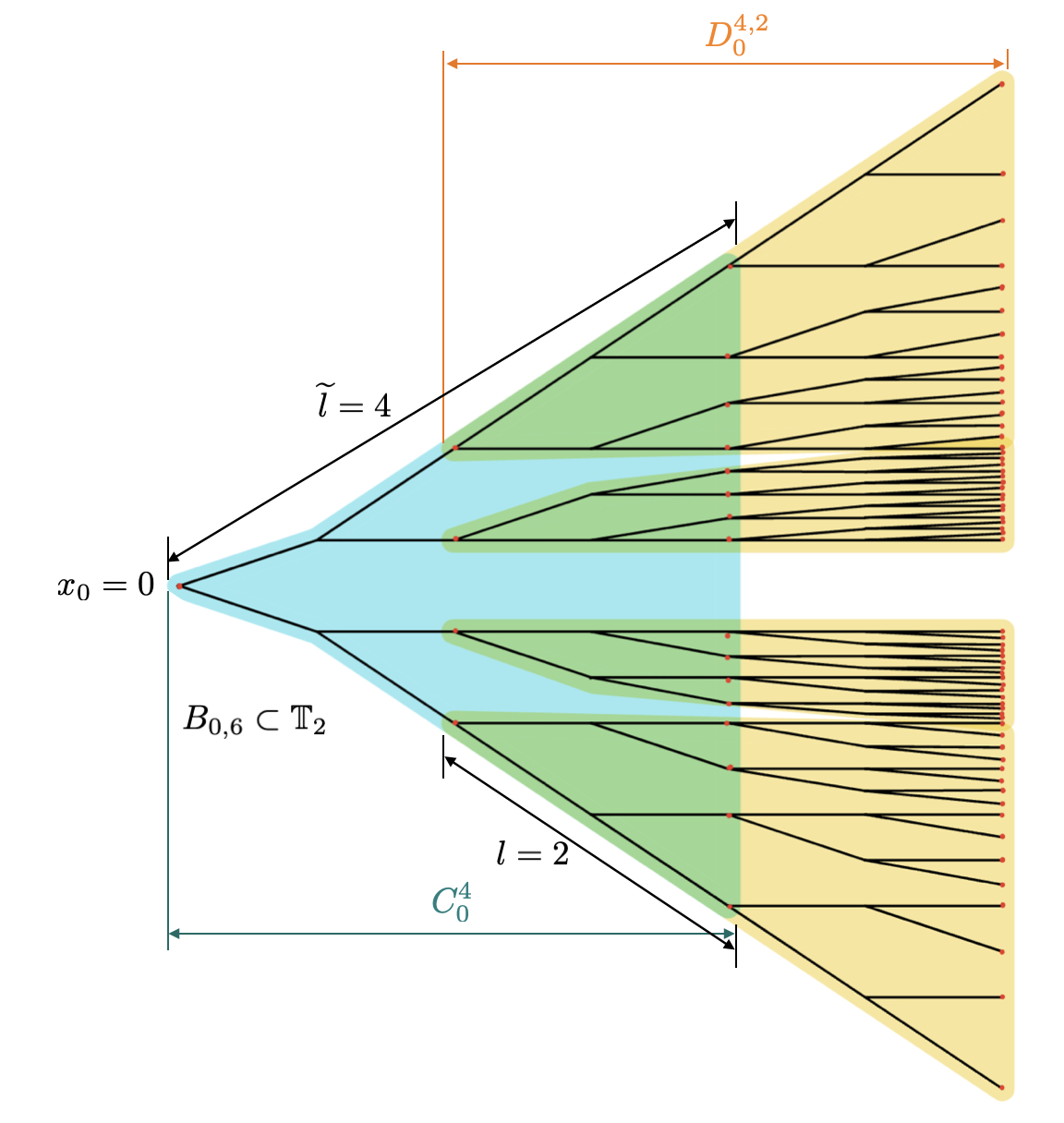}
    \caption{Example of a partition of $B_{0,6}\subset\mathbb{T}_2$ into the regions {\color{Turquise}$C^{\tilde{l}}_{0}$}, {\color{RoyalYellow}$D^{\tilde{l},l}_{0}$}  with height of the regions $\tilde{l}=4 $ and $  l=6-\tilde{l}+2=4$ and the height of their overlap being $l=2$. This is in $B_{0,6}$ as part of a 2-ary tree. The {\color{red} red vertices}  are the ones of index 0. Notice that each of these sets `begins' and `ends' in these sets.}
    \label{fig:Partition}
\end{figure}

In the $1D$ setting, these can be written as $C_k^{\tilde{l}}=[x_k,x_k+\tilde{l}]$ and $D_k^{\tilde{l},l}=[x_k+\tilde{l}-l,x_k+L]$, where the $\height$ becomes $\diam$, which is just the length of the intervals.

Now, recalling the definition of the two cases we are considering, we have the following key observation.
\begin{enumerate}\label{enumerate2}
       \item[\phantom{a}] \underline{\text{Case 1:}}  We can pick the overlap $l=\floor{\sqrt{L}}= \mathcal{O}(\sqrt{L})$. Then $\eta(l)\equiv\eta_{C_k^{\tilde{l}},D_k^{\tilde{l},l}}(l)$ 
    is exponentially decaying in $L$ for any $\xi>0$, by assumption. Set $l_{\text{min},1}$ to be the smallest $L$, s.t. $\eta(l)<\frac{1}{2}$ for all $l\geq l_{\text{min},1}$.
      \item[\phantom{a}] \underline{\text{Case 2:}}  We have to pick the overlap $l=\frac{L}{N}=\mathcal{O}(L)$ for some $N\in\mathbb{N}\setminus\{1\}$ and require the decay length ($\xi,\xi^\prime$) to be small enough such that we have eventual exponential decay in $\eta(l)\equiv\eta_{C_k^{\tilde{l}},D_k^{\tilde{l},l}}(l)$ in $l$. This suffices for $N=2$, but we will keep it general in the following derivation. 
    We set $l_0$ to be the smallest $l$ such that $\eta(l)<\frac{1}{2}$ for all $l\geq l_0$.
\end{enumerate}
We now write the proof for each of these cases separately.

\begin{proof}[Proof of \underline{Case 1} in Lemma \ref{lemma:OmegacMLSI}] 
The proof idea is to average the approximate tensorization result, \Cref{thm:ApproxTensorizationForOmega}, over all the above defined coverings $\{C_k^{\tilde{l}},D_k^{\tilde{l},l}\}_{\tilde{l}}$, for different $\tilde l$, to get the relative entropy between $\omega$ and its Schmidt conditional expectation on the whole of $B_{x_k,L}$ in terms of the relative entropy between it and the Schmidt conditional expectation on subregions of height $\epsilon L$, where $\epsilon<1$.  
This allows us to bound $C(L)$ in terms of the smaller $C(\epsilon L)$ times a multiplicative factor, which allows us to ultimately bound $C(L)$ by a constant through a recursive process. Hence to bound $C(L)$ by $C(l_0)$, where $l_0$ is a fixed finite size, requires $\mathcal{O}(\log(L))$ steps. However, since in \underline{Case 1} the decay is assumed strong enough, we will be able to upper bound this product uniformly by a constant.

Fix $x_j\in\Gamma_0$ and let $\frac{1}{2}<\epsilon<1$. 
We enumerate a maximal set of partitions of $B_{x_j,L}$ into $\{C_j^{\tilde{l}},D_j^{\tilde{l},\floor{\sqrt{L}}}\}$, s.t. $\height(C_j^{\tilde{l}})=\tilde{l},\height(D_j^{\tilde{l},\floor{\sqrt{L}}})=L+\floor{\sqrt{L}}-\tilde{l}\leq \epsilon L$ and s.t. different partitions have disjoint overlaps, i.e. $\left(C_j^{\tilde{l}_1}\cap D_j^{\tilde{l}_1,\floor{\sqrt{L}}}\right)\cap \left(C_j^{\tilde{l}_2}\cap D_j^{\tilde{l}_2,\floor{\sqrt{L}}}\right) = \emptyset$, whenever $\tilde{l}_1\neq \tilde{l}_2$.
This works as long as $\sqrt{L}\lesssim(2\epsilon-1)L$ which gives another condition on the minimal size of $l=\floor{\sqrt{L}}\geq: l_{\text{min},2}$. There exist $\frac{L}{\floor{\sqrt{L}}}=\mathcal{O}(\sqrt{L})$ of these partitions, since their overlap is of height $\floor{\sqrt{L}}=\mathcal{O}(\sqrt{L})$. To simplify notation, we refer to these partitions $\{C_i,D_i\}^{\mathcal{O}(\sqrt{L})}_{i=1}$, where the partition index $i$ is not to be confused with the fixed root $j$. We now average over all the approximate tensorization results of these partitions to get
\begin{align}
    D(\omega\|E^S_{B_{x_j,L}*}(\omega)) &\leq \frac{1}{\mathcal{O}(\sqrt{L})}\sum_{i=1}^{\mathcal{O}(\sqrt{L})} \frac{1}{1-2\eta_{C_i,D_i}(\sqrt{L})}\left[D(\omega\|E^S_{C_i*}(\omega))+D(\omega\|E^S_{D_i*}(\omega)) \right] \\
        &\leq \frac{1}{1-2\eta(\sqrt{L})}\frac{1}{\mathcal{O}(\sqrt{L})}\sum_{i=1}^{\mathcal{O}(\sqrt{L})} C(\height(C_i))D_{C_i}(\omega)+C(\height(D_i))D_{D_i}(\omega) \\
    &\leq C(\epsilon L)\frac{1}{1-2\eta(\sqrt{L})}\frac{1}{\mathcal{O}(\sqrt{L})}\Bigg(\sum_{i=1}^{\mathcal{O}(\sqrt{L})}\left(2D_{C_i\cap D_i}(\omega)+D_{C_i\setminus D_i \cup D_i\setminus C_i}(\omega)\right).\\ &\quad \quad +\sum_{\underset{R_k\cap (C_i\setminus D_i)\neq \emptyset}{R_k\cap (C_i\cap D_i)\neq \emptyset}}D(\omega\|E^S_{R_k*}(\omega))+\sum_{\underset{R_k\cap (D_i\setminus C_i)\neq \emptyset}{R_k\cap (C_i\cap D_i)\neq \emptyset}}D(\omega\|E^S_{R_k*}(\omega)) \Bigg) \\
    &\leq C(\epsilon L)\frac{1}{1-2\eta(\sqrt{L})}\frac{1}{\mathcal{O}(\sqrt{L})}\left(2D_{\bigcup_{i=1}^{\mathcal{O}(\sqrt{L})}(C_i\cap D_i)}(\omega)+\sum_{i=1}^{\mathcal{O}(\sqrt{L})}D_{(C_i\setminus D_i)\partial\cup (D_i\setminus C_i)\partial}(\omega)\right) \\
    &\leq C(\epsilon L)\frac{1}{1-2\eta(\sqrt{L})}\frac{\floor{\sqrt{L}}}{L}\left(2+\frac{L}{\floor{\sqrt{L}}}\right)D_{B_{x_j,L}}(\omega),
\end{align} 
where in the second line we used that $\eta_{C_i,D_i}(\sqrt{L})$ does not depend on $i$, and then the definition of $C(\height(\cdot))$. In the third line we used that $C(\height(C_i)),C(\height(D_i))\leq C(\epsilon L)$ since $\height(C_i)=\tilde{l},\height(D_i)=L+\floor{\sqrt{L}}-\tilde{l}\leq \epsilon L$, and also used Eq. \eqref{eq:additiverels}.
Hence it follows that $C(L)\leq C(\epsilon L)\frac{1}{1-2\eta(\sqrt{L})}(1+\frac{2}{\sqrt{L}})=:C(\epsilon L)f(L)$. This allows us to bound $C(L)$ through an iterative process, since by definition in Eq. \eqref{def:C-Funktion} and the sets $\{R_k\}_{k\in K}$, $C(l_0)=1$. Using $C(L)\le C(\epsilon L)f(L)$,   $M=\mathcal{O}(\log L)$ times, s.t. $\epsilon^M L=l_0=:\ceil{\max\{l_{\text{min},1},l_{\text{min},2}\}}$ then gives
\begin{align}
    C(L)\leq C(l_0)\prod_{k=1}^Mf(\epsilon^kL)\leq C(l_0)\prod_{k=0}^\infty f(l_0\epsilon^{-k})<\infty,
\end{align}
where it can be checked by inspection that the infinite product converges to a constant since $\eta(\sqrt{L})$ 
is exponentially decaying in $\sqrt{L}$ and $(1+\frac{2}{\sqrt{L}})\to 1$ fast enough. Note that $C(l_0)=1$, independent of the $x_j$ which we fixed initially. Hence the result follows.
\end{proof}
\begin{proof}[Proof of \underline{Case 2} in Lemma \ref{lemma:OmegacMLSI}]
The proof follows exactly the same idea and techniques as the one above. The main difference is that we need to choose the overlap of the coverings $C,D$ to scale as $\mathcal{O}(L)$, where $L=\height{B_{x_j,L}}$, to still get decay in the approximate tensorization. Hence, the number of partitions has to be constant in system size, which yields a constant multiplicative factor in each inductive step. Since we again need $\mathcal{O}(\log(L))$ steps in the iteration to bound $C(L)$, this gives the $\mathcal{O}(L)$, where $L=\height(B_{x_k,L})$ of the original set, scaling of the bound on $C(L)$. The details of this proof can be found in \Cref{sec:proof_2ndMartinelli} for completeness.
\end{proof}

After this key lemma on approximate tensorization, in Sec. \ref{sec:puttingtogether} below we will generalise this statement to arbitrary states $\rho$ and the Davies maps instead of the Schmidt conditional expectation. This will then allow us to extend the existence of a local cMLSI constant, see \Cref{thm:finiteregioncMLSI}, to the whole lattice with the cost of $C(L)^{-1}$. 


\subsubsection{Extension to hypercubic lattices}\label{sec:hypercube}

The main idea behind the proof of Lemma \ref{lemma:OmegacMLSI}, where we  average over all suitable partitions of a given set to obtain the desired upper bound, naturally extends to hypercubic lattices ($D>1$). To adapt it, we need to adap the definitions of both the coarse-grained sets $\{R_k \}_k$ and the `suitable partitions'. This then yields the second case of Theorem \ref{thm:mainGap} and the first case of Theorem \ref{thm:mainTemp}.

First, for the $\{R_k \}_k$, instead of the definition in equation \eqref{def:R-k-sets} we have the following definition.
\begin{definition}[Construction of Coarse-grain sets for $D$-dim hypercubic lattice]\label{def:hypercubes}
Denote with $x+A$ the set $A\subset\mathbb{Z}^D$ shifted by $x\in\mathbb{Z}^D$. Then, define the sets 
\begin{align}
    \{R_k\}_{k\in K}:=&\{(x_k+[-l_0-1,l_0-1]^D\cup(\Gamma_0\cap(x_k+[l_0,l_0]^D)))\cap \Gamma\}_{k\in K}  \\=& \{(x+[-l_0-1,l_0-1]^D\cup(\Gamma_0\cap(x+[l_0,l_0]^D)))\cap \Gamma\}_{x\in \Gamma_0}, \label{def:R_kcubes}
\end{align} where again $l_0\in\mathbb{N}$ is a suitably large constant to be fixed later on. These are essentially jagged hyper-cubes of side-length $2l_0+1$ around a center $x_k\in\Gamma_0$.  An example for one of these jagged hyper-cubes in 2D is shown in \Cref{fig:ameba}.
\end{definition}

\begin{figure}[H]
    \centering
    \includegraphics[scale=0.3]{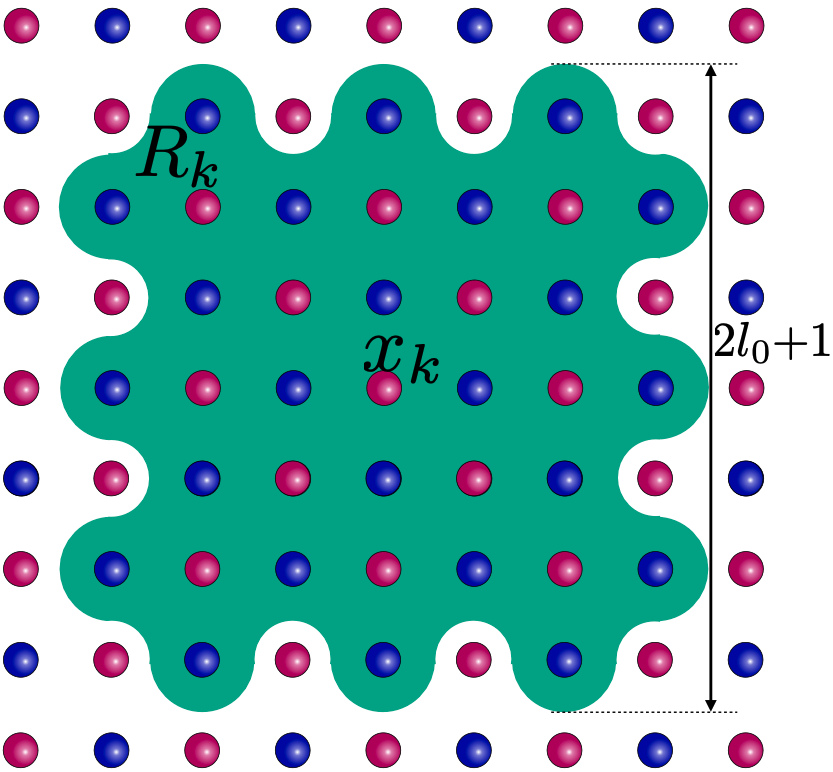}
    \caption{Example of a coarse-grain set from \Cref{def:hypercubes} in 2D. Here, we consider the \textcolor{purpledots}{$\Gamma_0$}, sites of index 0, to be the sites of \textcolor{purpledots}{purple} color, and \textcolor{bluedots}{$\Gamma_1$} those of \textcolor{bluedots}{blue}. The boundary of \textcolor{darkgreen2}{$R_k$} is completely composed of sites in purple. }
    \label{fig:ameba}
\end{figure}

For the suitable partitions $C_k,D_k$, the idea is to take jagged hyper-rectangles with sufficient overlap as the sets `$C_k,D_k$'. Then, we apply the approximate tensorization statement  $D $ times to get \eqref{equ:approxtens}, but with $2^D$ terms on the right hand side, where the conditioned relative entropy is on sub-regions with a strictly smaller diameter. The approximate tensorization statement we obtain then yields, through iterative application as in the proof of Lemma \ref{lemma:OmegacMLSI}, the analogous following Lemma.
\begin{lemma}\label{lemma:OmegacMLSIcubes} For sub-lattices of $\mathbb{Z}^D$ of form $x_j+[-L,L]^D$ for some $j\in K$, \eqref{def:C-Funktion} holds, with the  coarse-grained sets defined in \eqref{def:R_kcubes}, with  
\begin{enumerate}
     \item[\phantom{a}] \underline{\text{Case 1:}}  $C(L)=\mathcal{O}(1)_{|\Gamma|\to\infty}$ is uniformly upper bounded by a $C(\infty)<\infty$.
     \item[\phantom{a}] \underline{\text{Case 2:}}  $C(L)=\mathcal{O}(L)_{L\to\infty}=\mathcal{O}(\diam|\Gamma|)_{|\Gamma|\to\infty}=\mathcal{O}(|\Gamma|^{\frac{1}{D}})_{|\Gamma|\to\infty}$.
\end{enumerate}
\end{lemma}
For the $2$-D graph in \underline{Case 2}, which is the only use of this \underline{Case 2} in our main results, this gives us
$C(L)=\mathcal{O}(L)_{L\to\infty}=\mathcal{O}(\sqrt{|\Gamma|})_{|\Gamma|\to\infty}$. \underline{Case 1} instead yields the constant MLSI at high temperature. 

\subsubsection{Putting everything together}\label{sec:puttingtogether}
In the last steps, we generalise the result above for the state $\omega=E_{\Gamma_0}^S(\rho)$ and the Schmidt conditional expectations to arbitrary states $\rho$ and to the Davies conditional expectations. This finalises the bounds on the MLSI for the Davies generators.

\begin{proof}[Proof of \Cref{thm:technicalMainThm} (and hence of \cref{thm:mainGap} and \Cref{thm:mainTemp})]

Let $(\Lambda,\Phi,\beta)$ be as assumed in \Cref{thm:technicalMainThm} 
, $\Lambda = \{\Lambda_0,\Lambda_1\}$ be a two colouring, and $\{E^S_\Gamma\}_{\Gamma\subset\subset\Lambda}, \{E^D_\Gamma\}_{\Gamma\subset\subset\Lambda}$ the families of Schmidt and Davies conditional expectations. 
Let $\Gamma\subset\subset\Lambda$\footnote{e.g. for the $b-$ary tree, $b\in\mathbb{N}$, wlog $\Gamma=B_{0,L}$.} be a complete connected finite subgraph, where we set $\{x_k\}_{k\in K}=\Gamma_0:=\Gamma\cap\Lambda_0$, and set $\omega:=E^S_{\Gamma_0}(\rho)$ for a state $\rho\in\mathcal{D}(\mathcal{H}_\Gamma)$.  
We first apply the chain rule for the relative entropy \eqref{def:relentChainRule}, with $\sigma\equiv\sigma^\Gamma=E^D_{\Gamma*}(\rho)=E^S_{\Gamma_0*}(\sigma)$, so that
\begin{align}
    D(\rho\|E^D_{\Gamma*}(\sigma))=D(\rho\|\sigma) = D(\rho\|E^S_{\Gamma_0*}(\rho)) + D(E^S_{\Gamma_0*}(\rho)\|\sigma) = D(\rho\|\omega)+ D(\omega\|\sigma). 
\end{align}
The first summand $D(\rho\|\omega)$ satisfies exact tensorization (a form of strong subadditivity), since the Schmidt conditional expectations of two sets with distance two between each other commute, $E^S_{\{x_k\}*}\circ E^S_{\{x_j\}*}= E^S_{\{x_j\}*}\circ E^S_{\{x_k\}*} = E^S_{\{x_k\}\cup \{x_j\}*}$ for all $x_k,x_j\in\Gamma_0$. See Proposition \ref{prop:commutativityofSC}  and \cite{art:CompleteEntropicInequalities_GaoRouze_2022, art:Petz_RelEntropyAdditivity_1991}. This means that
\begin{align}
    D(\rho\|\omega)=D(\rho\|E^S_{\Gamma_0*}(\rho)) \leq \sum_{x_k\in\Gamma_0}D(\rho\|E^S_{\{x_k\}*}(\rho)) \overset{\text{Lemma }\ref{lem:RelEntCondExpBound}}{\leq}  \sum_{x_k\in\Gamma_0}D(\rho\|E^D_{\{x_k\}\partial *}(\rho)).
\end{align}
We can also bound the second summand $D(\omega\|\sigma)$ as follows, using Lemma \ref{lemma:OmegacMLSI}, the DPI for the relative entropy, and finally Lemma \ref{lem:RelEntCondExpBound}.
\begin{align}
    D(\omega\|\sigma) = D(\omega\|E^S_{\Gamma*}(\omega)) &\overset{\text{Lemma }\ref{lemma:OmegacMLSI}}{\leq} C\sum_{x_k\in\Gamma_0}D(\omega\|E^S_{R_k*}(\omega)) = C\sum_{x_k\in\Gamma_0}D(E^S_{\Gamma_0*}(\rho)\|(E^S_{R_k*}\circ E^S_{\Gamma_0*})(\rho)) \\ &\overset{DPI}{\leq} C\sum_{x_k\in\Gamma_0}D(\rho\|E^S_{R_k*}(\rho)) \overset{\text{Lemma }\ref{lem:RelEntCondExpBound}}{\leq} C\sum_{x_k\in\Gamma_0}D(\rho\|E^D_{R_k\partial*}(\rho)),
\end{align} where in the second line we used that $E^S_{R_k*}\circ E^S_{\Gamma_0*}=E^S_{\Gamma_0*}\circ E^S_{R_k*}$, which holds by the construction of the sets $R_k$ in \eqref{def:R-k-sets}.\footnote{This is since $\partial R_k\cap\Gamma_0=\emptyset$. Recall that $\Gamma_0$ is the union of single vertices each with distance 2 from each other. Hence the claim follows from Proposition \ref{prop:commutativityofSC}.}
Recall that by Lemma \ref{lemma:OmegacMLSI} in \textit{\underline{Case 1}}  the constant $C$ is independent of system size, whereas in \textit{\underline{Case 2}} for trees it scales logarithmically with system size.
Importantly, the regions $\{\{x_k\}\partial,R_j\partial\}_{k,j}$ are of fixed finite size. Hence by \Cref{thm:finiteregioncMLSI} there exists cMLSI constants $\alpha_0,\alpha_1>0$, s.t. for any $j,k$
\begin{align}
    \alpha_0D\left(\rho\|E^D_{\{x_k\}\partial*}(\rho)\right) \leq \EP_{\mathcal{L}^D_{\{x_k\}\partial}}(\rho), \hspace{1cm} \alpha_1D\left(\rho\|E^D_{R_j\partial*}(\rho)\right) \leq \EP_{\mathcal{L}^D_{R_j\partial}}(\rho).
\end{align}
Putting everything above together, we have 
\begin{align}
    D(\rho\|E^D_{\Gamma*}(\rho) = D(\rho\|\sigma) &\leq \frac{1}{\min\{\alpha_0,\alpha_1\}}\left(\sum_{x_k\in\Gamma_0}\EP_{\mathcal{L}^D_{\{x_k\}\partial}}(\rho) +C\sum_{k|x_k\in\Gamma_0}\EP_{\mathcal{L}^D_{R_k\partial}}(\rho)\right) \\ &\leq \frac{2mC}{\min\{\alpha_0,\alpha_1\}}\EP_{\mathcal{L}^D_{\Gamma}}(\rho),
\end{align} where in the last inequality we used the positivity and additivity of the entropy production and the fact that, by construction, each site $x\in\Gamma$ is contained in at most a constant number, say $2m$, of regions $R_k\partial$, since they are of fixed finite size. The same holds for $\{x_k\}\partial$ with $\leq2m$ regions. Thus, it follows that
\begin{align}
    \alpha(\mathcal{L}^D_{\Gamma})\geq \frac{\min\{\alpha_0,\alpha_1\}}{2mC} > 0,
\end{align}
so that the final scaling is $\Omega(C^{-1})$.
Therefore in \textit{\underline{Case 1}} from Lemma \ref{lemma:OmegacMLSI} we get $\alpha(\mathcal{L}^D_{\Gamma})= \mathcal{O}(1)_{|\Gamma|\to\infty}$ whereas in \textit{\underline{Case 2}} we get $\alpha(\mathcal{L}^D_{\Gamma})= \Omega((\diam|\Gamma|)^{-1})_{|\Gamma|\to\infty}$. Note that for trees $\diam(B_{0,L})=\mathcal{O}(\log|B_{0,L}|)$.
\end{proof}

\section{Further applications}\label{sec:applications}

We have established a constant lower bound on the cMLSI as our main result for subexponential graphs. Here we explore the consequences of the system size independence of this bound for physical and statistical properties of such systems under dissipative evolution and of its steady state $\sigma$. 
Assume $(\Lambda,\Phi,\beta)$ suitable, such that the system admits a system size independent lower bound on the Davies cMLSI constant $\alpha(\mathcal{L}^D_\Gamma)$ as in \ref{thm:main1D}.

\subsection{Exponential convergence to Gibbs states in the thermodynamic limit}
A direct consequence of the spin-system size $|\Gamma|$ independence of the cMLSI constant $\alpha(\Li^D_\Gamma)$ is that in the thermodynamic limit we have an exponential decay of the relative entropy density between the initial state at time $t$ and the thermal state, with a fixed decay rate $\alpha$. This illustrates how quantum materials in the thermodynamic limit thermalize when weakly coupled to an external environment. 
\begin{corollary}
For local Hamiltonians $D(\rho_t\|\sigma^\Gamma)=\mathcal{O}(|\Gamma|)$, and hence
\begin{align}
    \lim_{\Gamma\uparrow\Lambda}\frac{1}{|\Gamma|}D(\rho_t\|\sigma^\Gamma) \leq e^{-\alpha t}\lim_{\Gamma\uparrow\Lambda}\frac{1}{|\Gamma|}D(\rho\|\sigma^\Gamma).
\end{align}
\begin{proof}
For a local Hamiltonian we have $\|H_\Gamma\|\leq\sum_{\underset{\text{diam}(X)\leq r}{X\subset\Gamma}}\|\Phi_X\| = \mathcal{O}(|\Gamma|)$ and hence
\begin{align}
   d^{|\Gamma|}e^{-\beta\mathcal{O}(|\Gamma|)} = \Tr[\1e^{-\beta\|H\|}]\leq \Tr[e^{-\beta H}]\leq \Tr[\1e^{\beta\|H\|}] = d^{|\Gamma|}e^{\beta\mathcal{O}(|\Gamma|)},
\end{align} hence taking logarithms gives $\log Z_\Gamma= \log \Tr[e^{-\beta H_\Gamma}]=\mathcal{O}(|\Gamma|)$ and we can bound
\begin{align}
    D(\rho\|\sigma^\Gamma) \leq -\Tr\left[\rho\log\frac{e^{-\beta H}}{Z_\Gamma}\right] =\log Z_\Gamma+\beta\Tr[\rho H_\Gamma] =\mathcal{O}(|\Gamma|).
\end{align}
The result follows directly from MLSI with system size independent constant  $\alpha$ by dividing through $|\Gamma|$ and taking the limit.
\end{proof}
\end{corollary}

\subsection{Local Mixing}

The relative entropy $D(e^{t\mathcal{L}}(\rho) \| \sigma )$, as it appears in the rapid mixing results, is rarely an adequate figure of merit in many-body experiments, since it is associated with expectation values of observables with support across the whole system. Even if two states are very different in relative entropy, that difference may be invisible to reasonable read-out capabilities.

Instead, the measurements that typically take place there are of observables with support on a small number of subsystems, such as local magnetization or currents. These, however, are always expected to thermalize in $\mathcal{O}(1)$ time irrespective of the initial conditions, much faster than global observables. 

That this is the case can be shown to follow already from the relative entropy decay. The proof only requires rapid mixing with a constant rate and the Lieb-Robinson bound for local Lindbladians which, under the assumption that $N(l) = e^{o(l)}$, reads as follows \cite{art:BarthelKliesch}.
\begin{theorem}
Let $\mathcal{L}$ be the Lindbladian defined on the whole lattice $\Lambda$. Also, let $Y$ be a region of the lattice, and $V_l$ a larger region such that $Y \subset V_l \subset \Lambda$ and $l \equiv \textup{dist}(Y, \Lambda\setminus V_l)$. For any observable $O_Y$ with $\textup{supp}(O_Y)\subset Y$,
\begin{equation}
    \| e^{t \mathcal{L}^\dagger} (O_Y)- e^{t \mathcal{L}_{V_l}^\dagger} (O_Y)  \| \le  \| O_Y \| e^{-\Omega (l-vt)},
\end{equation}
where $\mathcal{L}_{V_l}=\sum_{\textup{supp}(\mathcal{L}_X) \subset V_l} \mathcal{L}_X $.
\end{theorem}
This means that the time evolution of $O_Y$ is to a good approximation generated by the jump operators around its vicinity. Notice that in \cite{art:BarthelKliesch} it is assumed that $N(l)=\text{poly}(l)$, but sub-exponential is also enough. With this, the result on local thermalization is as follows.

\begin{corollary}\label{co:co1}
    Assuming that the system has rapid mixing with constant decay rate $\alpha$, as in Corollary \ref{thm:main1D}, we have that, for any region $A$, 
\begin{equation}\label{eq:localrapid}
    D( \tr_{A^c}[e^{t \mathcal{L}}(\rho)] \| \sigma_A ) \le \mathcal{O}\left ( \vert A \vert + t^{\kappa+1} \right) \times e^{-\alpha t},
\end{equation}
    where $\sigma_A=\tr_{A^c}[\sigma]$ is the local marginal of the Gibbs state.
\end{corollary}

\begin{proof}
Our starting point is a local Lindbladian of the form of Eq. \eqref{eq:locallind}, as a sum of local terms. 

We now define three regions $A,B,C$ such that $A$ is our subsystem of interest, and $B$ shields $A$ from $C$ with $l_0=\text{dist}(A,C)$, and also such that $A \cup B \cup C=\Lambda$. We also define the Gibbs state on these regions as $\tau_\nu=\frac{e^{-\beta H_\nu}}{Z_\nu}$, with $H_{\nu}$ the subset of Hamiltonian terms with support in the respective $\nu \in \{A,B,C \}$ and $Z_\nu= \text{Tr}[e^{-\beta H_\nu}]$. Notice that these states are different from the respective marginals of the global Gibbs state $\sigma_{\nu}=\tr_{\nu^c} [\sigma]$.  Also $\tau_{AB},Z_{AB}$ are those defined on $A \cup B$.

Consider an arbitrary initial state on the lattice $\rho$, with $\rho_{AB}=\tr_C[\rho]$. A direct calculation shows that
\begin{equation}
    D(\rho_{AB} \otimes \tau_C \| \sigma)= D(\rho_{AB} \| \tau_{AB})+\log{\frac{Z}{Z_{AB}Z_C}}+\text{Tr}[(\rho_{AB} \otimes \tau_C)H_I],
\end{equation}
where $H_I=H-H_{AB}-H_C$, the interaction between $C$ and $AB$. Notice that $D(\rho_{AB} \| \tau_{AB})= \mathcal{O}(\vert A \vert + \vert B \vert)$, so that
\begin{equation} \label{eq:relent}
     D(\rho_{AB} \otimes \tau_C \| \sigma) \le D(\rho_{AB} \| \tau_{AB}) + \mathcal{O}\left(\vert \partial_{BC} \vert \right) \le \mathcal{O}(\vert A \vert + \vert B \vert),
\end{equation}
where $\vert \partial_{BC} \vert$ is the size of the boundary between $B$ and $C$, which by definition is smaller than $\vert B \vert$.

We now consider the time evolution of this state. By data processing, we have
\begin{equation}
    D( e^{t \mathcal{L}}(\rho_{AB} \otimes \tau_C) \| \sigma) \ge D( \tr_{BC}[e^{t \mathcal{L}}(\rho_{AB} \otimes \tau_C)] \| \sigma_A).
\end{equation}
We can compare this state with the one evolved under $\mathcal{L}_{AB}$,
\begin{align}
    \| \tr_{BC}[e^{t \mathcal{L}}(\rho_{AB} \otimes \tau_C)]  - \tr_{B}[e^{t \mathcal{L}_{AB}}(\rho_{AB} )] \|_1 &= \max_{\| O_A \|=1} \left \vert \Tr[O_A (e^{t \mathcal{L}}(\rho_{AB} \otimes \tau_C)  - e^{t \mathcal{L}_{AB}}(\rho_{AB} ) \otimes \tau_C)]  \right\vert \\
    & = \max_{\| O_A \|=1}  \left \vert \text{Tr}[(e^{t \mathcal{L}^\dagger}(O_A)- e^{t \mathcal{L}_{AB}^\dagger}(O_A))\rho_{AB} \otimes \tau_C]  \right \vert\\ & \le \max_{\| O_A \|=1} \| e^{t \mathcal{L}^\dagger}(O_A)- e^{t \mathcal{L}_{AB}^\dagger}(O_A) \| \le e^{-\Omega (l_0-vt)},
\end{align}
where we have used the definition of the $1$-norm and the Lieb-Robinson bound. Similarly, we also obtain
\begin{align}
    \| \tr_{BC}[e^{t \mathcal{L}}(\rho)]- \tr_{BC}[e^{t \mathcal{L}_{AB}}(\rho_{AB} \otimes \tau_C)] \|_1 \le e^{-\Omega (l_0-vt)}, 
\end{align}
so that by the triangle inequality,
\begin{equation}
\| \tr_{BC}[e^{t \mathcal{L}}(\rho_{AB} \otimes \tau_C)]  - \tr_{BC}[e^{t \mathcal{L}}(\rho)]  \|_1  \le e^{-\Omega (l_0-vt)}.
\end{equation}
From continuity, this translates into a bound on relative entropy
\begin{align} \label{eq:contrelent}
  \left  \vert D( \tr_{BC}[e^{t \mathcal{L}}(\rho_{AB} \otimes \tau_C)] \| \sigma_A ) - D( \tr_{BC}[e^{t \mathcal{L}}(\rho)] \| \sigma_A )  \right \vert \le \varepsilon \log \| \sigma_A^{-1} \| + 2h\left (\frac{\varepsilon}{1+\varepsilon}\right),
\end{align}
where $\varepsilon= e^{-\Omega (l_0-vt)}$ and $h(p)$ is the binary entropy. Notice that $\log \| \sigma_A^{-1} \|=\mathcal{O} \left( \vert A \vert \right)$.

Additionally, by assumption, $\mathcal{L}$ is such that it obeys rapid mixing with constant $\alpha$ independent of system size.
Hence, applied to the initial state $\rho_{AB}\otimes \tau_C$, it implies that 
\begin{equation}
    D(e^{t \mathcal{L}}(\rho_{AB}\otimes \tau_C) \| \sigma ) \le D( \rho_{AB}\otimes \tau_C \| \sigma) e^{-\alpha t}.
\end{equation}
Let us choose $l_0$ such that $e^{-\Omega (l_0-vt)} \le e^{-\alpha t}$, which is possible as long as the radius of $B$ increases linearly as $l_0 \propto t$. Data processing, together with Eq. \eqref{eq:contrelent} and Eq. \eqref{eq:relent}, then imply that
\begin{align}
D( \tr_{BC}[e^{t \mathcal{L}}(\rho)] \| \sigma_A ) \le \mathcal{O}\left ( \vert A \vert + t^{\kappa+1} \right) \times e^{-\alpha t},
\end{align}
which finishes the proof.
\end{proof}

Eq. \eqref{eq:localrapid} captures the most easily measurable character of the thermalization process, which occurs even in the thermodynamic limit. If we look at an individual subsystem, both the environment and the remainder of the system $A^c$ can be seen as the entire bath that the subsystem $A$ is coupled to. There is a crucial difference, however, between the two parts of this extended bath: while $A$ is weakly coupled to the external environment, it is strongly coupled to $A^c$. This means that, as Eq. \eqref{eq:localrapid} shows, the thermalization happens towards the state $\sigma_A$, often called the \emph{mean force} Gibbs state \cite{art:MeanForce}, instead of to the bare Gibbs state $\sigma^A=\frac{e^{-\beta H_A}}{Z_A}$ with $H_A=\sum_{\text{supp}(h_X) \subset A} h_X$. A similar global-to-local reduction of rapid mixing was already proven in \cite{art:StabilityCubitt_2015} for the $1$-norm.

\subsection{Wasserstein distance and transportation cost inequalities}

A distance on $\mathcal{D}(\mathcal{H}_\Gamma)$, which is important in what follows is the \textit{quantum Wasserstein distance of order 1} \cite{art:quantumWassersteinDistance1} between two finite dimensional quantum states $\rho,\sigma\in\mathcal{D}(\mathcal{H}_\Gamma)$. It is defined as
\begin{align}
    W_1(\rho,\sigma)&\equiv\|\rho-\sigma\|_{W_1}:= \label{def:WassersteinDistance1}\\&\frac{1}{2}\min\left\{\sum_{i\in\Lambda}\|X^{(i)}\|_1\Bigg|\Tr[X^{(i)}]=0, X^{(i)*}=X^{(i)}, \tr_iX^{(i)}=0 \ \forall i\in\Gamma, \rho-\sigma=\sum_{i\in\Lambda}X^{(i)}\right\}. \notag
\end{align} 
Its dual norm with respect to the Hilbert-Schmidt inner product distance is the \textit{Lipschitz distance} \cite{art:quantumWassersteinDistance1}, i.e. for any self-adjoint observable $A\in\mathcal{B}(\mathcal{H}_\Gamma)$
\begin{align}
    \|A\|_L:=\max\{\Tr[AX]|\Tr[X]=0, X=X^*, \|X\|_{W_1}\leq1\} = 2\max_{i\in\Gamma}\min_{A^{(i)}\in\mathcal{B}(\mathcal{H}_{\Gamma\setminus\{i\}})}\|A-\1_i\otimes A^{(i)}\|,
\end{align}where $\1_i\in\mathcal{B}(\mathcal{H}_i)$ is the identity on system $i$ and $A^{(i)}$ does not act on system $i$.
Thus by definition it holds that $|\Tr[AX]|\leq \|X\|_{W_1}\|A\|_L$ for suitable $X,A$. For a thorough overview and some properties see \cite{art:quantumWassersteinDistance1}.
From the existence of an MLSI constant $\alpha$ independent of the spin-system size $|\Gamma|$, one can show that the following transport cost inequality holds with a transport cost $c^\prime=c\frac{|\Gamma|}{\alpha}$ linear in spin-system size $|\Gamma|$ 
\cite{art:De_PalmaRouzeConcentrationInequalities}.
\begin{align}\label{equ:TransportInequality}
    \|\rho-\sigma\|_{W_1}\leq \sqrt{c^\prime D(\rho\|\sigma)}.
\end{align} 
For a proof of this see \cite[Proposition 16 and Theorem 5]{art:De_PalmaRouzeConcentrationInequalities} and \cite[Theorems 3 and 4]{art:WassersteinDistancesDattaRouze}. The fact that the transport cost scales linearly in $|\Gamma|$ has the following important consequences.

\subsubsection{Tighter bounds on the entropy difference and convergence via relative entropy}
The first application of this concerns bounding the von-Neumann entropy difference between two quantum states $\rho,\sigma$ by their relative entropy.
\begin{corollary}
Let $\rho\in\mathcal{D}(\mathcal{H}_\Gamma)$ be arbitrary and $\sigma\equiv\sigma^{\Gamma}$ the invariant Gibbs state in $\Gamma$. The following two bounds hold
\begin{align}
    i) \ & |S(\rho)-S(\sigma)|\leq g\left(c\sqrt{\frac{|\Gamma|}{\alpha}}\sqrt{D(\rho\|\sigma)}\right)+c\sqrt{\frac{|\Gamma|}{\alpha}}\log{(d^2|\Gamma|)}\sqrt{D(\rho\|\sigma)} \notag \\ &\hspace{2.3cm}= \mathcal{O}(\sqrt{|\Gamma|}\log |\Gamma|)_{|\Gamma|\to\infty}\sqrt{D(\rho\|\sigma)},
    \label{equ:EntropyDifferenceBound1}\\
    ii) \ & |S(\rho_t)-S(\sigma)| = \mathcal{O}\left(|\Gamma|\log |\Gamma|,\frac{e^{-\frac{\alpha t}{2}}}{\sqrt{\alpha}}\right),
\end{align} where $g(t)=(t+1)\log(t+1)-t\log t=o(t)_{t\to\infty}$,  $\rho_t:=e^{t\mathcal{L}_{\Lambda*}}(\rho)$, and $c$ is some constant independent of $|\Gamma|$ depending only on the locality of the Lindbladian $\mathcal{L}^D_\Lambda$.
\end{corollary}

These inequalities  represent a $\mathcal{O}\left(\frac{\sqrt{|\Gamma|}}{\log |\Gamma|}\right)$ improvement compared to Pinsker's inequality (see \cite{art:quantumWassersteinDistance1}) as 
\begin{equation}
    \|\rho-\sigma\|_{W_1}\leq \frac{|\Gamma|}{2}\|\rho-\sigma\|_{1}\leq \frac{|\Gamma|}{\sqrt{2}}\sqrt{D(\rho\|\sigma)} \, ,
\end{equation}
or to the following inequality from \cite{art:Reeb_2015Wolf}
\begin{equation}
  |S(\rho)-S(\sigma)|\leq \sqrt{3}\log(d^{|\Gamma|})\sqrt{D(\rho\|\sigma)}=\mathcal{O}(|\Gamma|)\sqrt{D(\rho\|\sigma)}  \, .
\end{equation}
The difference is that these apply to arbitrary quantum states $\rho,\sigma$, whereas our result in Eq.  \eqref{equ:EntropyDifferenceBound1} requires $\sigma$ to be the Gibbs state of a suitable Hamiltonian. 
Note that these entropy difference bounds are optimal in their scaling in $|\Gamma|$ up to logarithmic correction, since the entropy difference is an extensive quantity, i.e. it scales as $\mathcal{O}(|\Gamma|)$.

\begin{proof}
We first use the following continuity bound from \cite[Theorem 1]{art:quantumWassersteinDistance1}, that states that for any two density matrices $\rho,\sigma\in\mathcal{D}(\mathcal{H}_\Gamma)$
\begin{align}
    |S(\rho)-S(\sigma)|\leq g(\|\rho-\sigma\|_{W_1})+\|\rho-\sigma\|_{W_1}\log(d^2|\Gamma|),
\end{align} where $d$ is the local Hilbert space dimension, $\|\cdot\|_{W_1}$ the quantum Wasserstein distance of order 1 (see \eqref{def:WassersteinDistance1} for the definition), and $g(t)=(t+1)\log(t+1)-t\log t$. By the transport cost inequality from \cite[Proposition 16, Theorem 5]{art:De_PalmaRouzeConcentrationInequalities} and \cite[Theorems 3 and 4]{art:WassersteinDistancesDattaRouze}\footnote{Note that in that paper a different normalization convention for the quantum Wasserstein distance was used, as compared to here or the other references in this section}, which holds under the assumptions of the Corollary, we get Eq. \eqref{equ:TransportInequality}.
Combining this with the inequality just above gives i). For ii) use i), the bound $D(\rho\|\sigma^\Gamma)\leq\mathcal{O}(|\Gamma|)$ from above, and the MLSI in its integrated form  $D(\rho_t\|\sigma)\leq e^{-\alpha t}D(\rho\|\sigma)$.
\end{proof}

 
\subsubsection{Gaussian concentration bound} 
A system size independent cMLSI constant also yields the following Gaussian concentration bound.
\begin{corollary} 
Let $O\in\mathcal{B}(\mathcal{H}_\Lambda)$ be a $k$-local observable, i.e. 
\begin{align}
    O=\sum_{\underset{|X|\leq k}{X\subset\Lambda}}o_X,
    \label{def:LongrangeLocalObservables}
\end{align}
such that for all $i\in\Lambda$, \ $\sum_{X\subset\Lambda:X\ni i}\|o_X\|\leq g$, where each $o_X$ acts only non trivially on sites $X\subset\Lambda$. Let $\sigma$ be the Gibbs state of some geometrically-$2$-local, uniformly bounded, commuting Hamiltonian with uniform exponential decay of correlations, at any fixed inverse temperature $\beta>0$. Then for $r\geq 0$ it holds that
\begin{align}
    \mathbb{P}_\sigma(|O-\Tr[O\sigma]|\geq s)\leq 2\exp\left[-\frac{\alpha s^2}{\mathcal{O}(|\Gamma|)}\right].
    \label{equ:CorolarryGaussianConcentration}
\end{align}
\end{corollary}
This means that quantum spin systems with geometrically-local, commuting Hamiltonians with uniform exponential decay of correlations give rise to a sub-Gaussian random variable in their thermal equilibrium states for any observables of the above form. 
This bound reproduces the scaling of independent random variables from Hoeffding's inequality, and is thus optimal. Eq. \eqref{equ:CorolarryGaussianConcentration} constitutes a tightening in terms of its $|\Gamma|$-dependence and generalization to a larger class of observables and $s$-values of the previous best known Gaussian concentration bound in \cite{art:StateOfTheArtGausianConcentrationBoundAnshu_2016}\footnote{The statement there \cite[Theorem 4.2]{art:StateOfTheArtGausianConcentrationBoundAnshu_2016} is equivalent to $\mathbb{P}_\sigma(|O-\Tr[O\sigma]|\geq s)\leq 2\exp\left[-\frac{\alpha s}{\mathcal{O}(\sqrt{|\Gamma|})}\right] $.}.

\begin{proof}
This follows from \Cref{thm:mainGap} and \cite[Theorem 7 and Lemma 7]{art:De_PalmaRouzeConcentrationInequalities}, when using the transport cost $c^\prime=c\frac{|\Gamma|}{\alpha}$ in Eq. \eqref{equ:TransportInequality} and the fact that $\|\Delta^{\frac{1}{2}}_\sigma(O)\|_L\leq 4gC=\mathcal{O}(1)_{|\Gamma|\to\infty}$, since
\begin{align}
    \|\Delta^{\frac{1}{2}}_\sigma(O)\|_L &\leq 2\max_{i\in\Lambda}\left\|\Delta_\sigma^{\frac{1}{2}}(O)-\1^{(i)}_d\otimes\hat{\text{tr}}_i\Delta_\sigma^{\frac{1}{2}}(O)\right\| \\
    &=2\max_{i\in\Lambda}\left\|\sum_{X:|X|\leq k}\Delta_\sigma^{\frac{1}{2}}(o_X)-\1^{(i)}_d\otimes\hat{\text{tr}}_i\sum_{X:|X|\leq k}\Delta_\sigma^{\frac{1}{2}}(o_X)\right\| \\ &=2\max_{i\in\Lambda}\left\|\sum_{X:X\partial\ni i, |X|\leq k}\left(\Delta_\sigma^{\frac{1}{2}}(o_x)-\1^{(i)}_d\otimes\hat{\text{tr}}_i\Delta_\sigma^{\frac{1}{2}}(o_x)\right)\right\| \\
    &\leq 2\max_{i\in\Lambda}\sum_{X:X\partial\ni i, |X|\leq k}2\left\|\Delta_\sigma^{\frac{1}{2}}(o_x)\right\| \leq 4\max_{i\in\Lambda}\sum_{X:|X|\leq k, X\partial\ni i}\|o_X\|\left\|\text{exp}\left(\frac{\beta}{2}\sum_{\overset{B:B\cap X\neq\emptyset}{\text{diam}(B)\leq r}}\Phi_B\right)\right\| \\ &\leq 4g\text{exp}\left(\frac{\beta Jc_{r,k,\nu}}{2}\right) \equiv 4gC,
\end{align} where $\hat{\text{tr}}_i$ is the normalized partial trace. The last two inequalities follow from
\begin{align}
    \|\Delta_\sigma^{\frac{1}{2}}(o_X)\| &=\|\sigma^{\frac{1}{2}}o_X\sigma^{-\frac{1}{2}}\| = \|e^{-\frac{\beta}{2}\sum_{B\cap X\neq\emptyset}\Phi_B}o_Xe^{\frac{\beta}{2}\sum_{B\cap X\neq\emptyset}\Phi_B}\| \\ &\leq \|o_X\| \|\text{exp}\left(\frac{\beta}{2}\sum_{\overset{B:B\cap X\neq\emptyset}{|X|\leq k, \text{diam}(B)\leq r}}\Phi_B\right)\| \leq \|o_X\|\text{exp}\left(\frac{\beta J c_{r,k,\nu}}{2}\right).
\end{align}
Here $c_{r,k,\nu}$ is a constant which depends only on the locality of $O$, the geometric locality $r$ of the Hamiltonian, and the growth constant $\nu$ of the graph. 
\end{proof}

\subsubsection{Ensemble equivalence under (long-range) Lipschitz observables}
The canonical ensemble with inverse temperature $\beta$ is given by the Gibbs state $\sigma_\beta$, where we now write the temperature dependence explicitly. 

Let $\sigma_{E,\delta}$ be the microcanonical ensemble with average energy $E$ defined on an energy-shell width $\delta$, so that if $P$ is the spectral measure of the Hamiltonian $H=\sum_{E}EP(E)\equiv \sum_{m}E_m\ketbra{E_m}{E_m}$ then
\begin{align}
    \sigma_{E,\delta}:=\frac{P((E-\delta,E])}{\Tr[P((E-\delta,E])]}=\frac{1}{\mathcal{N}_{E,\delta}}\sum_{E_m\in(E-\delta,E]}\ketbra{E_m}{E_m}.
\end{align}
Now, let 
$E=\text{arg}\max_{E\in\mathbb{R}}\left(e^{-\beta E}\mathcal{N}_{E,\delta}\right)$, where $\mathcal{N}_{E,\delta}:=\Tr[P((E-\delta,E])]$ is the number of eigenstates in the energy interval $(E-\delta,E]$. This defines the average energy of the microcanonical ensemble corresponding to the energy of $\sigma_\beta$ \cite{art:Tasaki2018,art:KuwaharaSaitoETHClustering_2020}.

Two ensembles represented respectively by the families of states $\{\sigma_1^\Gamma,\sigma_2^\Gamma\}_{\Gamma\subset\subset\Lambda}$ are said to be \textit{equivalent} if, in the thermodynamic limit, they produce the same expectation values on averaged geometrically-local observable $\frac{O}{|\Gamma|}=\frac{1}{|\Gamma|}\sum_{i=1}^{|\Gamma|} O_i$, with $\|O_i\|\leq g$ \cite{art:KuwaharaSaitoETHClustering_2020}. That is, if for any such observable
\begin{align}
    \left|\Tr\left[\sigma_1^\Gamma \frac{O}{|\Gamma|}\right]-\Tr\left[\sigma_2^\Gamma \frac{O}{|\Gamma|}\right]\right|= \frac{1}{|\Gamma|}|\Tr[(\sigma_1^\Gamma-\sigma_2^\Gamma) O]|\overset{|\Gamma|\rightarrow\infty}{\longrightarrow} 0 .
\end{align} 
In \cite{art:KuwaharaSaitoEnsembleEquivalence_2020} it was shown that the microcanonical and canonical ensembles are equivalent in this sense when the system satisfies suitable concentration bounds, such as Eq. \eqref{equ:CorolarryGaussianConcentration}. 
We extend this notion of equivalence in the 1D case to the more general class of Lipschitz observables, defined as $O\in\mathcal{B}(\mathcal{H}_\Gamma)$, s.t. $\|O\|_L<\infty$. These notably include long-range locally bounded, $k-$local observables of the form of Eq.  \eqref{def:LongrangeLocalObservables}.
\begin{corollary}[Corollary 2 of \cite{art:De_PalmaRouzeConcentrationInequalities} applied to our setting]
For any Lipschitz observable $O\in\mathcal{A}_\Lambda$, i.e. $O\in\mathcal{A}_\Lambda$ s.t. $\|O\|_L<\infty$ and for $\sigma_{E,\delta}, \sigma_\beta$ the micro- and canonical ensemble states with the same energy $E=\textup{arg}\max_{E\in\mathbb{R}}\left(e^{-\beta E}\mathcal{N}_{E,\delta}\right)$, respectively,  it holds that 
\begin{align}
    \frac{1}{|\Gamma|}|\Tr[\sigma_{E,\delta}O]-\Tr[\sigma_\beta O]|\leq \|O\|_L\hspace{1mm}o(1)_{|\Gamma|\to\infty},
\end{align}
for all $\sigma_E,\delta$ such that $\delta=e^{-\mathcal{O}(|\Gamma|)}$.
\end{corollary}
\begin{proof}
This follows directly from \cite[Corollary 2]{art:De_PalmaRouzeConcentrationInequalities}  when employing the linearity of the transport cost in the system size. The proof follows from the transport cost inequality, which is used to bound the relative entropy between the microcanonical and canonical ensembles, which yields the result.
\end{proof}

\section{Outlook}

In this paper, we have addressed the problem of estimating the speed of convergence of certain quantum dissipative evolutions governed by Lindbladians to their steady states, given by Gibbs states of local, commuting Hamiltonians. We have particularly considered a large class of Davies generators with nearest neighbour, commuting Hamiltonians, and derived for them a positive (and, in some cases, system-size independent) MLSI from either a positive spectral gap in the Lindbladian, or at high enough temperature. 

After the completion of this work, the knowledge about existence of positive MLSIs, and therefore rapid mixing, for various physically relevant systems stands as follows:

\begin{itemize}
    \item For 1D systems, for Davies generators with k-local, commuting, translation-invariant Hamiltonians, at any finite, inverse temperature $\beta$, a log-decreasing MLSI $\alpha(\mathcal{L}_\Gamma^D) = \Omega ((\log |\Gamma| )^{-1})_{|\Gamma| \rightarrow \infty}$ was derived in \cite{art:EntropyDecayOf1DSpinChain-Cambyse, art:ImplicationsAndRapidTermalization-Cambyse}. Here, we have lifted that to a constant MLSI  $\alpha(\mathcal{L}_\Gamma^D) = \Omega (1)_{|\Gamma| \rightarrow \infty}$ and extended it to non-translation-invariant Hamiltonians. Both results yield rapid mixing.  
    \item For 2D systems, for Davies generators with k-local, commuting, Hamiltonians, at any finite, inverse temperature $\beta$, a positive spectral gap was shown to be equivalent to $\mathbb{L}_2$-clustering in \cite{art:QuantumGibbsSamplers-kastoryano2016quantum}. Here, we have shown that, whenever the correlation length is small enough, this is further equivalent to a square-root decreasing MLSI  $\alpha(\mathcal{L}_\Gamma^D) = \Omega (\sqrt{|\Gamma|}^{-1})_{|\Gamma| \rightarrow \infty}$, yielding a $\mathcal{O}(\sqrt{|\Gamma|} \log |\Gamma|)$-mixing time.  
    \item For high D systems and their Schmidt generators with nearest neighbour, commuting, Hamiltonians, a positive, constant MLSI $\alpha(\mathcal{L}_\Gamma^D) = \Omega (1)_{|\Gamma| \rightarrow \infty}$ was shown to hold at high-enough temperature in \cite{art:2localPaper}. Here, we have extended this to Davies generators of nearest neighbour, commuting, Hamiltonians at high-enough temperature, concluding for them rapid mixing. 
    \item For trees, we have provided the first results in the direction of rapid mixing, showing a log-decreasing MLSI $\alpha(\mathcal{L}_\Gamma^D) = \Omega ((\log |\Gamma| )^{-1})_{|\Gamma| \rightarrow \infty}$ whenever the correlation length is small enough. 
\end{itemize}
This shows definite progress in our understanding of the mixing times of Davies generators with commuting Hamiltonians, for which we aim to eventually have a complete picture. In this respect, there are still unsolved natural questions, such as:
\begin{itemize}
    \item Can we extend the results for high D systems to geometrically-$k$-local Hamiltonians? This would require a totally new approach, as the Schmidt generators cannot be defined beyond nearest neighbour interactions.
    \item Can we provide a more accurate result for specific models? For instance, for Davies generators associated to quantum double models, a positive spectral gap was shown to hold at any positive temperature in \cite{lucia2021thermalization}, extending the prior \cite{AlickiFannesHorodecki2009toriccode,Komar2016abelianquantumdoubles}. Therefore, it might be possible that a similar behaviour translates to the mixing time, yielding rapid mixing at every positive temperature.    
\end{itemize}
The natural extension of this problem to the context of non-commuting Hamiltonians is much less understood. As the Davies generator loses its desirable properties whenever the Hamiltonian considered is non-commuting, specially the locality, an alternative, physically-relevant Lindbladian needs to be considered in this case.  The recent \cite{art:chen2023noncommuting} introduced a possible quasi-local Lindbladian that by-passes the problems induced by the Davies, which some of the authors of this manuscript recently showed to be gapped in \cite{art:rouze2024gapnoncommuting}. The question whether this can be lifted to a positive MLSI, or any other means to derive rapid mixing, even for any specific case (e.g. 1D systems) remains open. An alternative approach to estimate mixing times of Lindbladians has just been introduced in \cite{fang2024hypocoercivity}, but it is unclear whether it could be employed in this context.

The problem of estimating mixing times of Lindbladians has been well explored in the quantum community in the past few years, with a notable increase of interest very recently. The reason for this partly stems from an upraise in the tools at our disposal to tackle this problem, and the appearance of more applications derived from it. In particular, rapid mixing constitutes a fundamental approach for quantum Gibbs sampling, as stated in \cite{art:QuantumGibbsSamplers-kastoryano2016quantum,art:chen2023fast,art:chen2023noncommuting,art:2localPaper}, which nevertheless can also be explored with alternative approaches as in \cite{tang2024separableGibbsstates}. It has also been recently shown \cite{trivedi2022simulators,trivedi2024simulators} that systems with rapid mixing are promising candidates for quantum advantage experiments in the context of analogue open quantum simulation.


\vspace{0.5cm}

\noindent \emph{Acknowledgments:} A.C.  acknowledges the support of the Deutsche Forschungsgemeinschaft (DFG, German Research Foundation) – Project-ID 470903074 – TRR 352. A.M.A. acknowledges support from the Spanish Agencia Estatal de Investigacion through the grants ``IFT Centro de
Excelencia Severo Ochoa CEX2020-001007-S" and ``Ram\'on
y Cajal RyC2021-031610-I", financed by
MCIN/AEI/10.13039/501100011033 and the European
Union NextGenerationEU/PRTR. This project was funded within the QuantERA II Programme that has received funding from
the EU’s H2020 research and innovation programme under the GA No 101017733.

\bibliographystyle{abbrv}
\bibliography{lit}

\appendix
\section{Proof of properties of the relation from Definition \ref{def:Relation}}\label{app:proofs_relation}

This section contains the proofs of Proposition \ref{prop:properties_relation}, Corollary \ref{cor:Relation}, and Proposition \ref{prop:RelationDmax}.

\begin{proof}[Proof of Proposition \ref{prop:properties_relation}]
\begin{enumerate} 
    \item[$1)$] The reflexivity is trivial.
    \item[$2)$] For symmetry, notice that the spectra and thus the spectral radii of $B^{-\frac{1}{2}}AB^{-\frac{1}{2}}-\1$ and $A^{\frac{1}{2}}B^{-1}A^{\frac{1}{2}}-\1$ are the same.  Since they are both normal, we have $\|B^{-\frac{1}{2}}AB^{-\frac{1}{2}}-\1\|\leq\epsilon$ which proves the second implication. Now the first implication follows via 
\begin{align}
    \|B^{\frac{1}{2}}A^{-1}B^{\frac{1}{2}}-\1\| = \|(B^{-\frac{1}{2}}AB^{-\frac{1}{2}})^{-1}-\1\|& = \Big\|\sum_{k=1}^\infty{(-1)^k}(B^{-\frac{1}{2}}AB^{-\frac{1}{2}}-\1)^k\Big\| \\
    &\leq \sum_{k=1}^\infty\|B^{-\frac{1}{2}}AB^{-\frac{1}{2}}-\1\|^k \\ &\leq \frac{\epsilon}{1-\epsilon},
\end{align}
where in the last line we used that by assumption $\|B^{-\frac{1}{2}}AB^{-\frac{1}{2}}-\1\|\leq\epsilon$.
 \item[$3)$] For transitivity, notice that 
\begin{align}
    \|A^{\frac{1}{2}}C^{-1}A^{\frac{1}{2}}-\1\| &= \|A^{\frac{1}{2}}B^{-\frac{1}{2}}(B^{\frac{1}{2}}C^{-1}B^{\frac{1}{2}}-\1+\1)B^{-\frac{1}{2}}A^{\frac{1}{2}}-\1\| \\
    &\leq \|A^{\frac{1}{2}}B^{-\frac{1}{2}}(B^{\frac{1}{2}}C^{-1}B^{\frac{1}{2}}-\1)B^{-\frac{1}{2}}A^{\frac{1}{2}}\| + \|A^{\frac{1}{2}}B^{-1}A^{\frac{1}{2}}-\1\| \\ &\leq \|(B^{\frac{1}{2}}C^{-1}B^{\frac{1}{2}}-\1)(B^{-\frac{1}{2}}AB^{-\frac{1}{2}})\| + \|A^{\frac{1}{2}}B^{-1}A^{\frac{1}{2}}-\1\| \\ &\leq \epsilon_2(1+\epsilon_1) + \epsilon_1 = \eta.
\end{align}
For the second inequality, note that $\|XYX^*\|=|\text{spec}(XYX^*)|=|\text{spec}(X^*XY)|\leq \|X^*XY\|$ holds, for $X,Y$ some operators with $Y$ self-adjoint, where $|\text{spec}(X)|$ denotes the spectral radius of $X$. Using this with $X=A^{\frac{1}{2}}B^{-\frac{1}{2}}$ and $Y=B^{\frac{1}{2}}C^{-1}B^{\frac{1}{2}}-\1$ then gives the second inequality. The third inequality follows from the assumptions and the fact that $\|B^{-\frac{1}{2}}AB^{-\frac{1}{2}}\|\leq 1+\epsilon_1$ by the second implication of symmetry.
\item[$4)$]  For the tensor multiplicativity, notice that
\begin{align}
    \|(A\otimes F)^{\frac{1}{2}}(\Tilde{A}\otimes \Tilde{F})^{-1}(A\otimes F)^{\frac{1}{2}}-\1\| &= \|A^{\frac{1}{2}}\Tilde{A}^{-1}A^{\frac{1}{2}}\otimes F^{\frac{1}{2}}\Tilde{F}^{-1}F^{\frac{1}{2}}-\1\otimes\1\| \\ &\leq 
    \|(A^{\frac{1}{2}}\Tilde{A}^{-1}A^{\frac{1}{2}}-\1)\otimes(F^{\frac{1}{2}}\Tilde{F}^{-1}F^{\frac{1}{2}}-\1)\| \\&\; \; \; +\|(A^{\frac{1}{2}}\Tilde{A}^{-1}A^{\frac{1}{2}}-\1)\otimes\1\|+\|\1\otimes(F^{\frac{1}{2}}\Tilde{F}^{-1}F^{\frac{1}{2}}-\1)\| \\ &\leq \epsilon_1\epsilon_2+\epsilon_1+\epsilon_2 = \eta.
\end{align}
\item[$5)$]  The following chain of implications holds
\begin{align}
    \|D^{\frac{1}{2}}E^{-1}D^{\frac{1}{2}}-\1\|\leq \epsilon &\Leftrightarrow (1-\epsilon)\,\1 \leq D^{\frac{1}{2}}E^{-1}D^{\frac{1}{2}} \leq (1+\epsilon)\,\1 \\ &\Leftrightarrow \frac{1}{1+\epsilon}\,\1 \leq D^{-\frac{1}{2}}ED^{-\frac{1}{2}} \leq \frac{1}{1-\epsilon}\,\1 \\
    & \Leftrightarrow \frac{D}{1+\epsilon} \leq E \leq \frac{D}{1-\epsilon} \, .
\end{align}
Since $P$ is self-adjoint the above implies that the following holds:
\begin{align}
    \frac{PDP}{1+\epsilon} \leq PEP \leq \frac{PDP}{1-\epsilon} \implies \frac{\tr_{\mathcal{H}^\prime}PDP}{1+\epsilon} \leq \tr_{\mathcal{H}^\prime}PEP \leq \frac{\tr_{\mathcal{H}^\prime} PDP}{1-\epsilon} \, ,
\end{align}
where we identified $P$ with $\1\otimes P$ for sake of clarity. The above holds also in reverse on $\supp{\tr_{\mathcal{H}^\prime}(PDP)}=\supp{\tr_{\mathcal{H}^\prime}(PEP)}$. The fact that the (operator algebraic) supports of these two are the same is directly implied by the previous sandwich inequality and positivity. So it implies that 
\begin{equation}
   \|(\tr_{\mathcal{H}^\prime}PDP)^{\frac{1}{2}}(\tr_{\mathcal{H}^\prime}PEP)^{-1}(\tr_{\mathcal{H}^\prime}PDP)^{\frac{1}{2}}-\1_{\supp(\tr_{\mathcal{H}^\prime}(PDP)})\|\leq \epsilon \, ,
\end{equation}
where $(\cdot)^{-1}$ represents the generalized inverse. 
\item[$6)$] First note that if $D\overset{\epsilon}{\sim} E$, then $D\overset{\mu}{\sim} \lambda E$ for $\lambda>0$, where $\mu=\frac{\epsilon}{\lambda}+|1-\frac{1}{\lambda}|$, since
\begin{align}
    \|D^{\frac{1}{2}}(\lambda E)^{-1}D^{\frac{1}{2}}-\1\| \le \lambda^{-1}\,\|D^{\frac{1}{2}}E^{-1}D^{\frac{1}{2}}-\1\|+\|\lambda^{-1}\1-\1\|.
\end{align}
Now since $D\overset{\epsilon}{\sim} E \implies \tr_{\mathcal{H}^\prime}PDP \overset{\epsilon}{\sim}\tr_{\mathcal{H}^\prime}PEP \Leftrightarrow \frac{\tr_{\mathcal{H}^\prime}PDP}{1+\epsilon}\leq \tr_{\mathcal{H}^\prime}PEP \leq \frac{\tr_{\mathcal{H}^\prime}PDP}{1-\epsilon} \implies 1+\epsilon \geq \frac{\Tr[PDP]}{\Tr[PEP]}\geq 1-\epsilon$ by the proof of property 5). It follows that
\begin{align}
\|\mathscr{X}\|&\equiv\Big\|\left(\frac{\tr_{\mathcal{H}^\prime}PDP}{\Tr[PDP]}\right)^{\frac{1}{2}}\left(\frac{\tr_{\mathcal{H}^\prime}PEP}{\Tr[PEP]}\right)^{-1}\left(\frac{\tr_{\mathcal{H}^\prime}PDP}{\Tr[PDP]}\right)^{\frac{1}{2}}-\1\Big\| \\&= \| (\tr_{\mathcal{H}^\prime}PDP)^{\frac{1}{2}}(\tr_{\mathcal{H}^\prime}PEP)^{-1}(\tr_{\mathcal{H}^\prime}PDP)^{\frac{1}{2}}\lambda-\1\|,
\end{align}
with $1-\epsilon \leq\lambda^{-1}=\frac{\Tr[PDP]}{\Tr[PEP]}\leq 1+\epsilon$.
Hence $\|\mathscr{X}\|\leq \epsilon\lambda^{-1}+|1-\lambda^{-1}|\leq \epsilon(1+\epsilon)+\epsilon=\epsilon(2+\epsilon)$.
\end{enumerate}
\end{proof}

\begin{proof}[Proof of Corollary \ref{cor:Relation}]
The corollary is easily proved by induction using the transitivity of the relation, spelled out here for convenience. Assume
$  A_i\overset{\epsilon}{\sim}A_{i+1}$ for all $i$ and set $A_0\overset{\eta_k}{\sim}A_k$, then
by transitivity we have the recursion 
\begin{align}
    \eta_{k+1} \leq \eta_k(1+\epsilon)+\epsilon.
\end{align}
For $k=1$, noting that $\eta_1=\epsilon$, this gives $\eta_2=\epsilon(1+\epsilon)+\epsilon=\epsilon^2+2\epsilon=(1+\epsilon)^2-1$.
For any $k$ we have by induction that 
\begin{align}
    \eta_{k+1} \leq \eta_k(1+\epsilon)+\epsilon = ((1+\epsilon)^{k}-1)(1+\epsilon)+\epsilon = (1+\epsilon)^{k+1}-1.
\end{align}
\end{proof}

\begin{proof}[Proof of Proposition \ref{prop:RelationDmax}]
Assume that $\supp(\omega)\subset\supp(\tau)$, else the statement clearly holds. We may now WLOG assume the supports to be equal.
    Set $\epsilon = \|\omega^{\frac{1}{2}}\tau^{-1}\omega^{\frac{1}{2}}\|$, then the lower bound follows from
    \begin{align}
        D_\text{max}(\omega\|\tau)=\log\|\omega^{\frac{1}{2}}\tau^{-1}\omega^{\frac{1}{2}}\|\leq\log\|\omega^{\frac{1}{2}}\tau^{-1}\omega^{\frac{1}{2}}-\1+\1\|\leq\log(1+\epsilon)\leq\epsilon.
    \end{align}
    For the upper bound see that for $X:=\omega^{\frac{1}{2}}\tau^{-1}\omega^{\frac{1}{2}} \geq 0$ we have
    \begin{align}
        \|\1-X\| &= \|\1-\exp\log X\| =\left\|\1-\sum_{n=0}^\infty \frac{(\log X)^n}{n!}\right\| = \left\|\sum_{n=1}^\infty \frac{(\log X)^n}{n!}\right\| = \left\|\sum_{n=1}^\infty \frac{(\log X)^n}{n!}\right\| \\ &= \left\|\log X\sum_{n=0}^\infty \frac{1}{n}\frac{(\log X)^n}{n!}\right\|\leq \|\log X\|\exp(\|\log X\|),
    \end{align} where the last inequality follows from subadditivity aubmultiplicativity of the norm and $\frac{1}{n}\leq 1$. This is the desired upper bound, hence we are done.
\end{proof}

\section{Proof of Lemma \ref{lemma:ImportantCommutingLemma1}}\label{sec:proof_lemma_commuting}

\begin{proof}[Proof of \Cref{lemma:ImportantCommutingLemma1}.]
Let us prove each item separately:
\begin{enumerate}
\item[$1)$]  If $\Lambda$ and $\Phi$ are chosen as in the statement of the lemma, then there exists a constant $c_{r,\nu}$ depending only on $r,\nu$ such that
\begin{align}
    \|E_{A,B}^{\pm1}\| &=\|\operatorname{e}^{\mp \beta H_{AB}}\operatorname{e}^{\pm\beta H_A\pm\beta H_B}\| \\&= \Big\|\exp \Big\{\mp\beta\sum_{X \cap A\neq \emptyset, X \cap B\neq \emptyset} \Phi_X \Big\}\Big\| \\ &\leq \exp \Big\{ \Big(\beta J \sum_{\overset{X \cap A\neq \emptyset, X \cap B\neq \emptyset}{ \text{diam}(X)\leq r}}1 \Big) \Big\} \\&\leq \exp(\beta J c_{r,\nu}\min\{|\partial_B A|,|\partial_A B|\})=:K_{A,B}=\exp{(\mathcal{O}(\beta|\partial_{A,B}|))}, 
\end{align}
where we are denoting $\partial_A B := (\partial B),  \cap A$, $\partial_B A := (\partial A),  \cap B$, and $\partial_{A,B}:=\min\{|\partial_B A|,|\partial_A B|\}$. Hence it also holds
\begin{align}
     &\|E_{A,B}^{\pm1}E_{AB,C}^{\pm1}\| \leq \|E_{A,B}^{\pm1}\| \|E_{AB,C}^{\pm1}\| \leq K_{A,B}K_{B,C}.
\end{align}
    \item[$2)$]  Consider the map
\begin{equation}
    \begin{array}{cccc}
        \mathcal{B}(\mathcal{H}_\Gamma)  & \to & \mathcal{B}(\mathcal{H}_{\Gamma\setminus B}) & \\
       Q  & \mapsto & \tr_{B}[\sigma^BQ] & =\frac{1}{\Tr[e^{-H_B}]}\tr_B[e^{-\frac{1}{2}H_B}Qe^{-\frac{1}{2}H_B}] \, ,
    \end{array}
\end{equation}     
   which is positive and unital. Note that if $Q>0$ is strictly positive, then so is $Q^{-1}>0$ and $\lambda_{\min}(Q)\1=\|Q^{-1}\|^{-1}\1\leq Q\leq \|Q\|\1=\lambda_{\max}(Q)\1$. Applying the aforementioned map to this inequality immediately gives that 
\begin{align}
   & \|Q^{-1}\|^{-1}\1_{B^c}\leq \tr_B[\sigma^BQ]\leq \|Q\|\1_{B^c}, \\
    &\|Q\|^{-1}\1_{B^c}\leq \tr_B[\sigma^BQ]^{-1}\leq \|Q^{-1}\|\1_{B^c},
\end{align} since inversion of two commuting operators is order reversing. We conclude by taking norms.
\item[$3)$] This is just a special case of $2)$ given $1)$, since each of the $E_{A,B}$ are strictly positive, as the potential $\Phi$ is commuting, e.g. by self-adjointness and the spectral theorem.
\end{enumerate}
\end{proof}

\section{Proof of strong local indistinguishability in 1D}\label{sec:proof_stronglocalind_1d}

This section is devoted to the proof of \Cref{thm:StrongLocalIndist1D}. 

\begin{proof}[Proof of \Cref{thm:StrongLocalIndist1D}:]
The proof follows a similar path as the proof of \cite[Proposition 8.1]{art:Bluhm2022exponentialdecayof} and uses local indistinguishability (see e.g. \cite[Theorem 5]{art:FiniteCorrelationLengthImpliesEfficient}, \cite[Corollary 2]{art:LocalityofTemperature}), which follows from the assumptions of the theorem. Thus, we know that
\begin{equation}
    \|\tr_{BC}\sigma^{ABC}-\tr_{B}\sigma^{AB}\|_1\leq K^\prime e^{-a^\prime l} \, ,
\end{equation}
 with $K^\prime,a^\prime>0$ depending only on $r,J\beta$.
 
Using the notation $E_{X,Y}:=e^{-H_{XY}}e^{H_X+H_Y}$ for $X,Y\subset I$ disjoint, we rewrite:
\begin{align}
    &(\tr_{BC}[\sigma^{ABC}])(\tr_{B}[\sigma^{AB}])^{-1} \\& \hspace{1.5cm}= \tr_{BC}[e^{-H_{ABC}}]\tr_{B}[e^{-H_{AB}}]^{-1}\Tr[e^{-H_{ABC}}]^{-1}\Tr[e^{-H_{AB}}] \\
    & \hspace{1.5cm} = \tr_{BC}[e^{-H_{ABC}}]e^{H_A}e^{-H_A}\tr_{B}[e^{-H_{AB}}]^{-1}\Tr[e^{-H_{ABC}}]^{-1}\Tr[e^{-H_{AB}}] \\[1mm]
    &\hspace{1.5cm} = \tr_{BC}[e^{-H_{BC}}e^{-H_{ABC}}e^{H_A}e^{H_{BC}}]\tr_{B}[e^{-H_B}e^{-H_{AB}}e^{H_A}e^{H_B}]^{-1}\Tr[e^{-H_{ABC}}]^{-1}\Tr[e^{-H_{AB}}] \\[1mm]
    &\hspace{1.5cm} = \tr_{BC}[\sigma^{BC}E_{A,BC}]\tr_B[\sigma^BE_{A,B}]^{-1}\underbrace{\frac{\Tr[e^{-H_{AB}}]\Tr[e^{-H_{BC}}]}{\Tr[e^{-H_{ABC}}]\Tr[e^{-H_{B}}]}} \\ &\hspace{1.5cm} \equiv \tr_{BC}[\sigma^{BC}E_{A,BC}]\tr_B[\sigma^BE_{A,B}]^{-1}\hspace{0.45cm}\cdot\hspace{0.45cm}\lambda_{ABC}^{-1}.
\end{align}
In the second line we multiplied by $\1=e^{H_A}e^{-H_A}$, which in the third line we can separate and pull into the partial traces, since neither of them trace out the region $A$. 
Thus we may rewrite
\begin{align}
    \|(\tr_{BC}\sigma^{ABC})(\tr_{B}\sigma^{AB})^{-1}-&\1\| \\ \leq &\|(\tr_{BC}\sigma^{ABC})(\tr_{B}\sigma^{AB})^{-1}- (\tr_{BC}[\sigma^{BC}E_{A,BC}])(\tr_B[\sigma^BE_{A,B}])^{-1}\| \\[1mm] &+ \|(\tr_{BC}[\sigma^{BC}E_{A,BC}])(\tr_B[\sigma^BE_{A,B}])^{-1}-\1\| \\[1mm]
    \leq  &\|\tr_{BC}[\sigma^{BC}E_{A,BC}]\|\|(\tr_B[\sigma^BE_{A,B}])^{-1}\||\lambda_{ABC}^{-1}-1| \\[1mm] &+ \|(\tr_B[\sigma^BE_{A,B}])^{-1}\| \|\tr_{BC}[\sigma^{BC}E_{A,BC}]-\tr_B[\sigma^BE_{A,B}]\| \, .
\end{align}
By  \cite[Corollary 4]{art:Bluhm2022exponentialdecayof}, it holds that $\|(\tr_B[\sigma^BE_{A,B}])^{-1}\|\leq C$ and $\|\tr_{BC}[\sigma^{BC}E_{A,BC}]\|\leq C$, for the same constant $C$ depending only on $r,J, \beta$. Furthermore in \cite[Step 2]{art:Bluhm2022exponentialdecayof}, it is proven that $\lambda_{ABC}-1$ decays exponentially, thus by the geometric series so does $\lambda_{ABC}^{-1}-1$, i.e. there exist $K^{\prime\prime},a^{\prime\prime}>0$, depending only on $r,J\beta$ s.t. 
\begin{equation}
  |\lambda_{ABC}^{-1}-1|\leq K^{\prime\prime}e^{-a^{\prime\prime}l}  \, .
\end{equation}
Hence, it remains to bound $\|(\tr_{BC}[\sigma^{BC}E_{A,BC}])-(\tr_B[\sigma^BE_{A,B}])^{-1}\|$ exponentially in $l$. To do this we adopt a similar strategy to \cite[Step 3]{art:Bluhm2022exponentialdecayof}. Let us split $B = B_1B_2$ into to halves, such that $|B_1|=|B_2|=l$, and write
\begin{align}\label{equ:AidIneq1}
    \|\tr_{BC}[\sigma^{BC}E_{A,BC}]-\tr_B[\sigma^BE_{A,B}]\| \leq &\|\tr_{BC}[\sigma^{BC}E_{A,BC}]-\tr_{BC}[\sigma^{BC}E_{A,B_1}]\|   \\[1mm] & + \|\tr_{BC}[\sigma^{BC}E_{A,B_1}]-\tr_B[\sigma^BE_{A,B_1}]\|  \\[1mm]& + \|\tr_B[\sigma^BE_{A,B_1}]- \tr_B[\sigma^BE_{A,B}]\| \, . 
\end{align}
Here we just used the triangle inequality of the operator norm twice.
For the first and third summand in Equation \eqref{equ:AidIneq1} use that the map $Q\mapsto\tr_X[\sigma^XQ]$ is a contraction in $B(\mathcal{H}_\Gamma)\to B(\mathcal{H}_{\Gamma\setminus X})$ by the Russo Dye theorem (see e.g lemma \ref{lemma:ImportantCommutingLemma1} or \cite[Section 3.4]{art:Bluhm2022exponentialdecayof}), and thus 
\begin{equation}
  \|\tr_{BC}[\sigma^{BC}E_{A,BC}]-\tr_{BC}[\sigma^{BC}E_{A,B_1}]\|\leq \|E_{A,BC}-E_{A,B_1}\| \, , 
\end{equation}
where the right-hand side is exponentially decaying in $|B_1|=l$ by \cite[Corollary 3.4 and Remark 3.5]{art:Bluhm2022exponentialdecayof}. The same holds true for 
\begin{equation}
 \|\tr_B[\sigma^BE_{A,B_1}]-\tr_B[\sigma^BE_{A,B}]\|\leq\|E_{A,B_1}-E_{A,B}\|  \, .   \end{equation}
 For the second summand in Equation \eqref{equ:AidIneq1} we use  \cite[Proposition 8.5]{art:Bluhm2022exponentialdecayof}, which gives 
\begin{equation}
  \|\tr_{BC}[\sigma^{BC}E^*_{A,B_1}]-\tr_B[\sigma^BE^*_{A,B_1}]\|\leq \Tilde{K}e^{-\Tilde{a}l}  
\end{equation}
for some $\Tilde{K},\Tilde{a}>0$, depending only on $r,J,\beta$. The same holds true for the adjoints, which is exactly the second summand. In total we have that there exist $K^{\prime\prime\prime},a^{\prime\prime\prime}>0$ depending only on $r$,$J, \beta$ such that 
\begin{align}
    \|\tr_{BC}[\sigma^{BC}E_{A,BC}]-\tr_B[\sigma^BE_{A,B}]\| \leq K^{\prime\prime\prime}e^{-a^{\prime\prime\prime}l}.
\end{align}
Now putting all of the above together we have our desired result.
\begin{align}
    \|(\tr_{BC}[\sigma^{ABC}])(\tr_B[\sigma^{AB}])^{-1}-\1\| &\leq \|\tr_{BC}[\sigma^{BC}E_{A,BC}]\|\|(\tr_B[\sigma^BE_{A,B}])^{-1}\||\lambda_{ABC}^{-1}-1| \\[1mm]&\quad + \|(\tr_B[\sigma^BE_{A,B}])^{-1}\| \|\tr_{BC}[\sigma^{BC}E_{A,BC}]-\tr_B[\sigma^BE_{A,B}]\| \\[1mm] &\leq C^2K^{\prime\prime}e^{-a^{\prime\prime}l}+CK^{\prime\prime\prime}e^{-a^{\prime\prime\prime}l} \, .
\end{align}
\end{proof}

\section{Proof of properties of Schmidt conditional expectations}\label{sec:proof_Schmidt}

\begin{proof}[Proof of Lemma \ref{lem:SchmidAlgebra}]
If $\dist(A_1,A_2)\geq2$, then by definition of the algebras $\mathcal{N}_{A_1}$ and $\mathcal{N}_{A_2}$ have only $\1 = \mathcal{N}_{\emptyset}=\mathcal{N}_{A_1\cap A_2}$ in common.
Since it holds, that $\partial A_1 \cup \partial A_2 = \partial(A_1\cup A_2)$ their union is given by 
\begin{align}
\mathcal{N}_{A_1}\cup\mathcal{N}_{A_2}&=\mathcal{B}(\mathcal{H}_{A_1})\otimes\mathcal{B}(\mathcal{H}_{A_2})\otimes\1_{\mathcal{H}_{(A_1)^c}}\otimes\1_{\mathcal{H}_{(A_2)^c}} \bigotimes_{i\in I_{A_1}}\bigotimes_{j\in J^{(i)}\setminus\{0\}}\mathcal{A}^j_{b_j}\otimes \bigotimes_{i\in I_{A_2}}\bigotimes_{j\in J^{(i)}\setminus\{0\}}\mathcal{A}^j_{b_j} \\ &= \mathcal{B}(\mathcal{H}_{A_1\cup A_2})\otimes\1_{\mathcal{H}_{(A_1\cup A_2)^c}}\otimes\bigotimes_{i\in I_{A_1\cup A_2}}\bigotimes_{j\in J^{(i)}\setminus\{0\}}\mathcal{A}^j_{b_j} \\ &= \mathcal{N}_{A_1\cup A_2}
\end{align}
For the other case we can WLOG assume $A_1\subset A_2$, hence $A_1\cup A_2=A_2$ and $A_1\cap A_2=A_1$. Then clearly $\mathcal{N}_{A_1}\cup \mathcal{N}_{A_2}=\mathcal{N}_{A_2}=\mathcal{N}_{A_1\cup A_2}$. Similarly $\mathcal{N}_{A_1}\cap\mathcal{N}_{A_2}=\mathcal{N}_{A_1}=\mathcal{N}_{A_1\cap A_2}$.
\end{proof}

\section{Proof of Case 2 in \Cref{lemma:OmegacMLSI}}\label{sec:proof_2ndMartinelli}

\begin{proof}[Proof of Case 2 in \Cref{lemma:OmegacMLSI}]
Let $N\in\mathbb{N}\setminus\{1\}$ and $\frac{1}{2}<\epsilon<1, s.t. \frac{1}{2}+\frac{1}{2N}\geq \epsilon$.
We enumerate a maximal set of partitions of $B_{x_j,L}$ into $\{C_j^{\tilde{l}},D_j^{\tilde{l},\frac{L}{N}}\}$, s.t.  $\height(C_j^{\tilde{l}})=\tilde{l},\height(D_j^{\tilde{l},\frac{L}{N}})=L+\frac{L}{N}-\tilde{l}\leq \epsilon L$ and s.t. different partitions have disjoint overlaps, i.e. $\left(C_j^{\tilde{l}_1}\cap D_j^{\tilde{l}_1,\frac{L}{N}}\right)\cap \left(C_j^{\tilde{l}_2}\cap D_j^{\tilde{l}_2,\frac{L}{N}}\right) = \emptyset$, whenever $\tilde{l}_1\neq \tilde{l}_2$.
There exist $\frac{L}{\frac{L}{N}}=N$ of these partitions, since their overlap is of 'width' $\frac{L}{N}$. Call these partitions $\{C_i,D_i\}^{N}_{i=1}$. Now we average over all the approximate tensorization results of these partitions to get 
\begin{align}
    D(\omega\|E^S_{B_{x_j,L}*}(\omega))&\leq \frac{1}{N}\sum_{i=1}^{N} \frac{1}{1-2\eta_{C_i,D_i}(\frac{L}{N})}\left[D(\omega\|E^S_{C_i*}(\omega))+D(\omega\|E^S_{D_i*}(\omega)) \right] \\
    &\leq \frac{1}{1-2\eta(\frac{L}{N})}\frac{1}{N}\sum_{i=1}^{N} C(\height(C_i))D_{C_i}(\omega)+C(\height(D_i))D_{D_i}(\omega) \\
    &\leq C(\epsilon L)\frac{1}{1-2\eta(\frac{L}{N})}\frac{1}{N}\Bigg(\sum_{i=1}^{N}\left(2D_{C_i\cap D_i}(\omega)+D_{C_i\setminus D_i \cup D_i\setminus C_i}(\omega) \right. \\ & \quad +\sum_{\underset{R_k\cap (C_i\setminus D_i)\neq \emptyset}{R_k\cap (C_i\cap D_i)\neq \emptyset}}D(\omega\|E^S_{R_k*}(\omega))+\sum_{\underset{R_k\cap (D_i\setminus C_i)\neq \emptyset}{R_k\cap (C_i\cap D_i)\neq \emptyset}}D(\omega\|E^S_{R_k*}(\omega)) \Bigg) \\
    &\leq C(\epsilon L)\frac{1}{1-2\eta(\frac{L}{N})}\frac{1}{N}\left(2D_{\bigcup_{i=1}^{N}(C_i\cap D_i)}(\omega)+\sum_{i=1}^{N}D_{(C_i\setminus D_i)\partial\cup (D_i\setminus C_i)\partial}(\omega)\right) \\
    &\leq C(\epsilon L)\frac{1}{1-2\eta(\frac{L}{N})}\frac{1}{N}\left(2+N\right)D_{B_{x_j,L}}(\omega),
\end{align} 
where in the second line we used the definition of $C(\height(\cdot))$ and
 in the third line we used that $\height(C_i),\height(D_i)\leq C(\epsilon L)$ since $\height(C_i)=\tilde{l},\height(D_i)=L+\frac{L}{N}-\tilde{l}\leq \epsilon L$.
 
Hence it follows that $C(L)\leq C(\epsilon L)\frac{1}{1-2\eta(\frac{L}{N})}(1+\frac{2}{N})=:C(\epsilon L)\tilde{f}(L)(1+\frac{2}{N})$.
Repeating this $M=\mathcal{O}(\log L)$ times s.t. $\epsilon^M L=l_0=:L_{\text{min},1}$ gives
\begin{align}
    C(L)&\leq C(l_0)(1+\frac{2}{N})^M\prod_{k=1}^M\tilde{f}(\epsilon^kL)\leq C(l_0)(1+\frac{2}{N})^M\prod_{k=0}^\infty \tilde{f}(l_0\epsilon^{-k})\\ &= \mathcal{O}\left(\left(1+\frac{2}{N}\right)^{\mathcal{O}(\log L)}\right)=\mathcal{O}(L),
\end{align}
where the infinite product converges since $\eta\left(\frac{L}{N}\right)$ 
is exponentially decaying in $L$ under the condition on the decay length $\xi,\xi^\prime$. Finally, notice that for exponential graphs (i.e. trees) we have that $|B_{x_j,L}|=b^L$ and hence $\mathcal{O}(L)=\mathcal{O}(\log|B_{x_j,L}|)$. Again, note that it is independent of the $x_j$ which we fixed, so the result follows.
The optimal $N$ for which this works is $N=2$, hence it suffices if we have decay when $\diam(C\cup D)=\diam B_{x_j,L}=L=2l=2\dist(C\setminus D,D\setminus C)$, which is the one we have in \underline{Case 2}.
\end{proof}

\end{document}